\documentclass[prd,tightenlines,nofootinbib,superscriptaddress]{revtex4}

\usepackage{pre_all}
\usepackage{pre_CS}
\usepackage{bm}
\allowdisplaybreaks

\hypersetup{
pdfstartview = {FitH},
}
\hypersetup{
	colorlinks=true,         
	linkcolor=brown,          
	citecolor=red,        
	urlcolor=blue            
}

\begin{document}

\title{Melonic Radiative Correction in Four-Dimensional Spinfoam Model with Cosmological Constant}

\author{{\bf Muxin Han}}\email{hanm@fau.edu}
\affiliation{Department of Physics, Florida Atlantic University, 777 Glades Road, Boca Raton, FL 33431, USA}
\affiliation{Institut f\"ur Quantengravitation, Universit\"at Erlangen-N\"urnberg, Staudtstr. 7/B2, 91058 Erlangen, Germany}

\author{{\bf Qiaoyin Pan}}\email{qpan@fau.edu}
\affiliation{Department of Physics, Florida Atlantic University, 777 Glades Road, Boca Raton, FL 33431, USA}

\date{\today}

\begin{abstract}
Infrared divergence is a common feature of spinfoam models with a vanishing cosmological constant but is expected to disappear in presence of a non-vanishing cosmological constant. In this paper, we investigate the spinfoam amplitude with cosmological constant \cite{Han:2021tzw} on the melon graph, which is known as the melonic radiative correction. The amplitude closely relates to the state-integral model of complex Chern-Simons theory. We prove that the melonic radiative correction is finite in presence of a non-vanishing cosmological constant, in contrast to the infrared divergence of spinfoam models with a vanishing cosmological constant. In addition, we also analyze the scaling behavior of the radiative correction in the limit of small cosmological constant.

\end{abstract}

\maketitle

\tableofcontents

\section{Introduction}

Spinfoam quantum gravity \cite{Rovelli:2014ssa,Perez:2012wv} provides a covariant formulation to Loop Quantum Gravity (LQG) and can be viewed as a discrete path integral of quantum gravity. A spinfoam model is based on a cellular decomposition, conventionally chosen to be a triangulation, of the spacetime manifold. By virtue of LQG, the geometrical areas in 3+1 dimensional (or 4D) spinfoam models have discrete spectra $\fa=\gamma\ell^2_{\p}\sqrt{j(j+1)}$ where $\gamma$ is called the Babero-Immirzi parameter, $\ell_\p=\sqrt{8\pi G\hbar/c^3}$ is the Plank length and $j$ is an $\SU(2)$ irreducible representation label. This setting directly leads to the consequence that spinfoam models are free of ultraviolet divergences. However, infrared divergences are still present in spinfoam models with a vanishing cosmological constant $\Lambda$. Such divergences are called the radiative corrections or self-energies of the spinfoam models. Understanding these divergences is essential for studying the renormalization of the theory, which should lead us from the quantum spacetime dynamics at the microscopic scale to physical predictions at large scale. 

In 2+1 dimensions (or 3D), the divergence of the spinfoam model with $\Lambda=0$, called the Ponzano-Regge model \cite{Ponzano:1968se}, is related to the diffeomorphism symmetry \cite{Freidel:2002dw} and implicit sum over orientations of the spacetime manifold \cite{Christodoulou:2012af}. 
This divergence is regularized in the Turaev-Viro model \cite{Turaev:1992hq}, which is a deformed version of the Ponzano-Regge model corresponding to $\Lambda>0$. 
However, divergence in 4D spinfoams with $\Lambda=0$ remains an open question. Group field theory (GFT) (see \eg \cite{Oriti:2006se,Krajewski:2011zzu}) suggests that spinfoam amplitude corresponding to the melonic spinfoam graph (see fig.\ref{fig:melon_graph}) contributes the most divergent part for the radiative correction (at least for simple enough spinfoam graphs) \cite{Perini:2008pd,Krajewski:2010yq}. The melon graph is the first-order correction of a spinfoam amplitude for a spinfoam edge, or a spinfoam propagator in the GFT language. 
Radiative correction corresponding to the melon graph has been studied for the Engle-Pereira-Rovelli-Livine-Freidel-Krasnov (EPRL-FK) model \cite{Engle:2007wy,Freidel:2007py}, which is one of the most studied 4D spinfoam models with $\Lambda=0$. A recent study based on numerical method \cite{Frisoni:2021uwx} reveals that the EPRL-FK spinfoam amplitude of a melon graph scales as $|\cA_{\rm melon}|\sim|\Lambda|^{-1}$ at small $\Lambda$ provided a standard choice of the face amplitude. It is consistent with earlier results of its lower bound $|\cA_{\rm melon}|\sim\ln\lb|\Lambda|^{-1}\rb$ \cite{Riello:2013bzw} and upper bound $|\cA_{\rm melon}|\sim|\Lambda|^{-9}$ \cite{Dona:2018pxq}. 

Another way to target the radiative correction is to study the spinfoam model with a non-vanishing $\Lambda$ and consider its amplitude at small $|\Lambda|$ limit. Inspired by the Turaev-Viro model, it has been conjectured that a 4D spinfoam model with $\Lambda\neq 0$ should be free of divergence by construction. A natural way to manifest finiteness is to consider the quantum group deformation of the Lorentz group in the spinfoam models \cite{Noui:2002ag,Han:2010pz,Fairbairn:2010cp} as they provide a cutoff in representation by definition. On the other hand, a valid spinfoam model is supposed to reproduce discrete gravity, \ie Regge calculus \cite{Regge:1961px,Hartle:1981cf,Friedberg:1984ma,Regge:2000wu}, at its semi-classical limit. This has been realized in 3D spinfoam models \cite{Ponzano:1968se,Turaev:1992hq} and 4D spinfoam models with $\Lambda=0$ \cite{Barrett:1997gw,Barrett:1998gs,Bianchi:2008ae,Conrady:2008mk}. Then a 4D spinfoam model with $\Lambda\neq0$ is legitimately expected to bring out, at the semi-classical regime, the Regge calculus for 4-simplex with constant curvature. Due to the formulation complexity, however, the semi-classical approximation for the quantum group deformation of 4D spinfoam models is difficult to examine. 

Recently, a 4D spinfoam model with $\Lambda\neq0$ \cite{Han:2021tzw} was proposed and shown to be featured with both finiteness and the expected semi-classical approximation. This spinfoam model is defined by the $\SL(2,\bC)$ Chern-Simons partition function on the boundary of a 4D manifold coupled with Chern-Simons coherent states. It is capable of describing 4D quantum gravity with $\Lambda$ in either positive or negative sign, which is not fixed {\it in priori} but emerges from the equations of motion semi-classically and the boundary states. This spinfoam model is a modified version of that introduced in \cite{Haggard:2014xoa}, wherein the role of coherent states are played by the projective $\SL(2,\bC)$ spin network states and the spinfoam amplitude expression therein is only formal hence finiteness is doubtful. Therefore, it is promising to study in more detail on this new spinfoam model. 

In this work, we analytically study the radiative correction corresponding to the melon graph of the spinfoam model introduced in \cite{Han:2021tzw} at the $\Lambda\rightarrow0$ approximation. 
In the line of analysis, we improve the spinfoam model by proposing a concrete face amplitude given by a function of the spin associated to the face, consistent with the face amplitude in EPRL-FK model at $\Lambda\rightarrow0$ limit. As an important result, we prove that in the presence of nonzero $\Lambda$, the spinfoam amplitude on the melon graph is finite. Moreover, the convergence of the amplitude is even stronger than the general discussion in \cite{Han:2021tzw}: We show that the finiteness still holds after removing an exponentially decaying factor inserted in the edge amplitude there. This result is in contrast to the divergent melonic radiative correction in the spinfoam models with vanishing $\Lambda$. This finiteness is one of the inviting features of the spinfoam model with cosmological constant.

\begin{figure}[t]
\centering
\includegraphics[width=0.7\textwidth]{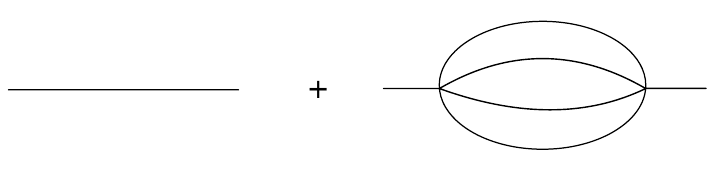}
\caption{The melonic spinfoam amplitude as the correction to the spinfoam edge amplitude.} 
\label{fig:melon_graph}
\end{figure}

We also discuss the scaling behavior of the melonic amplitude $\cA_{\rm melon}$ suppresses as $\Lambda\to 0$. in the small $\Lambda$ limit. The scaling behavior can be analyzed by applying the stationary phase approximation to the amplitude. The amplitude in the small $\Lambda$ regime is dominant by the contributions from the critical points. The scaling behavior is obtained by a power-counting argument. We find that the scaling behavior of the melonic amplitude has the lower bound as $|\cA_{\rm melon}|\sim 1/|\Lambda|^{15+6\mu}$ where $\mu$ is an undetermined power in the face amplitude.

This paper is organized as follows. In Section \ref{sec:review}, we give a rather self-consistent review of the spinfoam model with $\Lambda$ focusing on the vertex amplitude. We modify the way to impose the second-class simplicity constraints compared to the original work which we believe can simplify the construction. In Section \ref{sec:melon_graph}, we consider the full spinfoam amplitude for the melon graph in a similar way as for the vertex amplitude. That is to first consider the Chern-Simons partition function on the boundary then impose the simplicity constraints through Chern-Simons coherent states. 
Semi-classical approximation of the full melonic spinfoam amplitude is analyzed separately in Section \ref{sec:radiative_correction} and Section \ref{sec:coherent_integral} according to different scaling behaviours of different parts of the amplitude. The result of the melonic radiative correction is then drawn. 
Finally, we give a geometrical interpretation of critical points in Section \ref{app:geometry} and 
We conclude in Section \ref{sec:conclusion}.

\section{4D spinfoam amplitude with $\Lambda\neq 0$ from boundary Chern-Simons theory}
\label{sec:review}

In this section, we review the spinfoam model introduced in \cite{Han:2021tzw} which corresponds to four-dimensional quantum gravity with a non-vanishing cosmological constant $\Lambda$.

\subsection{From 4D gravity to Chern-Simons path integral}

The construction of the spinfoam amplitude is motivative by the formal path integral formalism of 4D gravity with $\Lambda\neq0$. 
We start from the Plebanski action \cite{Plebanski:1977zz}, which is a first-order formulation of 4D gravity expressed as a constrained $\SL(2,\bC)$ BF theory, adding a cosmological term. Consider an $\sl(2,\bC)$ two-form $B$ and an $\sl(2,\bC)$ connection $\cA$ which is a one-form on a 4-manifold $\cM_4$. The topological BF action, denoted as $S_{\LBF}$, is
\be
S_{\LBF}[B,\cA]=\int_{\cM_4} \tr \left[\lb \star +\f{1}{\gamma}\rb B\w \lb \cF (\cA)+ \f{|\Lambda|}{6} B\rb\right]\,,
\label{eq:LBF_action}
\ee
where $\cF(\cA)$ is the curvature 2-form of $\cA$, $\star$ is the Hodge star operation satisfying $\star^2=-1$ in Lorentzian signature and $\gamma$ is the Barbero-Immirzi parameter which takes a real value. The trace is taken in the $\sl(2,\bC)$ Lie algebra and it evaluates as $\tr[X\w Y]=X^{IJ}Y_{IJ}$ \footnote{
The form of the action \eqref{eq:LBF_action} relies on the self-dual and anti-self-dual decomposition of a complexified $\sl(2,\bC)$ element which gives two commuting copies of complexified $\su(2)$ elements, \ie $(\sl(2,\bC))_\bC=\su(2)^+_\bC\oplus\su(2)^-_\bC$. See \cite{Haggard:2014xoa} for a detailed derivation, wherein the global sign of the action is taken differently.
}. $S_{\LBF}$ depends on the absolute value of the cosmological constant $|\Lambda|$.

By imposing the {\it simplicity constraint}, which relates $B$ to the cotetrad one-form $e$ by
\be
B\cong \sgn(\Lambda) e\w e\,,
\ee
one recovers the first-order action of general relativity with a cosmological constant $\Lambda$, written in terms of the cotetrad $e$ and the connection $\cA$
\be
S_{\GR}[e,\cA]= \int_{\cM_4}\tr\left[\lb \star +\f{1}{\gamma}\rb (e\w e)\w \lb \cF(\cA)+\f{\Lambda}{6}(e\w e)\rb \right]\,.
\ee
The equations of motion of \eqref{eq:LBF_action} from varying the $B$ field leads to a linear relation between the $\cF$ field and the $B$ field, which transfers to the equation between the curvature and the cotetrad after imposing the simplicity constraints. 
\be
\f{\partial S_{\LBF}}{\partial B^{IJ}}=0
\quad \Longrightarrow\quad
\cF=\f{|\Lambda|}{3}B
\quad \xrightarrow{B\cong \sgn(\Lambda) e\w e}\quad
\cF\cong \f{\Lambda}{3}e\w e\,.
\label{eq:simplicity_constraint}
\ee
The right-most equation above is the simplicity constraint that we will implement to the theory. 

The path integral of the action \eqref{eq:LBF_action} contains a Gaussian integral for the $B$ field, performing which constrains $\cF=\f{|\Lambda|}{3}B$ and leads to two (conjugation related) second Chern-forms when separating the $\sl(2,\bC)$-valued curvature $\cF$ into its self-dual part $F$ and anti-self-dual part $\Fb$. By manipulating path integral \footnote{The first equality of \eqref{eq:path_integra_LBF} is formal. Indeed, the integration of $B$-field gives a divergent factor that corresponds to vacuum contributions in the field theory language, which would be cancelled out when one computes correlation functions. Same as the result of \eqref{eq:CS_actions} when one integrates out the $A$ and $\Ab$ fields in the bulk of $\cM_4$. },
\be\begin{split}
\cZ=\int\rd \cA \rd B \, e^{\f{i}{\ell_{\p}^2}S_{\LBF}} 
&=\int\rd \cA\, \exp\lb\f{3i}{2\ell_{\p}^2|\Lambda|} \int_{\cM_4}\tr\left[\lb\star+\f{1}{\gamma}\rb\cF(\cA)\w\cF(\cA)\right]\rb\\
&=\int\rd A\rd\Ab \, \exp\lb-\f{3}{2\ell_\p^2 |\Lambda|}\int_{\cM_4}\lb 1-\f{i}{\gamma}\rb\tr[F(A)\w F(A)] - \lb1+\f{i}{\gamma}\rb\tr[\Fb(\Ab)\w\Fb(\Ab)]\rb\,,
\end{split}
\label{eq:path_integra_LBF}
\ee
where $A$ and $\Ab$ are the self-dual and anti-self-dual parts of $\cA$ respectively and $\ell_\p$ is the Planck length. (Throughout this paper, we take the convention that the gravitational constant $G=1$ and that the speed of light $c=1$.)
As the exponent is a topological term, \eqref{eq:path_integra_LBF} becomes a path integral of $\SL(2,\bC)$ Chern-Simons action with complex level on the boundary $\partial \cM_4$. When $\mathcal{M}_4$ is topologically trivial, 
\be
\cZ=e^{-iS_{\CS}[A,\Ab]}\,,\quad S_{\CS}[A,\Ab]=
\f{t}{8\pi}\int_{\partial\cM_4} \tr\left[A\w\rd A+ \f23 A\w A\w A\right]+ \f{\tb}{8\pi}\int_{\partial\cM_4} \tr\left[\Ab\w\rd \Ab+ \f23 \Ab\w \Ab\w \Ab\right] 
\label{eq:CS_actions}
\ee
where the level $t$ and its complex conjugate $\tb$ can be separated into real and imaginary parts as
\be
t=k+is\,,\quad
\tb=k-is\,,\quad
\text{where}\,\,
k=\f{12\pi}{\ell_\p^2\gamma|\Lambda|}\in\Z_+\,,\quad
s=\gamma k\in\R_+\,.
\label{eq:def_t_k}
\ee
Therefore, the quantization of gravity on a 4-manifold $\cM_4$ with a cosmological constant $\Lambda$ now relates to quantization of the $\SL(2,\bC)$ Chern-Simons theory with complex coupling constant on the 3D boundary $\partial\cM_4$ of the manifold: 
\be
S_{\CS}[A,\Ab]=\f{t}{8\pi}\int_{\partial\cM_4} \tr\left[A\w\rd A+ \f23 A\w A\w A\right]+ \f{\tb}{8\pi}\int_{\partial\cM_4} \tr\left[\Ab\w\rd \Ab+ \f23 \Ab\w \Ab\w \Ab\right]\,.
\label{eq:CS_action}
\ee
The connection $A$ (as well as $\Ab$) is now restricted to the 3-boundary $\partial\cM_4$, where the simplicity constraints will be imposed. 

When constructing the spinfoam amplitude, we consider $\mathcal{M}_4$ to be a 4-simplex and quantize the Chern-Simons theory {\it canonically} on the boundary, followed by suitably imposing the (quantized) simplicity constraint. 
The result of the construction is the spinfoam vertex amplitude $\mathcal{A}_v$. Due to the fact that the simplicity constraint requires non-trivial magnetic flux by \eqref{eq:simplicity_constraint}, certain defect has to be introduced to the Chern-Simons theory (otherwise the Chern-Simons theory would imply $\mathcal{F}=0$ by the equation of motion). Some details about the quantization of the Chern-Simons theory with defect and the construction of $\mathcal{A}_v$ are reviewed in the following.

\subsection{Chern-Simons partition function on the triangulated 3-manifold}

Consider a 4-simplex which is topologically equivalent to a 4-ball $\cB_4$ whose boundary is a 3-sphere $S^3$. The triangulation ${\bf T}_3$ of $S^3$ contains 5 tetrahedra sharing 10 triangles. The dual graph, equivalently, contains 5 nodes connected by 10 links and is denoted as $\Gamma_5$ (See fig.\ref{fig:Gamma5}) \footnote{Throughout this paper, unless specification, we use the terminology that a 0-simplex and a 1-simplex in the triangulation of a manifold are denoted as a vertex and an edge respectively while a 0-complex and a 1-complex in the dual graph of the triangulation are denoted as a node and a link respectively. Note that the dual graph is different from the spinfoam graph, \eg the melon graph (see fig.\ref{fig:melon_graph}), where we denote the 0-, 1- and 2-complexes as spinfoam vertices, spinfoam edges and spinfoam faces. In the context with no ambiguity, we denote them simply as vertices, edges and faces for conciseness.}. 
Upon triangulation, the simplicity constraints take the form of smeared 2-forms hence it is natural to impose them on the triangles of ${\bf T}_3$. In the dual picture, the violation of flatness occurs {\it only} on the links of $\Gamma_5$. This means one can first study the quantum Chern-Simons theory on the graph complement $M_3:=S^3\backslash \Gamma_5$ which is the complement of an open tubular neighbourhood of $\Gamma_5$ in $S^3$ and then impose the simplicity constraints on the boundary $\partial M_3$ as certain boundary conditions.
\begin{figure}[h!]
\centering
\includegraphics[width=0.25\textwidth]{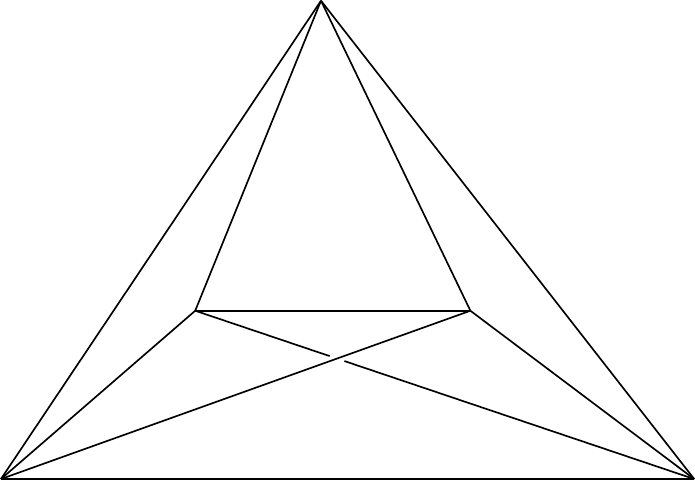}
\caption{The $\Gamma_5$ graph as the dual graph of the triangulation ${\bf T}_3$ of $S^3$.}
\label{fig:Gamma5}
\end{figure}
In this subsection, we review the main ingredients to perform the former step. Ref.\cite{Han:2021tzw} applied the method developed in a series of works \cite{Gaiotto:2009hg,Dimofte:2011gm,Dimofte:2011ju,Dimofte:2013lba,Dimofte:2014zga,andersen2014complex} to construct the Chern-Simons partition function $\cZ_{M_3}$ in terms of finite sums and finite-dimensional absolutely-convergent state-integral. 
Under this construction, $\cZ_{M_3}$ carries a complex gauge group $\SL(2,\bC)$ 
 and describes the quantization of the moduli space $\cM_{\Flat}(M_3,\SL(2,\bC))$ of flat $\SL(2,\bC)$ connection on $M_3$. 

The quantization of complex Chern-Simons theory uses the {\it ideal triangulation} of the graph-complement 3-manifold, say $\Gamma$-complement of $\cM_3$ denoted as $\cM_3\backslash\Gamma$. 
The building blocks of the ideal triangulation are the {\it ideal tetrahedra} $\triangle$'s, which are tetrahedra with vertices truncated into triangles as shown in fig.\ref{fig:ideal_tetra} \footnotemark{}. 
The original boundaries of an $\triangle$ before truncation are called the {\it geodesic boundaries} of $\triangle$ and the truncated vertices are called the {\it cusp boundaries} (or {\it disc cusp}) of $\triangle$. 
The boundaries of $\cM_3\backslash\Gamma$ can also separated into two types:
\begin{itemize}
    \item geodesic boundaries -- boundaries created by removing open balls around vertices of $\Gamma$, which are holed spheres, and
    \item cusp boundaries or {\it annulus cusp} -- boundaries created by removing the tubular neighbourhood of edges of $\Gamma$, which are annuli.
\end{itemize}
An ideal triangulation decomposes $\cM_3\backslash\Gamma$ into a set of ideal tetrahedra such that the geodesic boundaries are triangulated by the geodesic boundaries of $\triangle$'s while the annulus cusps are triangulated by the disc cusps of $\triangle$'s. An example of the ideal triangulation of a four-valent-node-complement of $S^3$ is illustrated in fig.\ref{fig:v_to_t}. It is part of the ideal triangulation of $M_3$.

\begin{figure}[h!]
\centering
\begin{minipage}{0.45\textwidth}
\centering
\includegraphics[width=0.5\textwidth]{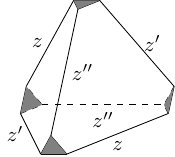}
\subcaption{}
\label{fig:ideal_tetra}
\end{minipage}\quad
\begin{minipage}{0.45\textwidth}
\includegraphics[width=0.7\textwidth]{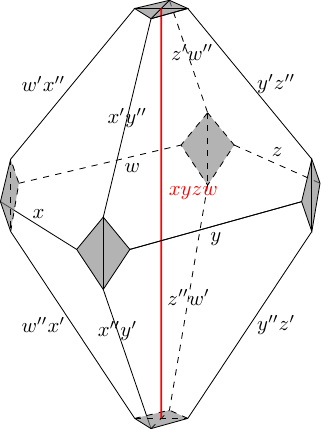}
\subcaption{}
\label{fig:ideal_octa}
\end{minipage}
\caption{{\it (a)} An ideal tetrahedron whose edges are dressed with edge coordinates $(z,z',z'')$. Each pair of opposite edges are dressed with the same coordinate. The disc cusps are filled {\it in gray}. {\it (b)} An ideal octahedron. Choose the equator to be edges dressed with $x,y,z,w$. Adding an internal edge ({\it in red}) orthogonal to the equator separates the ideal octahedron into four ideal tetrahedra, each of which is dressed with different copies of coordinates $(x,x',x'')\,,(y,y',y'')\,,(z,z',z'')\,,(w,w',w'')$. For edges shared by different ideal tetrahedra, coordinates are multiplied together.}
\end{figure}

The triangulation of $M_3$ can be decomposed into 5 {\it ideal octahedra} (see fig.\ref{fig:triangulation_All}), then each ideal octahedron can be further decomposed into 4 ideal tetrahedra by adding an internal edge (see fig.\ref{fig:ideal_octa}). 
As a result, the triangulation contains 20 ideal tetrahedra in total. (One should not confuse the ideal tetrahedra from triangulating $M_3$ with the tetrahedra from triangulating $S^3$ as the boundary of the 4-simplex.) The boundary $\partial M_3$ of $M_3$ is made of five 4-holed spheres $\{\cS_a\}_{a=1}^{5}$ and 10 annuli $\{(ab)|a < b, \,a,b=1,\cdots,5\}$ connecting these holes. 
The triangulation of $M_3$ induces the ideal triangulation on $\partial M_3$. 
The ideal triangulation of a 4-holed sphere $\cS_a$ contains four triangles located at the holes and four hexagons as illustrated in fig.\ref{fig:v_to_t_b}. 
\begin{figure}[h!]
\centering
\begin{minipage}{0.3\textwidth}
\centering
\includegraphics[width=0.9\textwidth]{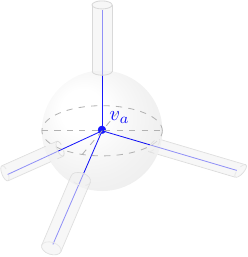}
\subcaption{}
\label{fig:v_to_t_a}
\end{minipage}\quad
\begin{minipage}{0.3\textwidth}
\centering
\includegraphics[width=0.9\textwidth]{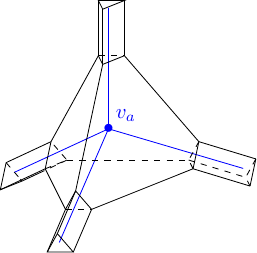}
\subcaption{}
\label{fig:v_to_t_b}
\end{minipage}
\caption{{\it (a)} Illustration of part of the $S^3\backslash\Gamma_5$. A four-valent node $v_a\in\Gamma_5$ and its neighbourhood is removed from $S^3$ and generates a part of the boundary as a 4-holed sphere $\cS_a$ whose holes are connected to annuli. 
{\it (b)} The ideal triangulation of $(a)$. Vertices created by edges of the graph piercing through the sphere are truncated into triangles. Each such triangle is connected to the boundary of a triangular prism which is the ideal triangulation of an annulus in $(a)$. (The triangulation of the parallelograms in triangular prisms is not shown for a clear visual effect.) 
In the full triangulation of $S^3\backslash\Gamma_5$, each triangular prism is connected to a pair of truncated vertices from two different triangulated 4-holed spheres. }
\label{fig:v_to_t}
\end{figure}
\footnotetext{An ideal tetrahedron can be lifted to the hyperbolic 3-plane $\bH^3$ with all the vertices located at infinity and all faces along geodesic surfaces of $\bH^3$. See \eg \cite{Dimofte:2010wxa}. }
On the other hand, an annulus is triangulated into the boundary of a triangular prism whose two triangles are identified with the cusp discs the annulus connects and the four parallelograms are split into four triangles.
Combinatorially, $\partial M_3$ is triangulated into 20 hexagonal geodesic boundaries and 30 quadrangular cusp boundaries. 

\begin{figure}[h!]
\centering
\includegraphics[width=0.9\textwidth]{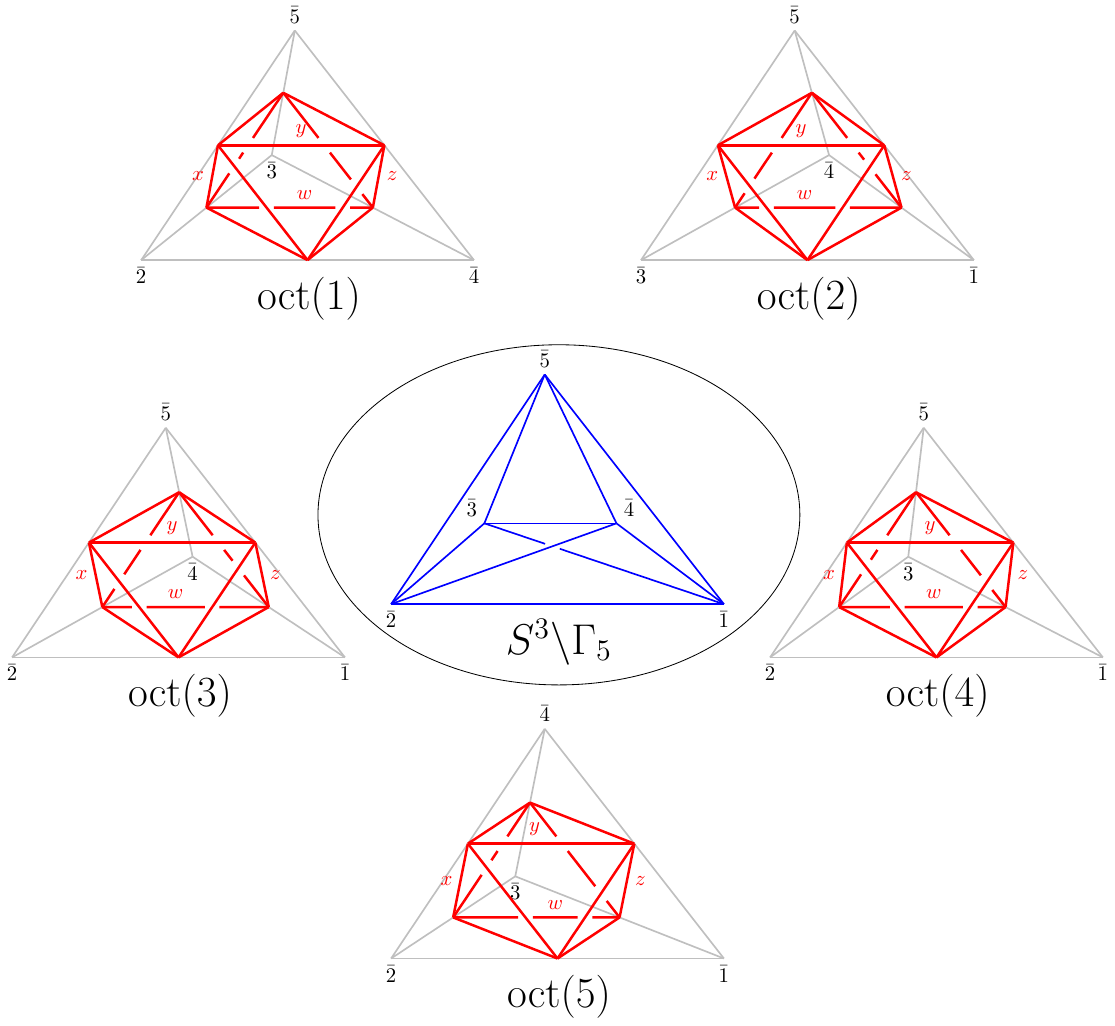}
\caption{The decomposition of the ideal triangulation of $M_3\equiv S^3\backslash\Gamma_5$ into 5 ideal octahedra ({\it in red}), each of which can be decomposed into 4 ideal tetrahedra. The cusp boundaries of the ideal octahedra are shrunk to vertices in this figure. (See fig.\ref{fig:ideal_octa}) for the ideal octahedron with un-shrunk cusp boundaries.)  
Numbers $\bar{1},\bar{2},\bar{3},\bar{4},\bar{5}$ with bars denote the 4-holed spheres on $\partial M_3$. 
The faces ${a}, {b}, {c}, {d}, {e}, {f}, {g}, {h}, {i}, {j}$ ({\it labelled in green and each is on a boundary triangle of the tetrahedron in gray}) are the faces where a pair of octahedra are glued. Two ideal octahedra are glued through pairs of faces having the same label (with different subscripts). In each ideal octahedron, $x, y, z, w$ ({\it labelled in red}) are chosen to form the equator of the octahedron. The same figure appears in \cite{Han:2015gma,Han:2021tzw}.}
\label{fig:triangulation_All}
\end{figure}
The building block to construct the partition function $\cZ_{M_3}$ is therefore provided by the $\SL(2,\bC)$ Chern-Simons partition function for an ideal tetrahedron $\triangle$, which is identical to the Chern-Simons wave function on $\triangle$ given boundary condition. The Chern-Simons wave function can be defined following the canonical quantization of the moduli space of framed flat connections on $\triangle$, which has been well studied in the literature (see \eg \cite{Dimofte:2011gm,Dimofte:2013lba,Dimofte:2010wxa}). Here, we use the result of \cite{Han:2021tzw} to write the partition function for $M_3$ and sketch the necessary steps in constructing this partition function in Appendix \ref{app:partition}.
 
As shown in fig.\ref{fig:triangulation_All}, the triangulation of $S^3\backslash\Gamma_5$ contains 5 ideal octahedra (see fig.\ref{fig:triangulation_All}) with all edges on the boundary $\partial(S^3\backslash\Gamma_5)$. The phase space $\cP_{\partial(S^3\backslash\Gamma_5)}$, which is the moduli space $\cM_{\Flat}(\partial(S^3\backslash\Gamma_5),\SL(2,\bC))$ of $\SL(2,\bC)$ flat connection on $\partial(S^3\backslash\Gamma_5)$, has 15 holomorphic position coordinates, which we group into a vector $\vec{\cQ}$, and 15 holomorphic momentum coordinates, which we group into a vector $\vec{\cP}$, as well as their anti-holomorphic counterparts $\vec{\widetilde{\cQ}}$ and $\vec{\widetilde{\cP}}$. Their elements are denoted as
\be
\vec{\cQ}=(\{2L_{ab}\}_{(ab)},\{\cX_a\}_{a=1}^{5})\,,\quad
\vec{\cP}=(\{\cT_{ab}\}_{(ab)},\{\cY_a\}_{a=1}^{5})\,,\quad
\vec{\widetilde{\cQ}}=(\{2\widetilde{L}_{ab}\}_{(ab)},\{\widetilde{\cX}_a\}_{a=1}^{5})\,,\quad
\vec{\widetilde{\cP}}=(\{\widetilde{\cT}_{ab}\}_{(ab)},\{\widetilde{\cY}_a\}_{a=1}^{5})\,.
\ee
They can be parametrized in terms of two complex vectors $\vec{\mu},\vec{\nu}\in\bC^{15}$ and a vector with discrete-valued entries $\vec{m},\vec{n}\in(\Z/k\Z)^{15}$ where $k=\f{12\pi}{\ell_\p^2\gamma|\Lambda|}\in\Z_+$. Precisely,
\be
\vec{\cQ}=\f{2\pi i}{k}\lb -ib\vec{\mu}-\vec{m} \rb\,,\quad
\vec{\cP}=\f{2\pi i}{k}\lb -ib\vec{\nu}-\vec{n} \rb\,,\quad
\vec{\widetilde{\cQ}}=\f{2\pi i}{k}\lb -ib^{-1}\vec{\mu}+\vec{m} \rb\,,\quad
\vec{\widetilde{\cP}}=\f{2\pi i}{k}\lb -ib^{-1}\vec{\nu}+\vec{n} \rb\,,\quad
\label{eq:Q_P_param}
\ee
where $b$ is a phase parameter related to the Barbero-Immirzi parameter satisfying
\be
b^2=\f{1-i\gamma}{1+i\gamma}\,,\quad \re(b)>0\,,\quad \im(b)\neq0\,,\quad |b|=1\,.
\ee
Conversely,
\be
\vec{\mu}=\f{kb}{2\pi(b^2+1)}\lb\vec{\cQ}+\vec{\widetilde{\cQ}}\rb\,,\,\,
\vec{m}=\f{ik}{2\pi(b^2+1)}\lb\vec{\cQ}-b^2\vec{\widetilde{\cQ}}\rb\,,\,\,
\vec{\nu}=\f{kb}{2\pi(b^2+1)}\lb\vec{\cP}+\vec{\widetilde{\cP}}\rb\,,\,\, 
\vec{n}=\f{ik}{2\pi(b^2+1)}\lb\vec{\cP}-b^2\vec{\widetilde{\cP}}\rb\,.
\ee
$(2L_{ab},\cT_{ab})$ are associated to the annulus $(ab)$ and $(\cX_a,\cY_a)$ are associated to two edges connected to a common hole of the ideal triangulation of the 4-holed sphere $\cS_a$ (see fig.\ref{fig:v_to_t}).  
 $2L_{ab}$ and $\cT_{ab}$ are called the complex (logarithmic) Fenchel-Nielsen (FN) length and FN twist respectively \footnote{$e^{\cT_{ab}}$'s are the coordinates of $\cM_\Flat(\partial(S^3\backslash\Gamma_5),\SL(2,\bC))$ because they involve square roots of FG coordinates due to the $1/2$ entries in $(\bB^\top)^{-1}$.}. The FN length $2L_{ab}$ is related to the squared eigenvalue of the meridian holonomy for the annulus $\lambda_{ab}^2=e^{2L_{ab}}$. 
$\cX_a$ and $\cY_a$ are called (the logarithm of) the {\it Fock-Goncharov (FG) coordinate} on $\mathcal{S}_a$ \cite{Fock:2003alg}. Same as their anti-holomorphic counterparts. 
In fact, each of the six edges in the ideal triangulation of each 4-holed sphere $\cS_a\in\partial(S^3\backslash\Gamma_5)$ is addressed with an FG coordinate, denoted as $\chi^{(a)}_{ij}$ when the edge is shared by two ideal octahedra $\Oct(i)$ and $\Oct(j)$. 

We also denote the continuous and discrete parameterization of the new set of coordinates as follows. 
\be
\vec{\mu}=\f{k}{2\pi Q}\lb \vec{\cQ} +\vec{\tcQ} \rb\,,\,\, 
\vec{m}=\f{ik}{2\pi bQ}\lb \vec{\cQ}-b^2\vec{\tcQ} \rb\,,\,\, 
\vec{\nu}=\f{k}{2\pi Q}\lb \vec{\cP}+\vec{\tcP} \rb\,,\,\,  
\vec{n}=\f{ik}{2\pi bQ}\lb \vec{\cP}-b^2\vec{\tcP}\rb\,.
\label{eq:new_coordinate_param}
\ee
We will also use the notations $\mu_{ab},m_{ab}$ (\resp $\nu_{ab},n_{ab}$) to denote the coordinates corresponding to $2L_{ab}$ (\resp $T_{ab}$) and use $\mu_{a},m_{a}$ (\resp $\nu_{a},n_{a}$) to denote the coordinates corresponding to $\cX_{a}$ (\resp $\cY_{a}$). 

The partition takes the following expression \cite{Han:2021tzw} \footnote{The integration contour $\cC^{\times 15}$ in \eqref{eq:partition_S3G5_1} is chosen to be on the plane $\R^{15}+i\vec{\alpha}_2$ where $\vec{\alpha}_{2}$ is within the first vector component of the positive angle structure $\fP_{2}$ after the $U$-type and $T$-type transformations. $\fP_2$ is related to the positive angle structure $\fP(\oct)^{\times 5}$ of 5 ideal octahedra in the following way (see Appendix \ref{app:partition}):
\be
\fP_{2}=T\circ U\circ \fP(\oct)^{\times 5}
\quad \Rightarrow\quad
\text{ If } \, 
(\vec{\alpha}_0,\vec{\beta}_0)\in\fP(\oct)^{\times 5}\,,\quad \text{ then }\,
(\vec{\alpha}_2,\vec{\beta}_2)=(-(\bB^{-1})^\top\vec{\alpha}_0,-\bB\vec{\beta}_0-\bA\vec{\alpha}_0)\in \fP_{2}\,.
\nn\ee}
\be
\cZ'_{S^3\backslash\Gamma_5}(\vec{\mu}|\vec{m})
=\f{4i}{k^{15}}\sum_{\vec{n}\in(\Z/k\Z)^{15}}\int_{\cC^{\times 15}}\rd^{15}\vec{\nu}\,
(-1)^{\vec{n}\cdot \bA\bB^\top\cdot \vec{n}}
e^{\f{i\pi}{k}(-\vec{\nu}\cdot \bA\bB^\top\cdot \vec{\nu}+\vec{n}\cdot \bA\bB^\top\cdot \vec{n}})
 e^{\f{2\pi i}{k}\left[-\vec{\nu}\cdot(\vec{\mu}-\f{iQ}{2}\vec{t})+\vec{n}\cdot \vec{m}\right]}\cZ_\times(-\bB^\top\vec{\nu}|-\bB^\top\vec{n})\,,
\label{eq:partition_S3G5_1}
\ee
where $\bA$ and $\bB$ are $15\times 15$ matrices with integer entries and $\vec{t}$ is a vector with integer elements. See \eqref{eq:ABt+} for the explicit expressions of $\bA$, $\bB$ and $\vec{t}$. They correspond to the order $\{(12),(13),(14),(15),(23),(24),(25),(34),(35),(45)\}$ of the annuli $(ab)$'s. We will use this ordering throughout this paper. 
$\cZ_\times$ is a product of 5 partition functions $\cZ_{\oct}$'s for ideal octahedra:
\be
\cZ_\times(\vec{\mu}|\vec{m}):=\prod_{a=1}^{5}\cZ_{\oct}(x_a,y_a,z_a;\xt_a,\yt_a,\zt_a)\,,
\label{eq:Z_times}
\ee
where each $\cZ_{\oct}$ is a product of 4 partition functions for ideal tetrahedra:
\be
\cZ_{\oct}(x,y,z;\xt,\yt,\zt ) := \prod_{i,j,k,l=0}^{\infty} 
\f{1-\qt^{i+1}\xt^{-1}}{1-q^{-i}x^{-1}}
\f{1-\qt^{j+1}\yt^{-1}}{1-q^{-j}y^{-1}}
\f{1-\qt^{k+1}\zt^{-1}}{1-q^{-k}z^{-1}}
\f{1-\qt^{l}\xt\yt\zt}{1-q^{-l-1}xyz}\,.
\label{eq:Z_oct}
\ee
On the right-hand sides of both \eqref{eq:Z_times} and \eqref{eq:Z_oct}, the variables $x,y,z$ and $\xt,\yt,\zt$ are the edge coordinates on the ideal octahedra. Their logarithmics are parametrized in the same way as in \eqref{eq:Q_P_param}.  That is,
\be
\fz_a=\exp\left[\f{2\pi i}{k}\lb -ib\mu_{\fz_a}-m_{\fz_a} \rb \right]\,,\quad 
\tilde{\fz}_a=\exp\left[\f{2\pi i}{k}\lb -ib^{-1}\mu_{\fz_a}+m_{\fz_a} \rb \right]\,,\quad
\text{with}\,\,\fz_a=x_a,y_z,z_a\,,\quad
\tilde{\fz}_a=\xt_a,\yt_z,\zt_a\,.
\ee
These parameters give the entries of the variables $\vec{\mu},\vec{m}$ on the left-hand side of \eqref{eq:Z_times}. We refer to Appendix \ref{subsec:ideal_tetra} and \ref{subsec:ideal_octa} for a more systematic derivation of these partition functions. 

Observe that $\bA\bB^\top$ is a symmetric matrix with integer entries, $(-1)^{\vec{n}\cdot\bA\bB^\top\cdot\vec{n}}$ in \eqref{eq:partition_S3G5_1} can be simplified to be $(-1)^{\vec{D}\cdot\vec{n}}$ where $\vec{D}:=\diag(\bA\bB^\top)$ is a vector whose elements are the diagonal elements of $\bA\bB^\top$. The sign $(-1)^{\vec{n}\cdot\bA\bB^\top\cdot\vec{n}}$ depends on the parity of elements in $\vec{D}$ and $\vec{n}$. Also notice that the parity of $D_I$ is the same as the parity of $t_I$, $\forall I=1,\cdots,15$. Combining these facts, we rewrite the sign factor $(-1)^{\vec{n}\cdot\bA\bB^\top\cdot\vec{n}}$ in \eqref{eq:partition_S3G5} to be $(-1)^{\vec{t}\cdot\vec{n}}$ \footnote{One can check using the explicit expressions \eqref{eq:ABt+} of matrices $\bA$, $\bB$ and vector $\vec{t}$ that the odd elements of $\vec{D}$ and $\vec{t}$ are both the $1$th, $2$nd, $6$th, $8$th, $11$th, $12$th and $13$th elements.}. Different from \cite{Han:2021tzw}, we will use the following expression for the Chern-Simons partition function on $S^3\backslash\Gamma_5$: 
\be
\cZ_{S^3\backslash\Gamma_5}(\vec{\mu}|\vec{m})
=\f{4i}{k^{15}}\sum_{\vec{n}\in(\Z/k\Z)^{15}}\int_{\cC^{\times 15}}\rd^{15}\vec{\nu}\,
{(-1)^{\vec{t}\cdot\vec{n}}}
e^{\f{i\pi}{k}(-\vec{\nu}\cdot \bA\bB^\top\cdot \vec{\nu}+\vec{n}\cdot \bA\bB^\top\cdot \vec{n})}
 e^{\f{2\pi i}{k}\left[-\vec{\nu}\cdot(\vec{\mu}-\f{iQ}{2}\vec{t})+\vec{n}\cdot \vec{m}\right]}\cZ_\times(-\bB^\top\vec{\nu}|-\bB^\top\vec{n})\,.
\label{eq:partition_S3G5}
\ee
We will see in Section \ref{sec:radiative_correction} that such a change will not alter the equations of motion compared to \cite{Han:2021tzw}.

The finiteness of $\cZ_{S^3\backslash\Gamma_5}(\vec{\mu}|\vec{m})$ is guaranteed by the so-called {\it positive angle structure} $\fP(S^3\backslash\Gamma_5)$ which is proven in \cite{Han:2021tzw} to be non-empty. 
Given an $2N$-dimensional positive angle structure $\fP$, we define the functional space
\be
\cF_{\fP}=\left\{ \text{holomorphic }f:\bC^N\rightarrow \bC \mid \forall(\vec{\alpha},\vec{\beta})\in\fP\,, e^{-\f{2\pi}{k}\vec{\beta}\cdot\vec{\mu}}f(\vec{\mu}+i\vec{\alpha})\in \cS(\R^N)\, \text{ is Schwartz class} \right\}\,.
\ee
Combining a discrete representation part $\lb\bC^{k}\rb^{\otimes N}$, we define
\be
\cF_{\fP}^{(k)}=\cF_{\fP}\otimes_{\bC}(\bC^k)^{\otimes N}\,.
\label{eq:def_Fk}
\ee 
In our case, $N=15$. By the theorem \cite{EllegaardAndersen:2011vps,andersen2013new,Dimofte:2014zga} that the Chern-Simons partition function converges absolutely as long as the 3-manifold admits a non-empty positive angle structure, the finiteness of the Chern-Simons partition function on $S^3\backslash\Gamma_5$ is manifest. 
This means, given any $(\vec{\alpha},\vec{\beta})\in \fP(S^3\backslash\Gamma_5)$ and let $\im(\vec{\mu})=\vec{\alpha}$, the integration contours $\cC^{\times 15}$ satisfying $\im(\vec{\nu})=\vec{\beta}$ renders the finiteness of $\cZ_{S^3\backslash\Gamma_5}(\vec{\mu}|\vec{m})$, or in other words, $\cZ_{S^3\backslash\Gamma_5}\in \cF_{S^3\backslash\Gamma_5}^{(k)}$.

\subsection{Impose the simplicity constraints towards a spinfoam vertex amplitude}
\label{subsec:simplicity_constraint}

The second step in constructing the vertex amplitude is to impose the simplicity constraints. Ref.\cite{Han:2021tzw} applies the spinfoam techniques, especially those applied to the EPRL-FK model \cite{Engle:2007wy,Freidel:2007py}. As discussed below, the simplicity constraints contain the first-class and second-class pieces, according to the Chern-Simons symplectic structure. The first-class constraints are imposed strongly on $\cZ_{S^3\backslash \Gamma_5}$ which amount to restricting the FN coordinates on the annuli. On the other hand, the second-class constraints are imposed weakly on the nodes of $\Gamma_5$. This is done by firstly coupling $\cZ_{S^3\backslash\Gamma_5}$ with 5 coherent states, each on one node of $\Gamma_5$ which is peaked at certain phase space point in $\cM_{\Flat}(S^3\backslash\Gamma_5,\PSL(2,\bC))$, then imposing constraints to the allowed phase space points where the coherent states are peaked. 

The simplicity constraints (see below) imposed on the Chern-Simons theory on $S^3\backslash\Gamma_3$ can be seen as the generalization of the simplicity constraints in the EPRL-FK model. Recall that, at the classical discrete level, the simplicity constraints in the (Lorentzian) EPRL-FK model are \cite{Engle:2007qf,Engle:2007wy,Freidel:2007py}
\begin{subequations}
\begin{align}
&\text{first-class (diagonal constraints): }\quad\quad\quad\,\,\,\, \epsilon_{IJKL}B_f^{IJ}(t)B_f^{KL}(t)=0\,,\quad\forall f\in t\,,
\label{eq:EPRL_1st_calss}\\
&\text{second-class (off-diagonal constraints): }\quad  \epsilon_{IJKL}B_f^{IJ}(t)B_{f'}^{KL}(t)=0 \,,\quad\forall f,f'\in t\,, f\neq f'\,,
\label{eq:EPRL_2nd_calss}
\end{align}
\label{eq:EPRL_simplicity}
\end{subequations}
where $f$ and $t$ denote a triangle and a tetrahedron respectively and $f\in t$ denotes that $f$ is on the boundary of $t$. $B_f^{IJ}(t)=\int_f B^{IJ}(t)$ is the discretized $B$-field associated to $f$ with $I,J=0,1,2,3$ being the internal labels and 0 is identified to be the time direction. 
These quadratic constraints can be strengthened to a single set of linear constraints
\be
\text{linear constraints: }\quad \exists\,N_J \text{ such that } N_JB_f^{IJ}(t)=0\,,\quad\forall f\in t\,.
\label{eq:EPRL_linear}
\ee
The replacement from \eqref{eq:EPRL_simplicity} to \eqref{eq:EPRL_linear} is for the purpose of selecting a single solution sector and is beneficial for quantization. We will treat \eqref{eq:EPRL_linear} as the full set of simplicity constraints, different from the original papers \cite{Engle:2007qf,Engle:2007wy} while following \cite{Ding:2010fw}, and generalize it in the new spinfoam model.

The simplicity constraints then imply that the discretized $B$-field $B^{IJ}_f(t)$ measures the area $\fa_f=|\f12\epsilon_{IJKL}N^JB^{KL}_f(t)|$ of the triangle $f$. One can gauge fix the vector $N_J=N_0$ to be timelike, then \eqref{eq:EPRL_linear} is equivalent to the statement that the tetrahedron $t$ is spacelike. Moreover, the SU(2) gauge symmetry implies the closure condition in the EPRL-FK model. That is, for each tetrahedron $t$:
\be
\sum_{f\in t} B_f^{IJ}(t)=0 
\quad\Longleftrightarrow\quad
\sum_{f\in t}\fa_f \fn_f^I=0\,,
\label{eq:EPRL_closure}
\ee
where $\fn_f^I$ is the normal vector to $f$ satisfying $|\fn_f|=1$. 
By Minkowski's theorem, the simplicity constraint \eqref{eq:EPRL_linear} together with the closure condition \eqref{eq:EPRL_closure} allows us to identify a convex tetrahedron whose face areas and normals are given by $\fa_f$'s and $\fn_f^I$'s. 

The generalization of simplicity constraints to the $\Lambda\neq 0$ case at the discrete level can be implemented as follows.
Consider the (non-ideal) triangulation, denoted as $\tau_a$, of a 4-holed sphere $\cS_a$ such that each hole, denoted by $\fp$, is inside a triangle $f_\fp$. See the red lines in fig.\ref{fig:XY_choice}. Define the the discretized $B$-field associated to $f_\fp$ as in the EPRL-FK model, \ie $B_{f_\fp}(\tau_a)=\int_{f_\fp} B(\tau_a)$. 
One the other hand, let us recall the relation $\cF=\f{|\Lambda|}{3}B$ discussed in \eqref{eq:simplicity_constraint}. The discretization of this relation gives $\cF_\fp(\cS_a)=\f{|\Lambda|}{3}B_{f_\fp}(\tau_a) \delta(\vec{x})\rd x^1\w\rd x^2$ at the local coordinate $(x^1,x^2)$ on one patch of $\cS_a$ with the hole $\fp$ at the origin. it allows us to write the simplicity constraints in the same form as \eqref{eq:EPRL_linear} in terms of the Chern-Simons curvature. That is, for all holes $\fp$'s of $\cS_a$,
\be
 \exists\,  N_J \text{ such that } N_J\cF_\fp^{IJ}(\cS_a)=0\,.
\label{eq:CS_simplicity}
\ee
By the non-abelian stock's theorem, the holonomy around each triangle $f_\fp$ of $\tau_a$ takes the form $O_{f_\fp}(\tau_a)=e^{\f{|\Lambda|}{3}B_{f_\fp}(\tau_a)}\in\PSL(2,\bC)$. Eq.\eqref{eq:CS_simplicity} can be translated into constraints in terms of $\{O_{f_\fp}(\tau_a)\}_{\fp=1}^4$:
\be
 \exists\,  N_J \text{ such that } N_J(O_{f_\fp})_{I}^{\phantom{I}J}(\tau_a)=N_I\,,\quad\forall f_\fp \in \tau_a\,.
 \label{eq:CS_discrete_simplicity}
\ee 
Similar to the EPRL-FK case, \eqref{eq:CS_simplicity} (or \eqref{eq:CS_discrete_simplicity}) means that the 4-holed sphere $\cS_a$, or its triangulation $\tau_a$, is orthogonal to a common vector $N^J\in\R^4$. Gauge fixing $N_J=(1,0,0,0)$ implements that all the holonomies $\{O_{f_\fp}(\tau_a)\}_{\fp=1}^{4}$ are in a common $\PSU(2)$ subgroup of $\PSL(2,\bC)$. 
In other words, the simplicity constraints restrict the moduli space $\cM_\Flat(\cS_a,\PSL(2,\bC))$ of flat $\PSL(2,\bC)$ connection to a moduli space $\cM_\Flat(\cS_a,\PSU(2))$ of flat $\PSU(2)$ connection, which is a symplectic submanifold of $\cM_\Flat(\cS_a,\PSL(2,\bC))$. 

The flat connection in $\cM_\Flat(\cS_a,\PSU(2))$ defines a representation of the fundamental group of $\cS_a$ into $\PSU(2)$ modulo gauge transformations. Let the holonomies $\{O_{f_\fp}(\tau_a)\}$ have the same base point $\fb\in\cS_a$. Then they satisfy the non-linear closure condition (we fix the ordering of the holonomies here and for the rest of this paper)
\be
O_{f_1}(\tau_a)O_{f_2}(\tau_a)O_{f_3}(\tau_a)O_{f_4}(\tau_a)=\id_{\PSU(2)} 
\ee
due to the isomorphism
\be
\cM_\Flat(\cS_a,\PSU(2))\cong\{O_1,O_2,O_3,O_4\in\PSU(2):O_1O_2O_3O_4=\id_{\PSU(2)}\}/\PSU(2)\,.
\ee
The correspondence between $\PSU(2)$ flat connection and constant curvature tetrahedron has been established in \cite{Haggard:2015ima}.

The simplicity constraint on $O_{f_\fp}(\tau_a)$ can be expressed in terms of the coordinates $(\vec{\mathcal{Q}},\vec{\mathcal{P}})$ defined in \eqref{eq:change_coordinate}. We will classify the constraints into first- and second-class parts and treat them differently in the following.

\subsubsection{The first-class simplicity constraints}

The first-class constraints are obtained by the commutative functions of the holonomies $\{O_{f_\fp}(\tau_a)\}$. In $\partial(S^3\backslash\Gamma_5)$, a hole $\fp$ of $\cS_a$ is connected to a hole of $\cS_b(\neq \cS_a)$ via a annulus cusp. Classically, $O_{f_\fp}(\tau_a)\in\PSU(2)$ implies that $\lambda_{\fp}^2\equiv\lambda_{ab}^2=e^{i2\theta_{ab}}$ with some $\theta_{ab}\in\R$. Ref.\cite{Haggard:2015ima} has shown that $\theta_{ab}$ encodes the area $\fa_{f_p}$ of the triangle $f_\fp$ surrounding $\fp$ in the triangulation $\tau_a$.
Therefore, the first-class simplicity constraints can be formulated as
\be
2L_{ab}:=\f{2\pi i}{k}\lb -ib\mu_{ab}-m_{ab}\rb\in i\R 
\quad\Longleftrightarrow\quad
\mu_{ab}=0
\quad\xrightarrow{\text{quantization}}\quad
\re(\bmu_{ab})\cZ_{S^3\backslash\Gamma_5}(\vec{\mu}|\vec{m})=0\,,
\label{eq:1st_class_quantum}
\ee
where the right-most quantum constraint is written in terms of $\re(\bmu_{ab})$ as the analytic continuation of $\mu_{ab}$ to $\bC$ is allowed at the quantum level. If the requirement ``4d area = 3d area'' \cite{Ding:2010fw} is further imposed, the first-class constraint is strengthened to $\bmu_{ab}\cZ_{S^3\backslash\Gamma_5}(\vec{\mu}|\vec{m})=0$. Following \cite{Han:2021tzw}, we keep the weaker condition $\im(\mu_{ab})\equiv \alpha_{ab}\neq 0$. Then $e^{2L_{ab}}\in \bU(1)$ is realized only at the classical level. Define the ``spin'' $j_{ab}$ such that
\be
2j_{ab}= m_{ab}\quad\rightarrow\quad
j_{ab}=0,\f12,\cdots,\f{k-1}{2}\,.
\label{eq:m_to_j}
\ee
$j_{ab}$ encodes the area $\fa_{f_\fp}$ of the triangle $f_\fp$ in a tetrahedron (when we fix the orientation of $f_\fp$) by \cite{Haggard:2015ima} \footnote{If the orientation of $f_\fp$ is not fixed, there is an ambiguity for the area $\fa_{f_\fp}$ for a given $j_{ab}$. More precisely, the area is related to $j_{ab}$ by \eqref{eq:spin_to_area} or $2\pi-\f{|\Lambda|}{3}\fa_{f_\fp}=\f{4\pi }{k}j_{ab}$.} 
\be
\f{|\Lambda|}{3}\fa_{f_\fp}=\f{4\pi }{k}j_{ab}. 
\label{eq:spin_to_area}
\ee
The quantum states satisfying the constraint \eqref{eq:1st_class_quantum} are then labelled by
\be
\cZ_{S^3\backslash\Gamma_5}(\{i\alpha_{ab}\}_{(ab)}, \{\mu_{a}\}\mid \{j_{ab}\}_{(ab)}, \{m_{a}\})\,.
\label{eq:Z_after_1st_constraints}
\ee 

Therefore, the first-class simplicity constraints can be seen to be imposed on the FN coordinates on the annulus cusps on the triangulation of $\partial(S^3\backslash\Gamma_5)$. The remaining (second-class) simplicity constraints will be imposed on each $\cS_a$.

\subsubsection{The second-class simplicity constraints and the Chern-Simons coherent states}
\label{subsubsec:second_simplicity}

The moduli space $\cM_\Flat(\cS_a,\PSL(2,\mathbb{C}))$ is not a symplectic manifold but a Poisson manifold, due to the presence of Poisson commutative $\{\lambda^2_{\fp}\}_{\fp=1}^4$. Fixing $\{\lambda^2_{\fp}\}_{\fp=1}^4$ by \eqref{eq:spin_to_area}
 reduces the moduli space $\cM_\Flat(\cS_a,\PSL(2,\mathbb{C}))$ to a two-complex-dimensional symplectic space $\mathcal{M}_{\vec{\lambda}}$ with symplectic coordinates $(\cX_a,\cY_a)$, on which we should impose the second-class simplicity constraints. 

It will be more convenient to work with the trace coordinates of flat connections rather than the FG coordinates $(\cX_a,\cY_a)$ when analyzing these simplicity constraints \footnote{See \cite{Han:2021tzw} for using spinors instead of trace coordinates to impose constraints on $(\cX_a,\cY_a)$.}. 
Consider the triangulation $\tau_a$ of $\cS_a$ as described above. Label the holes by numbers $1,2,3,4$ and denote each edge connecting the holes $\fp_1$ and $\fp_2$ ($\fp_i=1,2,3,4$) by $e_{\fp_1\fp_2}$. Denote the (exponential) FG coordinate on $e_{\fp_1\fp_2}$ as $z_{\fp_1\fp_2}$. With no loss of generality, let $z_{12}=e^{\cX_a}$ and $z_{13}=e^{\cY_a}$ as shown in fig.\ref{fig:XY_choice}. 
(If this choice is taken for all the five $\{\cS_a\}_{a=1}^5$, the way of gluing different 4-holed spheres is unique. See Appendix \ref{app:identify_holes} for details of the gluing.) We choose a lift by defining $y_{\fp_1\fp_2}:=\sqrt{-z_{\fp_1\fp_2}}\equiv \exp\lb\f12 \lb Z_{\fp_1\fp_2}+i\pi\rb\rb$ for all edges $\{e_{\fp_1\fp_2}\}$ and work with $\mathrm{SL}(2,\mathbb{C})$ flat connections in stead of $\mathrm{PSL}(2,\mathbb{C})$ flat connections.  
\begin{figure}[h!]
\centering
\includegraphics[width=0.3\textwidth]{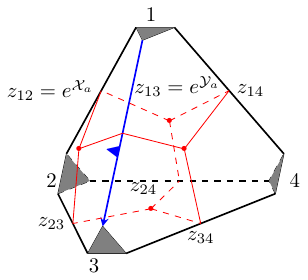}
\caption{The ideal triangulation ({\it in black}) and the (normal) triangulation $\tau_a$ ({\it in red}) of a 4-holed sphere $\cS_a$. Numbers $1,2,3,4$ label the holes of $\cS_a$.  A choice of dressing $e^{\cX_a}=z_{12},e^{\cY_a}=z_{13}$ in terms of the edge coordinates  is given. The relative location of holes is consistent with the Poisson relation of $\cX_a$ and $\cY_a$. The arrow in blue dressed with a fin -- called a {\it snake} -- is used to calculate the holonomies around holes by the snake rule \cite{Dimofte:2013lba}. See also Appendix \ref{app:snake} for a brief description of the snake rule.}
\label{fig:XY_choice}
\end{figure}

\medskip

\noindent{\bf Trace coordinates of $\cM_\Flat(\cS_a,\SL(2,\bC))$. }
In order to write the trace coordinate explicitly, we now work on one 4-holed sphere and lift the holonomies $O_{f_\fp}(\tau_a)\in\PSL(2,\bC)$ to $h_\fp \in\SL(2,\bC)$ for all holes. They describe solutions to $\cM_{\Flat}(\cS_a,\SL(2,\bC))$ by the closure constraint
\be
h_1h_2h_3h_4=\id_{\SL(2,\bC)}\,.
\ee
$\{h_\fp\}$ can be calculated by the snake rule \cite{Dimofte:2013lba} (see Appendix \ref{app:snake}) based on the ideal triangulation of $\cS_a$ (see fig.\ref{fig:XY_choice}).  
Their traces are determined by commutative eigenvalues $\{\lambda_\fp\}$ when we choose the lift $y_{\fp_1\fp_2}=\sqrt{-z_{\fp_1\fp_2}}$ \footnote{If we choose the other lift $y_{\fp_1\fp_2}=-\sqrt{-z_{\fp_1\fp_2}}$, the traces become $\fm_\fp=-\lambda_\fp-\lambda_\fp^{-1}.$}
\be
\fm_1:=\tr(h_1)=\lambda_1+\lambda_1^{-1}\,,\quad
\fm_2:=\tr(h_2)=\lambda_2+\lambda_2^{-1}\,,\quad
\fm_3:=\tr(h_3)=\lambda_3+\lambda_3^{-1}\,,\quad
\fm_4:=\tr(h_4)=\lambda_4+\lambda_4^{-1}\,.
\label{eq:trace}
\ee

Apart from $\{\fm_\fp\}$, two more trace coordinates are needed to describe $\cM_{\Flat}(\cS_a,\SL(2,\bC))$. They correspond to holonomies $h_{\fp_1\fp_2}$ around two holes $\fp_1$ and $\fp_2$. The snake rule gives 
\begin{subequations}
\begin{align}
\tr(h_{12})&={-}\f{y_{14}^2 y_{24}^2 y_{13}^2+y_{14}^2 y_{23}^2 y_{24}^2 y_{13}^2+y_{23}^2 y_{24}^2 y_{13}^2+y_{24}^2 y_{13}^2+y_{13}^2+y_{24}^2+1}{y_{13} y_{14} y_{23} y_{24}}\,,\\
\tr(h_{23})&=
{-}\f{y_{12}^2 y_{13}^2 y_{24}^2 y_{34}^2+y_{12}^2 y_{13}^2 y_{34}^2+y_{12}^2 y_{24}^2 y_{34}^2+y_{12}^2 y_{34}^2+y_{12}^2+y_{34}^2+1}{y_{12} y_{13} y_{24} y_{34}},\\
\tr(h_{13})&=-\f{y_{12}^2 y_{14}^2 y_{23}^2 y_{34}^2+y_{12}^2 y_{14}^2 y_{23}^2+y_{14}^2 y_{23}^2 y_{34}^2+y_{14}^2 y_{23}^2+y_{14}^2+y_{23}^2+1}{y_{12} y_{14} y_{23} y_{34}}\,.
\end{align}
\label{eq:trace_coordinates}
\end{subequations}
These expressions are consistent with those in \cite{Gaiotto:2009hg,Coman:2015lna}. 
On the other hand,  the traces of monodromies around one hole are fixed by the first-class constraints, \ie
\be
y_{12} y_{13} y_{14}=\lambda_1\,,\quad
y_{12} y_{23} y_{24}=\lambda_2\,,\quad
y_{13} y_{23} y_{34}=\lambda_3\,,\quad
y_{14} y_{24} y_{34}=\lambda_4\,,\quad
\text{where }\, \lambda_{\fp}:=e^{L_{ab}}\,.
\label{eq:lambda}
\ee
Eq.\eqref{eq:lambda} can be inserted into \eqref{eq:trace_coordinates} to rewrite the trace coordinates in terms of $\{\lambda_{\fp}\}_{\fp=1}^4$ and $z_{12},z_{13}$. 
\begin{subequations}
\begin{align}
\fkt_1:=\tr(h_{12})&
=\frac{\lambda _1 \lambda _2}{z_{12}}+\frac{-z_{12} z_{13}+z_{12}+\lambda _2^2 z_{13}}{\lambda _1 \lambda _2}+\frac{\lambda _4 z_{12} (z_{13}-1) z_{13}}{\lambda _1^2 \lambda _3}+\frac{\lambda _4 z_{13}}{\lambda _3}\,,\\
\fkt_2:=\tr(h_{23})&
=\frac{\lambda _1 (z_{12}-1)}{\lambda _4 z_{12} z_{13}}+\frac{\lambda _3 (z_{12} (z_{13}-1)+1)}{\lambda _2 z_{13}}-\frac{\lambda _4 z_{12} (z_{13}-1)}{\lambda _1}\,,\\
\fkt_3:=
\tr(h_{13})&=-\frac{\lambda _2 \lambda _1^2 (z_{12}-1)}{\lambda _4 z_{12}^2 z_{13}}+\frac{\lambda _2}{\lambda _4 z_{12}}+\frac{\lambda _1 \left(\lambda _3^2 (z_{12}-1)+z_{13}\right)}{\lambda _3 z_{12} z_{13}}+\frac{z_{13}}{\lambda _1 \lambda _3}\,.
\end{align}
\label{eq:trace_coordinates_2}
\end{subequations}
The algebra functions on $\cM_{\Flat}(\cS_a,\SL(2,\bC))$ can be described by the polynomial ring generated by the trace coordinates $\{\fm_1,\fm_2,\fm_3,\fm_4,\fkt_1,\fkt_2,\fkt_3\}$ quotient by a polynomial relation \cite{Teschner:2013tqy,Coman:2015lna,Nekrasov:2011bc} 
\be
P=\fkt_1\fkt_2\fkt_3+\fkt_1^2+\fkt_2^2+\fkt_3^2+\fm_1\fm_2\fm_3\fm_4+\fm_1^2+\fm_2^2+\fm_3^2+\fm_4^2
-(\fm_1\fm_2+\fm_3\fm_4)\fkt_1 -(\fm_2\fm_3+\fm_1\fm_4)\fkt_2 - (\fm_1\fm_3+\fm_2\fm_4)\fkt_3 -4\,.
\label{eq:polynomial}
\ee
It can be easily verified that $\{\fm_1,\fm_2,\fm_3,\fm_4,\fkt_1,\fkt_2,\fkt_3\}$ defined by \eqref{eq:trace} and \eqref{eq:trace_coordinates_2} is a set of solutions to $P=0$. The second-class simplicity constraints are implemented by
\be
\fkt_1,\fkt_2,\fkt_3\in[-2,2]\,,
\label{eq:simplicity_t123}
\ee
where only two are independent as they are functions of $z_{12},z_{13}$.
Inversely solving $z_{12},z_{13}$ from given $\{\fm_1,\fm_2,\fm_3,\fm_4,\fkt_1,\fkt_2\}$ satisfying the simplicity constraints, one can find two solutions. Indeed, $P$ is a quadratic polynomial of $\fkt_3$ hence there are generally two roots to $\fkt_3$ given data of $\{\fm_1,\fm_2,\fm_3,\fm_4,\fkt_1,\fkt_2\}$ which corresponds to these two solutions of $\{z_{12},z_{23}\}$. However, further knowing $\fkt_3$ uniquely fixes to one of the solutions. 

 \medskip

\noindent{\bf Darboux coordinates of $\cM_\Flat(\cS_a,\SL(2,\bC))$. }
The trace coordinates are not the symplectic coordinates on the moduli space of flat connection on the 4-holed sphere (see e.g. \cite{Goldman1986,Nekrasov:2011bc,Coman:2015lna} for the discussion about their Poisson brackets). In order to have a well-defined state integral for the spinfoam amplitude, one needs to replace them with a new set of symplectic coordinates, which can be defined as follows. 

The holomorphic Darboux coordinates $(\theta,\phi)$ of $\cM_{\vec \lambda}$ 
relate to $\fkt_1=\tr(h_{12})\,,\fkt_2=\tr(h_{23})$ and $\fkt_3=\tr(h_{13})$ by (see e.g \cite{Nekrasov:2011bc}).
\begin{subequations}
\begin{align}
2\cos\theta&=\fkt_1\,,\\
2\cos\phi\sqrt{c_{12}c_{34}}&=
\fkt_2(\fkt_1^2-4)+2(\fm_1\fm_4+\fm_2\fm_3)-\fkt_1(\fm_1\fm_3+\fm_2\fm_4)\,,
\label{eq:theta_phi_2}\\
\sin\phi \sqrt{c_{12}c_{34}}&= (2\fkt_3+\fkt_1\fkt_2-\fm_1\fm_3-\fm_2\fm_4)\sin\theta\,,
\label{eq:t3_theta_phi}
\end{align}	
\label{eq:theta_phi}
\end{subequations}
where 
\be
c_{ij}=\fkt_1^2+\fm_i^2+\fm_j^2-\fkt_1\fm_i\fm_j-4\,,\quad i,j=1,\cdots, 4\,.
\label{eq:cij}
\ee
Generically, we can solve for $(\fkt_1,\fkt_2,\fkt_3)$ as functions of $(\theta,\phi)$:
\begin{subequations}
\begin{align}
\fkt_1 =& 2\cos\theta \,,\\
\fkt_2 =& -\frac{1}{2} \csc ^2\theta \left(\cos \phi  \sqrt{c^\theta_{12} c^\theta_{34}}
+\cos \theta (\fm_1 \fm_3+\fm_2 \fm_4)-\fm_1 \fm_4-\fm_2 \fm_3\right)\,,
\label{eq:t2}\\
\fkt_3 =& 
\frac{1}{2} \csc ^2\theta  \left(\sqrt{c^\theta_{12} c^\theta_{34}} \cos (\theta -\phi )-\cos \theta  (\fm_1 \fm_4+\fm_2 \fm_3)+\fm_1 \fm_3+\fm_2 \fm_4\right)\,.
\end{align}
\label{eq:t_to_Darboux}
\end{subequations}
where 
\be
c_{ij}^\theta=\fm_i^2+\fm_j^2+2\cos\theta\,\fm_i\fm_j-4\sin^2\theta\,,\quad i,j=1,\cdots, 4\,.
\label{eq:cij_theta}
\ee
Therefore, given $(\theta,\phi)$, the solution to $\fkt_3$ is fixed from the two solutions solved from $P=0$. 
The FG coordinates $\{z_{\fp_1\fp_2}\}$ become functions of $(\theta,\phi)$, since they are uniquely determined by $(\fkt_1,\fkt_2,\fkt_3)$.
Therefore, the simplicity constraints \eqref{eq:simplicity_t123} can be converted to functions of $(\theta,\phi)$. 

We denote the Darboux coordinates satisfying the simplicity constraints to be $(\htheta,\hphi)$. Together with $\{\fm_1,\fm_2,\fm_3,\fm_4\}$, they uniquely determine the geometry of a (curved) tetrahedron on $S^3$ as follows.

Consider 4 points $\{v_i\}_{i=1}^4$ on $\SU(2)\cong S^3$ located at 
\be
v_1=\id_{\SU(2)}\,,\quad
v_2=h_1\,,\quad
v_3=h_1h_2\,,\quad
v_4=h_1h_2h_3\,.
\ee
\begin{figure}[h!]
    \centering
    \includegraphics[width=0.4\textwidth]{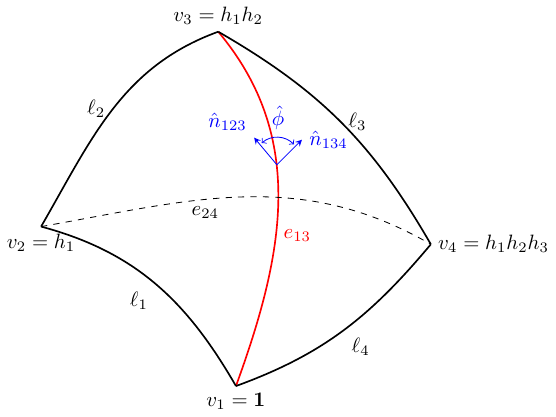}
    \caption{A 4-gon on $\SU(2)\cong S^3$ formed by geodesic curves $\{\ell_i\}_{i=1}^4$ connecting four points $v_1=\id,v_2=h_1,v_3=h_1h_2,v_4=h_1h_2h_3$ in cyclic order. The geodesic curve $e_{13}$ connecting $v_1$ and $v_3$ ({\it in red}) has length $\htheta$. Further connecting $v_2$ and $v_4$ with a geodesic curve $e_{24}$ ({\it dashed}) forms a curved tetrahedron on $S^3$ whose faces are geodesics. $\hat{n}_{123}$ and $\hat{n}_{134}$ ({\it one-way arrows in blue}) are outgoing (relative to the tetrahedron) normal vectors of the geodesic triangle $f_{123}$ bounded by $\ell_1,\ell_2,e_{13}$ and the geodesic triangle $f_{134}$ bounded by $\ell_3,\ell_4,e_{13}$ respectively. $\hphi\in[0,\pi]$ is the dihedral angle between $f_{123}$ and $f_{134}$ around $e_{13}$. }
    \label{fig:4-gon}
\end{figure}
A 4-gon is formed by 4 geodesic curves $\ell_1\equiv e_{12}\,,\,\ell_2\equiv e_{23}\,,\,\ell_3\equiv e_{34}\,,\,\ell_4\equiv e_{41}$ where $e_{ij}$ is the geodesic connecting $v_{i}$ and $v_{j}$, as shown in fig.\ref{fig:4-gon}. The geodesic length $a_i\in[0,\pi)$ of $\ell_i$ satisfies 
\be
\cos(a_i)=\fm_i/2=\cos\left(\frac{2\pi}{k}j_i\right),
\ee
for $i=1,\cdots,4$. $\htheta\in[0,\pi]$ is then the length of the diagonal geodesic curve $e_{13}$ connecting $v_1$ and $v_3$, which separates the 4-gon into two (curved) triangles $f_{123}$ bounded by $\ell_1,\ell_2,e_{13}$ and $f_{134}$ bounded by $\ell_3,\ell_4,e_{13}$. Here $j_i=0,1/2,\cdots,(k-1)/2$, but since our discussion here is semiclassical, we may extend $j_i$ to be continuous and belonging to $[0,k/2)$.

On the other hand, $\hphi\in[0,\pi]$ describes the bending angle between the two triangles. 
Adding the other diagonal geodesic curve (whose geodesic length is $\arccos(\fkt_2/2)$), one forms a curved tetrahedron in $S^3$. 
Given fixed lengths $\{a_1,a_2,a_3,a_4\}$ of the four geodesic curves of the 4-gon, $\htheta$ and $\hphi$ uniquely determine the shape of this curved tetrahedron embedded in $S^3$.
See \cite{Han:2021tzw,Nekrasov:2011bc} for more discussion.  
From their geometrical interpretations, we restrict $(\htheta,\hphi)$ to be real with the range 
\be
\max(|a_1-a_2|,|a_3-a_4|)\leq\htheta\leq\min(a_1+a_2,a_3+a_4)\,,\quad
\hphi\in[0,\pi]\,.
\label{eq:range_theta_phi}
\ee
The range of $\htheta$ corresponds to $c_{12},c_{34}<0$ by \eqref{eq:cij}, which fixes the orientation of the two geodesic triangles separated by the $e_{13}$. 
In this way, a solution to the simplicity constraints can be geometrically described by a curved tetrahedron in $S^3$. We denote the above range of $(\htheta,\hphi)$ by $\overline{\cM}_{\vec j}$. Note that the definition of $\overline{\cM}_{\vec j}$ is valid for continuous $\vec{j}\in [0,k/2)$.

Let us now consider a special limit when $a_1\to 0$ thus $\ell_1$ in fig.\ref{fig:4-gon} shrinks to vanishing. The result will be useful later in the asymptotic analysis (see Section \ref{sec:coherent_integral}). Under this limit, the triangle inequality restricts $\htheta$ to equal $a_2$, so $\overline{\cM}_{\vec j}$ becomes 1-dimensional, in which case $\hat{\phi}$ is the only degree of freedom. Therefore, we have $\fkt_1=\fm_2$ as well as $\fm_1=2$ when $a_1\to 0$. 
Inserting them in \eqref{eq:t_to_Darboux} gives the simple result $\fkt_2=\fm_4$ and $\fkt_3=\fm_3$. The result is also expected since $\fkt_2=\tr(h_{23})$ and $\fkt_3=\tr(h_{13})$ while $h_1\rightarrow\id_{\SU(2)}$ is trivial in this limit. An interesting observation is that the traces $\fkt_1,\fkt_2,\fkt_3$ are independent of $\hphi$ in this limit, and the same is true for the FG coordinates $z_{12},z_{13}$ since they are functions of $\fkt_1,\fkt_2,\fkt_3$, \ie they are constants on $\overline{\cM}_{\vec j}$. {Moreover, when $\ell_1,\ell_2,\ell_3,\ell_4$ all shrink to vanishing, $a_1,a_2,a_3,a_4\rightarrow 0$, $\fkt_1,\fkt_2,\fkt_3\rightarrow2$ and $z_{13}\rightarrow \f{z_{12}-1}{z_{12}}$.}

\medskip

As second-class constraints, we will impose them weakly by using Chern-Simons coherent states, which we define in the following. By definition, coherent states are peaked at the classical phase space points hence the labels of coherent states are given by both the position variables $\{\cX_a\}$ and the momentum variables $\{\cY_a\}$. Recall the notations
\be
\cX_a=\f{2\pi i}{k}\lb -ib\mu_a-m_a\rb\,,\quad
\cY_a=\f{2\pi i}{k}\lb -ib\nu_a-n_a\rb\,.
\ee

\noindent{\bf Chern-Simons coherent states on $\cS_a$. }
After fixing the FN coordinates $\{L_{ab}\}_{(ab)}$ to be given by the spins $\{j_{ab}\}_{(ab)}$, the Hilbert space of each 4-holed sphere $\cS_a$ is locally $\bC^2$. We also fix $\im(\mu_a)=\alpha_a$ and consider the degrees of freedom $\re(\mu_a)\in \R$ and $m_a\in \Z/k\Z$. To simplify the notation, we will denote $\re(\mu_a)$ by $\mu_a\in\R$ in the rest of this subsection. The Hilbert space for $\cS_a$ is
\be
\cH_{\cS_a}= L^2(\R)\otimes_{\bC}\bC^k\,.
\nn\ee

Firstly, the coherent state $\psi_{z_a}^0(\mu)$ on $L^2(\R)$ is defined as 
\be
\psi_{z_a}^0(\mu)= \lb\f{2}{k}\rb^{1/4} e^{-\f{\pi}{k}\lb\mu-\f{k}{\pi\sqrt{2}}\re(z_a)\rb^2}e^{-i\sqrt{2}\mu\im(z_a)}\,,
\label{eq:coherent_2}
\ee
with the over-completeness property
\be
\f{k}{2\pi^2}\int_{\bC}\rd \mathrm{Re}(z_a)\rd\mathrm{Im}(z_a)\,\psi^0_{z_a}(\mu)\bar{\psi}^0_{z_a}(\mu')=\delta_{\mu,\mu'}\,.
\label{eq:over-completeness_coherent}
\ee
The coherent state label $z_a\in\bC$ parameterizing a complex plane is related to the classical coordinates by $z_a=\f{\sqrt{2}\pi}{k}(\mu_a+i\nu_a)$. 

Secondly, the coherent state $\xi_{(x_a,y_a)}(m)$ on $\bC^k$ is labelled by $(x_a,y_a)\in [0,2\pi)\times[0,2\pi)$, which can be viewed as the angle coordinates on a torus $\bT^2$. It is defined as \cite{Gazeau:2009zz} 
\be
\xi_{(x_a,y_a)}(m)=\lb\f{2}{k}\rb^{1/4} e^{\f{ikx_ay_a}{4\pi}}\sum_{p_a\in \Z}
e^{-\f{k}{4\pi}\lb \f{2\pi m}{k}-2\pi p_a-x_a \rb^2} e^{\f{ik}{2\pi}y_a\lb \f{2\pi m}{k}-2\pi p_a-x_a\rb}\,.
\label{eq:coherent_Vk}
\ee
$x_a,y_a$ are related to the classical coordinates by $x_a=\text{mod}(\f{2\pi}{k}m_a,2\pi)\,,y_a=\text{mod}(\f{2\pi}{k}n_a,2\pi)$. The over-completeness property of $\xi_{(x_a,y_a)}(m)$ reads
\be
\f{k}{4\pi^2}\int_{\bT^2} \rd x_a\rd y_a \, \xi_{(x_a,y_a)}(m)\bar{\xi}_{(x_a,y_a)}(m')=\delta_{e^{\f{2\pi i}{k}(m-m')},1}\,.
\ee

The coherent state in $\cH_{\cS_a}$ is the tensor product of these two coherent states \footnotemark{}
\be
\Psi^0_{\rho_a}(\mu|m):=\psi^0_{z_a}\otimes \xi_{(x_a,y_a)} \in \cH_{\cS_a}\,,\quad
\rho_a=(z_a,x_a,y_a)\,,
\label{eq:coherent_for_amplitude}
\ee
with the over-completeness relation
\be
\lb\f{k}{4\pi^2}\rb^2\int_{\bC\times\bT^2}\rd \rho_a\,\Psi^0_{\rho_a}(\mu|m)\bar{\Psi}^0_{\rho_a}(\mu'|m') =\delta_{\mu,\mu'}\delta_{e^{\f{2\pi i}{k}(m-m')},1}\,,
\label{eq:coherent_over-completeness_combined}
\ee
where $\rd\rho_a:=\rd \mathrm{Re}(z_a)\rd\mathrm{Im}(z_a)\rd x_a\rd y_a$. 
{It will be convenient to define $\bar{\rho}_a:=(\bar{z}_a,x_a,-y_a)$ (still with $y\in[0,2\pi)$) then we can write $\bar{\Psi}^0_{\rho_a}(\mu|m)=\Psi^0_{\bar{\rho}_a}(\mu|m)$. }
\footnotetext{The coherent state used in \cite{Han:2021tzw} to define the vertex amplitude is a rescaled version. We change in this paper to use \eqref{eq:coherent_for_amplitude} as this does not change the finiteness of the melonic amplitude, as shown below in Section \ref{sec:melon_graph}. {Apart from that, the coherent states defined in this paper is the complex conjugate of those defined in \cite{Han:2021tzw}.}} 
It is easy to confirm that the expectation values of the operators $\bmu,\bnu,{\bf m},{\bf n}$ calculated by the coherent state $\Psi^0_{\rho_a}(\mu|m)$ are given by the coherent state labels, or the classical phase space coordinate at large $k$ limit, \ie 
\be
\la\bmu\ra\xrightarrow{k\rightarrow\infty}\mu_a\,,\quad
\la\bnu\ra\xrightarrow{k\rightarrow\infty}\nu_a\,,\quad 
\la\exp(\frac{2\pi i}{k}{\bf m})\ra\xrightarrow{k\rightarrow\infty}\exp(\frac{2\pi i}{k}m_a)\,,\quad
\la\exp(\frac{2\pi i}{k}{\bf n})\ra=n_a\,.
\label{eq:expectation}
\ee
It is only valid at the large $k$ limit since the torus part of the coherent state $\xi_{(x,y)}(m)$ is normalized only at this limit. We give a derivation for \eqref{eq:expectation} in Appendix \ref{app:expectation}.

The transformation from 
\begin{eqnarray}
    \mathcal{X}_a	=\frac{2\pi i}{k}(-ib\mu_a-m_a)\,,\quad
    \bar{\mathcal{X}}_a=\frac{2\pi i}{k}(-ib^{-1}\mu_a+m_a)\,,\quad 
\mathcal{Y}_a	=\frac{2\pi i}{k}(-ib\nu_a-n_a),\,,\quad
\bar{\mathcal{Y}}_a=\frac{2\pi i}{k}(-ib^{-1}\nu_a+n_a)
\end{eqnarray}
to $\theta,\bar{\theta},\phi,\bar{\phi}$ is canonical \cite{Nekrasov:2011bc}.  
So the following change of variables in the integral has only a constant Jacobian 
\begin{eqnarray}
\int\rd \rho_a\cdots 
&=&\frac{1}{2}\int\rd \left(\frac{2\pi}{k}\mu_a\right)\wedge\rd\left(\frac{2\pi}{k}\nu_a\right)\wedge\rd \left(\frac{2\pi}{k}m_a\right)\wedge \rd\left(\frac{2\pi}{k}n_a\right) \cdots \nonumber\\
&=&\frac{1}{2Q^2}\int\rd \cX_a\wedge\rd\bar{\cX}_a\wedge\rd \cY_a\wedge \rd\bar{\cY}_a \cdots=-\frac{1}{2Q^2}\int\rd \Omega_a\wedge \rd\bar{\Omega}_a \cdots\nonumber\\
&=&\frac{1}{2Q^2}\int\rd \theta_a\wedge\rd\bar{\theta}_a\wedge\rd \phi_a\wedge \rd\bar{\phi}_a \cdots
\end{eqnarray}
where $\Omega_a$ is the holomorphic Atiyah-Bott-Goldman symplectic from on $\cS_a$ with fixed $\{\lambda_{\mathfrak{p}}\}$, and $\cdots$ stands for $\Psi^0_{\rho_a}(\mu|m)\bar{\Psi}^0_{\rho_a}(\mu'|m')$.

The imposition of the simplicity constraints inserts the delta functions $\delta (\mathrm{Im}\theta_a)\delta (\mathrm{Im}\phi_a)$ in the above integral followed by restricting the range of $(\mathrm{Re}\theta_a,\mathrm{Re}\phi_a)$ to $\overline{\cM}_{\vec{j}}$. We denote the coherent state label satisfying the constraints by $\hrho_a$ and the corresponding coherent state by $\Psi^0_{\hrho_a}(\mu|m)$. Imposing the simplicity constraints reduces the above integral to 
\be
\int_{\overline{\cM}_{\vec{j}}}\rd\hrho_a\,\Psi^0_{\hrho_a}(\mu|m)\bar{\Psi}^0_{\hrho_a}(\mu'|m'):=\frac{1}{2Q^2}\int_{\overline{\cM}_{\vec{j}}}\rd\htheta_a\w\rd\hphi_a \,\Psi^0_{\hrho_a}(\mu|m)\bar{\Psi}^0_{\hrho_a}(\mu'|m') \,.
\ee
Since ${\overline{\cM}_{\vec{j}}}$ is compact, any integration on ${\overline{\cM}_{\vec{j}}}$ is finite as long as the integrand is bounded. This fact is important to guarantee the finiteness of the spinfoam amplitude defined below in Section \ref{sec:melon_graph}.

\subsubsection{The vertex amplitude: finiteness and semi-classical approximation}

With the second-class simplicity constraints imposed on the coherent state labels, one can define the vertex amplitude by the inner product of partition function \eqref{eq:Z_after_1st_constraints} and five coherent states \eqref{eq:coherent_for_amplitude}, each associated to one $\cS_a$. That is 
\be
\cA_v(\iota):=\la\prod\limits_{a=1}^5\bar{\Psi}^0_{\hrho_a}|\cZ_{S^3\backslash\Gamma_5}\ra
=\sum_{\{\mt_a\}\in (\Z/k\Z)^5}\int_{\R^5} \rd^5\tmu_a \,
\cZ_{S^3\backslash\Gamma_5}(\{i\alpha_{ab}\}_{(ab)}, \{\tmu_{a}+i\alpha_a\}\mid\{j_{ab}\}_{(ab)}, \{\mt_{a}\})
\prod\limits_{a=1}^5\Psi^0_{\hrho_a}(\tmu_a|\mt_a)\,,
\label{eq:vertex_amplitude}
\ee
where $\iota=(\{\alpha_{ab},j_{ab}\}_{(ab)}, \{\hrho_a\}_{a=1}^5, \{\alpha_a,\beta_a\}_{a=1}^5)$. 
Ref.\cite{Han:2021tzw} has proven that $\cA_v(\iota)$ is finite for given $\{\hat{\rho}_a\}_{a=1}^5$ with finite $\{\re(\hat{z}_a)\}_{a=1}^5$.

The large-$k$ approximation of the $\cA_v(\iota)$ reproduces the form as given in \cite{Haggard:2015yda,Haggard:2015nat},
\be 
\cA_v(\iota)\stackrel{k\rightarrow \infty}{\sim}\lb \cN_+ e^{iS^{\Lambda}_{\text{Regge}}+C} +  \cN_- e^{-iS^{\Lambda}_{\text{Regge}}-C} \rb(1+O(1/k))\,,
\label{eq:Regge_action}
\ee
where $\cN_\pm$ are factors related to the Hessian of the effective action when performing the saddle point analysis, $C$ is a geometric-independent integration constant and $S^{\Lambda}_{\text{Regge}}$ is the Regge action for a 4-simplex with constant curvature determined by the value of $\Lambda$. Explicitly,
\be
S^\Lambda_{\text{Regge }}=\frac{\Lambda k \gamma}{12 \pi}\left(\sum_{(ab)} \mathfrak{a}_{a b} \Theta_{a b}-\Lambda\left|V_4\right|\right)\,,
\label{eq:Regge_action_2}
\ee
where $\fa_{ab}$ is the area of the triangle $f_{ab}$ shared by tetrahedron $a$ and $b$ on the boundary of the 4-simplex, $\Theta_{ab}$ is the hyper-dihedral angle hinged by $f_{ab}$ and $|V_4|$ the volume of the 4-simplex. 

The finiteness of $\cA_v(\iota)$ and the appearance of the Regge action for a curved 4-simplex at the large-$k$ approximation \eqref{eq:Regge_action} renders the eligibility of the spinfoam model constructed with the vertex amplitude defined by \eqref{eq:vertex_amplitude}. 
By a valid choice of edge amplitude and face amplitude, one can define a finite amplitude for a general 4-manifold. Such a choice of edge and face amplitude was not given in the original paper \cite{Han:2021tzw}. We will give a proposal in the next section that is suitable for a simple spinfoam graph containing two spinfoam vertices and can be easily generalized to a general spinfoam graph.

\section{melon graph and spinfoam amplitude}
\label{sec:melon_graph}

We now consider the spinfoam amplitude corresponding to two 4-simplices with 4 boundary tetrahedra identified. In the dual picture, the spinfoam graph is called the ``melon graph'', which contains two spinfoam vertices, four internal spinfoam edges and two external spinfoam edges as shown in fig.\ref{fig:melon_graph}. It is the one-loop self-energy correction in the quantum field theory language. For the EPRL-FK model, it has been shown using GFT techniques that it is the most divergent part of the radiative correction of a spinfoam amplitude (at least compared to other simple enough spinfoam graphs, \eg a ``starfish graph'') \cite{Krajewski:2010yq,BenGeloun:2010qkf}. 

The way to define the spinfoam amplitude for the melon graph is similar to the way to define $\cA_v(\iota)$ reviewed in Section \ref{sec:review}. That is to first write the Chern-Simons partition function for the boundary of the manifold corresponding to the melon graph then impose the simplicity constraints (strongly for the first-class types and weakly for the second-class types). The first step is described in Section \ref{subsec:melon_partition} and the second step is sketched in Section \ref{subsec:separate_amplitudes}. The partition function for the melon graph can be separated into a pair of partition functions $\cZ_{S^3\backslash\Gamma_5}(\vec{\mu}|\vec{m})$'s for one spinfoam vertex defined in \eqref{eq:partition_S3G5} as well as some extra terms (which can be absorbed in the two vertex amplitudes), as explained in Section \ref{subsubsec:constraint_system}. The spinfoam amplitude for the melon graph is completed by adding a face amplitude for each internal spinfoam face. We write the full amplitude in \ref{subsec:full_amplitude} and prove its finiteness.

\subsection{Constraint system and the Chern-Simons partition function}
\label{subsec:melon_partition}

Denote the two three-manifolds $S^3\backslash \Gamma_5$'s as $M_+$ (containing 4-holed spheres $\cS_{1,2,3,4,5}$ on its boundary) and $M_-$ (containing 4-holed spheres $\cS_{1,2,3,4,6}$ on its boundary). They are glued through identifying $\cS_1,\cS_2,\cS_3,\cS_4$ on their boundaries and form the three-manifold $M_{+\cup -}$ whose spinfoam graph is a melon graph. See fig.\ref{fig:glue} for the GFT graph after gluing, where each blue line corresponds to an identification of holes from different spheres. 
After gluing, the connected holes become annuli or tori as boundaries of $M_{+\cup -}$. The blue lines can also be seen as the defects of an ambient 3-manifold of $M_{+\cup -}$ which possesses non-contractible cycles. 
\begin{figure}[h!]
\centering
\includegraphics[width=0.6\textwidth]{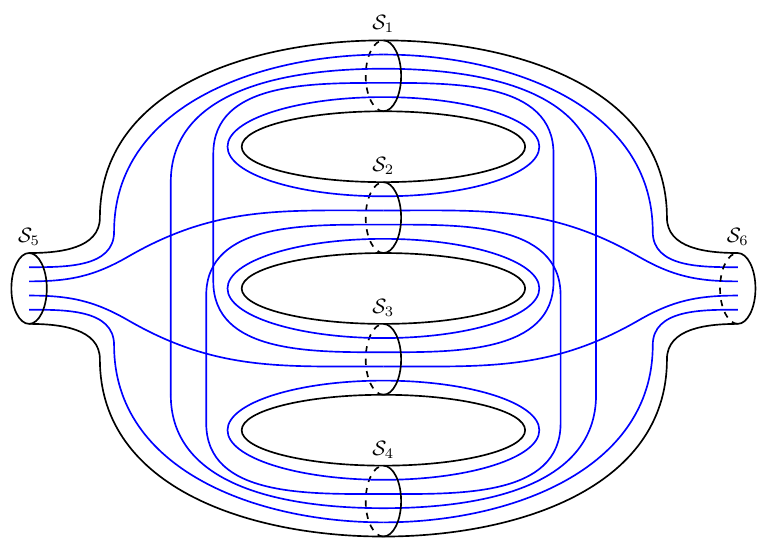}
\caption{The GFT graph denoting the manifold $M_{+\cup -}$ after gluing two spinfoam vertices corresponding to $M_+$ and $M_-$. 
The 4-holed spheres $\cS_1,\cS_2,\cS_3,\cS_4$ from $\partial M_+$ and $\partial M_-$ are identified. Each blue line relates to the identification of a pair of holes and becomes an annulus or a torus boundary of $M_{+\cup -}$. (There are no intersections among the blue lines.) $M_{+\cup -}$ is a graph complement of an ambient 3-manifold which has non-contractible cycles. } 
\label{fig:glue}
\end{figure}
The ideal triangulation of $M_{+\cup -}$ is obtained by the ideal triangulations of $M_+$ and $M_-$, which leads to 60 edges in total. On each edge, we assign an FG coordinate as we did in the previous section. To be consistent, we dress the edges on $M_+$ with FG coordinates in the same way as in the previous section (and as in \cite{Han:2021tzw}). $M_-$ and its ideal triangulation is simply given by the mirror of $M_+$ (see fig.\ref{fig:octahedra}). The (logarithmic) FG coordinates are listed in Table \ref{tab:edges}. 
Consequently, in Table \ref{tab:edges}, the relations for $M_-$ are translated from the ones for $M_+$ by changing each $i$ to $i+5$, where $i=1,\cdots,5$ labels the octahedra in $M_+$. 
\begin{figure}[h!]
\centering
\includegraphics[width=0.9\textwidth]{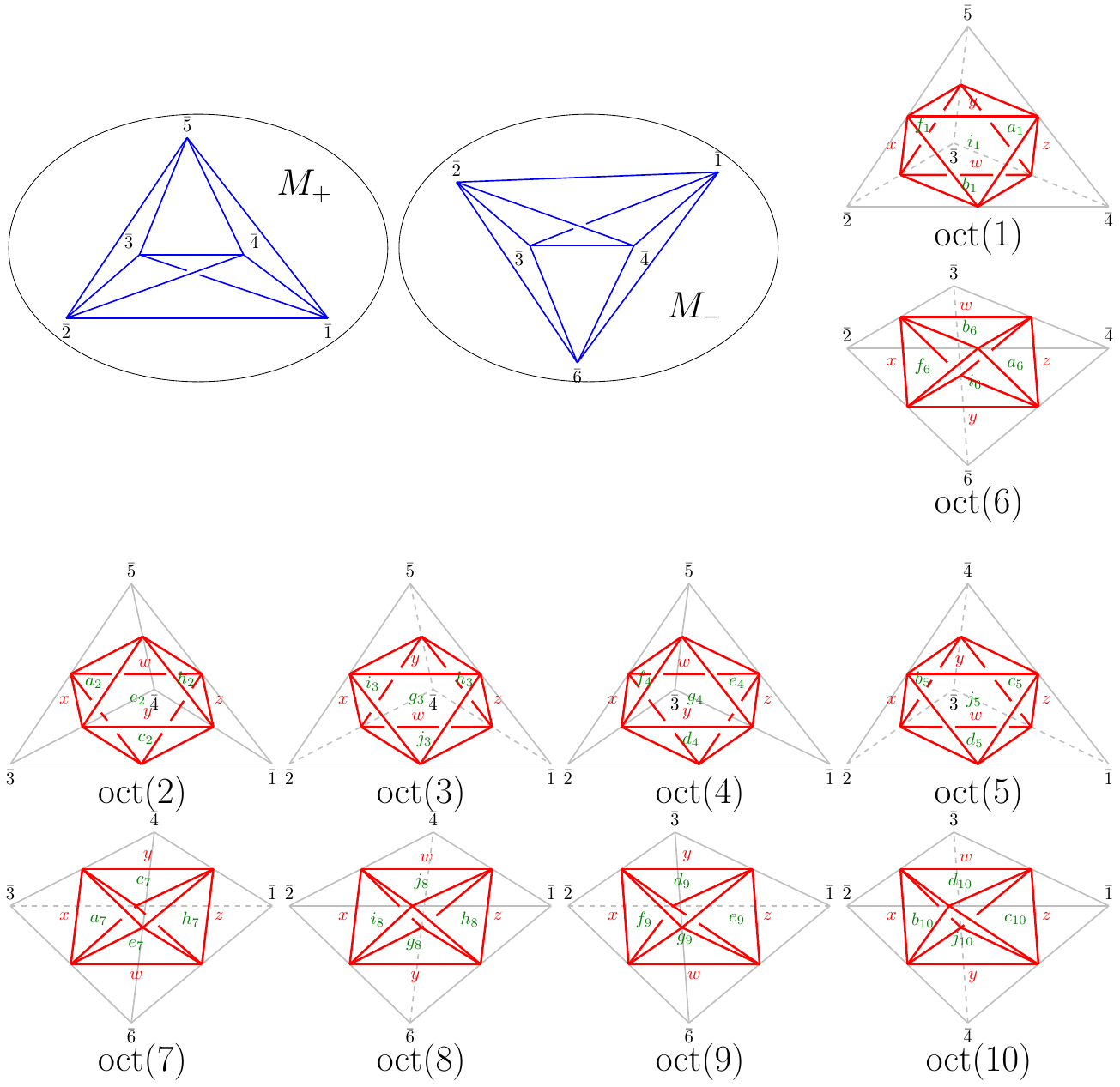}
\caption{The decomposition of $M_+$ and $M_-$ into ten octahedra ({\it in red}). The notations are the same as in fig.\ref{fig:triangulation_All}. The labels of faces ({\it in green}) in $\Oct(i+5)$ is the same as those in $\Oct(i)$ ($i=1,\cdots,5$) except for the subscripts. }
\label{fig:octahedra}
\end{figure}

\subsubsection{Gluing constraints and the Darboux coordinates}

When gluing $\cS_1,\cS_2,\cS_3,\cS_4$, we let the edges on the ideal triangulations of $M_+$ and $M_-$ dressed with the same FG coordinate be identified if they become internal edges in the gluing process. 
Indeed, if we parameterize all the edges on $M_+$ and $M_-$ in the same way, there is a twist between the Poisson brackets from the two three-manifolds due to the opposite orientations. Each edge $E$ of an ideal tetrahedron is dressed with an edge coordinate $z_E=e^{Z_E}$, as illustrated in fig.\ref{fig:ideal_tetra}. Let $Z^+_E$ be a (logarithmic) edge coordinate on one ideal tetrahedron $\triangle_+$ of $M_+$ and $Z^-_E$ be one on an ideal tetrahedron $\triangle_-$ of $M_-$. Then
\be
\{Z^+_E,Z^+_{E'}\}=\epsilon_{EE'}\delta_{\triangle_+,\triangle_+'}\,, \quad
\{Z^-_E,Z^-_{E'}\}=-\epsilon'_{EE'}\delta_{\triangle_-,\triangle_-'}\,, \quad 
\{Z^+_E,Z^-_{E'}\}=0\,,
\ee
where $\epsilon_{EE'}=0,\pm1$ counts the oriented triangles shared by $E,E'$ and $\epsilon_{EE'}=1$ if $E'$ occurs to the left of $E$ in the triangle, and $\delta_{\triangle_\pm,\triangle_\pm'}=1$ if $\triangle_\pm=\triangle_\pm'$ and $0$ otherwise.

Or equivalently, one can keep the Poisson brackets for $M_+$ and $M_-$ the same (as in \eqref{eq:ZZ'_Poisson}) but parameterize the edges differently for all ideal tetrahedra on $M_+$ and $M_-$, as shown in fig.\ref{fig:orientation}. 
This is the way we treat the two 3-manifolds in this paper. Such a parameterization has been used in fig.\ref{fig:octahedra} where edges with the same FG coordinate were glued. The algebraic curve for ideal tetrahedra on $M_+$ and $M_-$ written in terms of the edge coordinates on ideal octahedra (see fig.\ref{fig:ideal_octa}) are respectively
\be
\ba{ll}
\text{on }M_+:\quad&\fz^{-1}+\fz''-1=0\\
\text{on }M_-:\quad&\fz^{-1}+\fz'-1=0
\ea\,,\quad
\left|\ba{l}\fz=x,y,z,w\\\fz'=x',y',z',w'\\\fz''=x'',y'',z'',w''\ea\right..
\label{eq:A-polynomial}
\ee 
 It is easy to see that $\fZ'$ and $\fZ''$, which are the logarithms of $\fz'$ and $\fz''$ respectively, shift their roles on $M_+$ versus $M_-$. 
Therefore, we define the momenta on $M_+$ in terms of $\fZ''$ (see \eqref{eq:momentum}) while in terms of $\fZ'$ on $ M_-$. That is
\begin{subequations}
\begin{align}
&P_{X_i}=X_i''-W_i''\,,\quad
 P_{Y_i}=Y_i''-W_i''\,,\quad
 P_{Z_i}=Z_i''-W_i''\,,\quad
 \Gamma_i=W_i''\,,\quad 
\text{ for }\, i=1,\cdots,5\,,\\
&P_{X_j}=X_j'-W_j'\,,\quad
 P_{Y_j}=Y_j'-W_j'\,,\quad
 P_{Z_j}=Z_j'-W_j'\,,\quad
 \Gamma_j=W_j'\,,\quad 
\text{ for }\, j=6,\cdots,10\,.
\label{eq:momenta_ML_MR}
\end{align}
\end{subequations}
They are momenta conjugate to $X_i,\ Y_i,\ Z_i,\ C_i=X_i+Y_i+Z_i+W_i$, $i=1,\cdots,10$, respectively and satisfy
\be
\{X_i,P_{X_j}\}=\{Y_i,P_{Y_j}\}=\{Z_i,P_{Z_j}\}= \{C_i,\Gamma_j\}=\delta_{ij}\,.
\ee

\begin{figure}[h!]
\centering
\begin{minipage}{0.3\textwidth}
\centering
\includegraphics{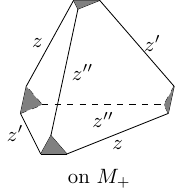}
\end{minipage}
\begin{minipage}{0.3\textwidth}
\centering
\includegraphics{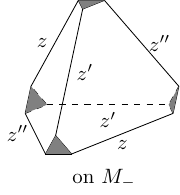}
\end{minipage}
\caption{Different parameterizations of edges of ideal tetrahedra on $M_+$ and $M_-$. }
\label{fig:orientation}
\end{figure}

The procedure of gluing triangulated 3-manifolds is a generalization of the treatment for an ideal octahedron: Every internal edge of the ideal triangulation corresponds to a (classical) constraint restricting the sum of the involved (logarithmic) FG coordinates at the edge to $2\pi i$ \cite{Dimofte:2011gm}. The gluing of $M_+$ and $M_-$ leads to 24 constraints $\{\cC_A^\chi\}_{A=1}^{24}$ in terms of the FG coordinates $\{\chi_{ij}^{(a)}\}_{a,i\neq j}$. Here, $\chi_{ij}^{(a)}$ dresses the edge of the ideal triangulation of $\cS_a$ that is shared by octahedra $\Oct(i)$ and $\Oct(j)$ (see Table \ref{tab:edges} in Appendix \ref{app:FG_FN_coordinates} for their explicit expressions in terms of the edge coordinates). They are explicitly
\be\ba{llll}
\text{on }\cS_1:  & \cC^\chi_1=\chi^{(1)}_{45} + \chi^{(1)}_{9,10} -2\pi i \,, 
		& \cC^\chi_2=\chi^{(1)}_{35} + \chi^{(1)}_{8,10} -2\pi i \,, 
		& \cC^\chi_3=\chi^{(1)}_{34} + \chi^{(1)}_{89} -2\pi i \,,\\[0.15cm]
		& \cC^\chi_4=\chi^{(1)}_{25} + \chi^{(1)}_{7,10} -2\pi i \,, 
		& \cC^\chi_5=\chi^{(1)}_{24} + \chi^{(1)}_{79} -2\pi i \,, 
		& \cC^\chi_6=\chi^{(1)}_{23} + \chi^{(1)}_{78} -2\pi i \,,\\[0.25cm]
\text{on }\cS_2:  & \cC^\chi_7=\chi^{(2)}_{45} + \chi^{(2)}_{9,10} -2\pi i \,, 
		& \cC^\chi_8=\chi^{(2)}_{35} + \chi^{(2)}_{8,10} -2\pi i \,, 
		& \cC^\chi_9=\chi^{(2)}_{34} + \chi^{(2)}_{89} -2\pi i \,,\\[0.15cm]
		& \cC^\chi_{10}=\chi^{(2)}_{15} + \chi^{(2)}_{6,10} -2\pi i \,, 
		& \cC^\chi_{11}=\chi^{(2)}_{14} + \chi^{(2)}_{69} -2\pi i \,, 
		& \cC^\chi_{12}=\chi^{(2)}_{13} + \chi^{(2)}_{68} -2\pi i \,,\\[0.25cm]
\text{on }\cS_3:  & \cC^\chi_{13}=\chi^{(3)}_{45} + \chi^{(3)}_{9,10} -2\pi i \,, 
		& \cC^\chi_{14}=\chi^{(3)}_{25} + \chi^{(3)}_{7,10} -2\pi i \,, 
		& \cC^\chi_{15}=\chi^{(3)}_{24} + \chi^{(3)}_{79} -2\pi i \,,\\[0.15cm]
		& \cC^\chi_{16}=\chi^{(3)}_{15} + \chi^{(3)}_{6,10} -2\pi i \,, 
		& \cC^\chi_{17}=\chi^{(3)}_{14} + \chi^{(3)}_{69} -2\pi i \,, 
		& \cC^\chi_{18}=\chi^{(3)}_{12} + \chi^{(3)}_{67} -2\pi i \,,\\[0.25cm]
\text{on }\cS_4:  & \cC^\chi_{19}=\chi^{(4)}_{35} + \chi^{(4)}_{8,10} -2\pi i \,, 
		& \cC^\chi_{20}=\chi^{(4)}_{25} + \chi^{(4)}_{7,10} -2\pi i \,, 
		& \cC^\chi_{21}=\chi^{(4)}_{23} + \chi^{(4)}_{78} -2\pi i \,,\\[0.15cm]
		& \cC^\chi_{22}=\chi^{(4)}_{15} + \chi^{(4)}_{6,10} -2\pi i \,, 
		& \cC^\chi_{23}=\chi^{(4)}_{13} + \chi^{(4)}_{68} -2\pi i \,, 
		& \cC^\chi_{24}=\chi^{(4)}_{12} + \chi^{(4)}_{67} -2\pi i \,.
\ea
\label{eq:constaint_chi}
\ee  
It is easy to check that there are only 18 independent first-class constraints out of these 24 constraints. The dimension of the Chern-Simons phase space $\cP_{\partial M_{+\cup -}}$ on the boundary of $M_{+\cup -}$, as the moduli space of framed flat $\PSL(2,\bC)$ connection on $\partial M_{+\cup -}$, is $60-2\times18=24$. To single out the first-class constraints, using the FN coordinates as the Darboux coordinates of $\cP_{\partial M_{+\cup -}}$ is more convenient. 

\medskip

We denote the FN coordinates on $\cP_{\partial M_+}$ (\resp $\cP_{\partial M_-}$) as $\{L_{ab}\}$ (\resp $\{L'_{ab}\}$) where $a,b$ denote the 4-holed spheres $\cS_a$ and $\cS_b$. They are indeed the linear combinations of the FG coordinates $\{\chi_{ij}^{(a)}\}_{i,j=1,i\neq j}^5$ (\resp $\{\chi_{ij}^{(a)}\}_{i,j=6,i\neq j}^{10}$). $\{L_{ab}\}$ are defined in the same way as in Section \ref{sec:review}. The definition of each $L'_{ab}$ is copied from that of $L_{ab}$ followed by shifting all the octahedron labels therein by 5, \ie $i\rightarrow i+5,j\rightarrow j+5$ (see \eqref{eq:all_FN_coordinates}).

The FN coordinates can be naturally understood as assigned on the annuli connecting holes from different spheres as they, by definition, satisfy the relations 
\be
L_{ab}=L_{ba}\,,\quad
L'_{ab}=L'_{ba}\,,\quad
\forall a,b\,.
\ee

The gluing constraints \eqref{eq:constaint_chi} can be then partially written in terms of these FN coordinates: 
\be\ba{llll}
  \cC_1=2L_{12}+2L'_{12}=0\,,
& \cC_2=2L_{13}+2L'_{13}=0\,,
& \cC_3=2L_{14}+2L'_{14}=0\,,
& \cC_4=2L_{15}+2L'_{16}=0\,,\\[0.25cm]
  \cC_5=2L_{23}+2L'_{23}=0\,,
& \cC_6=2L_{24}+2L'_{24}=0\,,
& \cC_7=2L_{25}+2L'_{26}=0\,,
& \cC_8=2L_{34}+2L'_{34}=0\,,\\[0.25cm]
  \cC_9=2L_{35}+2L'_{36}=0\,, 
& \cC_{10}=2L_{45}+2L'_{46}=0\,.
\ea
\label{eq:constraint_FN}
\ee 
Denote the Darboux coordinates on $\partial M_+$ as $(\cQ^+_i,\cP^+_i)_{i=1,\cdots,15}$ and those on $\partial M_-$ as $(\cQ^-_i,\cP^-_i)_{i=1,\cdots,15}$ where (we denote coordinates from $M_-$ with prime) 
\begin{subequations}
\begin{align}
&\cQ^+_i=\{\{2L_{ab}\}_{(ab)},\{\cX_a\}_{a=1}^5\}\,,\quad
\cP^+_i=\{\{\cT_{ab}\}_{(ab)},\{\cY_a\}_{a=1}^5\}\,,\\
&\cQ^-_i=\{\{2L'_{ab}\}_{(ab)},\{\cX'_a\}_{a=1}^5\}\,,\quad
\cP^-_i=\{\{\cT'_{ab}\}_{(ab)},\{\cY'_a\}_{a=1}^5\}
\end{align}
\label{eq:QP_epsilon_def}
\end{subequations}
 with the Poisson brackets
\be
\{\cQ^+_i,\cP^+_j\}=\{\cQ^-_i,\cP^-_j\}=\delta_{ij}\,,\quad
\{\cQ^+_i,\cQ^-_j\}=\{\cP^+_i,\cP^-_j\}=\{\cQ^+_i,\cP^-_j\}=\{\cQ^-_i,\cP^+_j\}=0\,,\quad
\forall i,j=1,\cdots,15\,.
\label{eq:Poisson-L_-}
\ee

The Darboux coordinates for the $M_{+ \cup -}$ are thus $({}^0\cQ_I,{}^0\cP_I)_{I=1,\cdots,30}$ with $^0\cQ_I=\{\cQ_i^+,\cQ_i^-\}$ and $^0\cP_I=\{\cP_i^+,\cP_i^-\}$, which span a 60-dimensional phase space. 
The explicit choices for $(\cQ^+_i,\cP^+_i)$ and $(\cQ^-_i,\cP^-_i)$ are as follows.
\begin{subequations}
\begin{align}
L_{ab} &=\{L_{12},L_{13},L_{14},L_{15},L_{23},L_{24},L_{25},L_{34},L_{35},L_{45}\}\,,\\
L'_{ab}&=\{L'_{12},L'_{13},L'_{14},L'_{16},L'_{23},L'_{24},L'_{26},L'_{34},L'_{36},L'_{46}\}\,,\\
\cX_a  &=\{\chi^{(1)}_{25},\chi^{(2)}_{15},\chi^{(3)}_{15},\chi^{(4)}_{15},\chi^{(5)}_{14}\}\,,\\
\cX'_a &=\{\chi^{(1)}_{7,10},\chi^{(2)}_{6,10},\chi^{(3)}_{6,10},\chi^{(4)}_{6,10},\chi^{(6)}_{69}\}\,,\\
\cT_{ab} &=\{\cT_{12},\cT_{13},\cT_{14},\cT_{15},\cT_{23},\cT_{24},\cT_{25},\cT_{34},\cT_{35},\cT_{45}\}\,,\\
\cT'_{ab}&=\{\cT'_{12},\cT'_{13},\cT'_{14},\cT'_{16},\cT'_{23},\cT'_{24},\cT'_{26},\cT'_{34},\cT'_{36},\cT'_{46}\}\,,\\
\cY_a  &=\{\chi^{(1)}_{23},\chi^{(2)}_{14},\chi^{(3)}_{45}-2\pi i,-\chi^{(4)}_{35}+2\pi i,\chi^{(5)}_{34}-2\pi i\}\,,\\
\cY'_a &=\{-\chi^{(1)}_{78},-\chi^{(2)}_{69},-\chi^{(3)}_{9,10}+2\pi i,\chi^{(4)}_{8,10}-2\pi i,-\chi^{(6)}_{89}+2\pi i\}\,.
\end{align}
\end{subequations}
Following \eqref{eq:Poisson-L_-}, the Darboux coordinates on $\partial M_{+ \cup -}$ also satisfy the desired Poisson brackets:
\be
\{{}^0\cP_I,{}^0\cQ_J\}=\delta_{IJ}\,,\quad
\{{}^0\cP_I,{}^0\cP_J\}=\{{}^0\cQ_I,{}^0\cQ_J\}=0\,,\quad
\forall I,J=1,\cdots,30\,.
\label{eq:Poisson-LUR}
\ee

Apart from the 10 constraints \eqref{eq:constraint_FN}, we need to define the remaining 8 independent first-class constraints. We choose them to be 
\be\ba{llll} 
 \cC_{11}=\cX_1+\cX'_1 -2\pi i \equiv \cC_4^\chi \,,
&\cC_{12}=\cX_2+\cX'_2 -2\pi i \equiv \cC_{10}^\chi \,,
&\cC_{13}=\cX_3+\cX'_3 -2\pi i \equiv \cC_{16}^\chi \,,
&\cC_{14}=\cX_4+\cX'_4 -2\pi i \equiv \cC_{22}^\chi \,,\\[0.25cm]
 \cC_{15}=\cY_1-\cY'_1 -2\pi i \equiv \cC_6^\chi \,, 
&\cC_{16}=\cY_2-\cY'_2 -2\pi i \equiv \cC_{11}^\chi \,, 
&\cC_{17}=\cY_3-\cY'_3 +2\pi i \equiv \cC_{13}^\chi\,, 
&\cC_{18}=\cY_4-\cY'_4 -2\pi i \equiv -\cC_{19}^\chi\,. 
\ea
\label{eq:constraint_XY}
\ee
The 18 constraints \eqref{eq:constraint_FN} and \eqref{eq:constraint_XY} are all independent and can be verified to be first-class, $\ie$
\be
\{\cC_A,\cC_B\}=0\quad \forall A,B=1,\cdots,18\,.
\ee

The relation between $\{\cC_A\}_{A=1}^{18}$ and the original constraints $\{\cC^\chi_K\}_{K=1}^{24}$ can be understood in the following way. If we add 6 more constraints such that $\cC_{A=19,\cdots,24}=2L_{ba}+2L_{ba}'$ with $a<b$ and $a,b=1,
\cdots,4$, there is a nondegenerate linear transformation relating $\{\cC_A\}$ and $\{\cC_K^\chi\}$. The redundancy in the set of constraints $\{\cC_K^\chi\}$ is reflected by the fact that $\{\cC_A\}_{A=19}^{24}$ are not independent of the rest of $\{\cC_A\}$, since $L_{ab}=L_{ba}$ and $L'_{ab}=L'_{ba}$ by definition. This is related to the topology of $M_{+\cup -}$. One can find one linear relation among the constraints for every torus cusp (depicted by a closed blue loop in fig.\ref{fig:glue}) and there are in total 6 of them which remove 6 constraints from $\{\cC_K^\chi\}_{K=1}^{24}$. The generality of this topological relation is argued in \cite{Dimofte:2013iv, Dimofte:2013lba}.

The reduced phase space $\cP_{\partial M_{+\cup -}}$, therefore, is the symplectic quotient of the tensor product of phase spaces from $\partial M_+$ and $\partial M_-$ by the gluing constraints, \ie $\cP_{\partial M_{+\cup -}}=\lb\cP_{\partial M_+}\otimes \cP_{\partial M_-}\rb//\{\cC_A\}_{A=1}^{18}$. 

\subsubsection{Symplectic transformation and the partition function}
\label{subsubsec:constraint_system}

After imposing the 18 first-class constraints \eqref{eq:constraint_FN} and \eqref{eq:constraint_XY}, one is left with a 24-dimensional phase space with 12 positions and 12 momenta variables. We perform a series of symplectic transformations from $(^0\cQ_I,^0\cP_I)_{I=1,\cdots,30}$ to $(\fQ_J,\fP_J)_{J=1,\cdots,30}$ parameterized as
\be
\fQ_J=\{\{2L_{ab}^{+\cup -}\}_{(ab)},\cX_5,\cX'_5,\{\cC_A\}_{A=1}^{18}\}\,,\quad
\fP_J=\{\{\cT_{ab}^{+\cup -}\}_{(ab)},\cY_5,\cY'_5,
\{\Gamma_A\}_{A=1}^{18}\}\,.
\label{eq:final_QP}
\ee 
We choose the first 10 position variables to be 
\be
2L_{ab}^{+\cup -}=\{2L_{12},2L_{13},2L_{14},2L_{15},2L_{23},2L_{24},2L_{25},2L_{34},2L_{35},2L_{45}\}\,,
\label{eq:final_Q}
\ee 
then
\be
\cT_{ab}^{+\cup -}=\{\cT_{12}-\cT'_{12},\cT_{13}-\cT'_{13},\cT_{14}-\cT'_{14},\cT_{15}-\cT'_{16},\cT_{23}-\cT'_{23},\cT_{24}-\cT'_{24},\cT_{25}-\cT'_{26},\cT_{34}-\cT'_{34},\cT_{35}-\cT'_{36},\cT_{45}-\cT'_{46}\}\,.
\label{eq:final_T}
\ee 
The explicit expressions of $\{\Gamma_A\}_{A=1}^{18}$ will calculated by the symplectic matrices (see \eqref{eq:Gamma_A}).
The transformations from $(^0\cQ_I,^0\cP_I)$ to $(\fQ_J,\fP_J)$ contains one $U$-type transformation, one {\it partial} $S$-type transformation and one affine translation illustrated as follows. 

\begin{enumerate}
\setcounter{enumi}{-1}
\item The starting point is the product of the partition functions for $M_+$ and $M_-$:
\be
\cZ_\times(\vec{\mu}|\vec{m})=\cZ_{M_+}(\vec{\mu}_+|\vec{m}_+)\cZ_{M_-}(\vec{\mu}_-|\vec{m}_-)\,,
\ee
where $\vec{\mu}=\{\vec{\mu}_+,\vec{\mu}_-\}\,,\vec{m}=\{\vec{m}_+,\vec{m}_-\}$ and the two distinct partition functions are defined in terms of {\it different} symplectic transformations
\be
\cZ_{M_\pm}(\vec{\mu}_\pm|\vec{m}_\pm)
=((\sigma_{\vec{t}_\pm}\circ S\circ T_\pm \circ U_\pm )\triangleright \cZ_\times)(\vec{\mu}_\pm|\vec{m}_\pm)\,.
\ee
The symplectic transformations are encoded in the transformation matrices $\mat{cc}{\bA_\pm & \bB_\pm \\ -(\bB_\pm^\top)^{-1} & 0}$ and the affine translation vector $\vec{t}_\pm$. 
We use the expression \eqref{eq:partition_S3G5} for both partition functions $\cZ_{M_+}$ and $\cZ_{M_-}$ since $\bA_\pm\bB^\top_\pm$ are both symmetric matrices with integer entries \footnote{There are the same number of odd elements in $\diag(\bA_\pm\bB_\pm^\top)$ and these elements are at the same locations, \ie the $1$th, $2$nd, $6$th, $8$th, $11$th, $12$th and $13$th elements, as can be checked from the explicit expressions \eqref{eq:ABt+} and \eqref{eq:ABt-} of $\bA_\pm$ and $\bB_\pm$ respectively.}.
Note that the ``$+$'' sector is the same as in Section \ref{sec:review} while the ``$-$'' sector is {\it not} due to a different choice \eqref{eq:momenta_ML_MR} of momentum variables in $\cP_{\partial\oct}$ for $\partial M_-$. See Appendix \ref{app:symplectic_tranf}
for the explicit expressions for $\bA_\pm$, $\bB_\pm$ and $\vec{t}_\pm$. 
We denote the positive angle structures for $\cZ_{M_\pm}(\vec{\mu}_\pm|\vec{m}_\pm)$ by $\fP_{M_\pm}$, then $\cZ_\times\in\cF_{\fP_{M_+}\otimes\fP_{M_-}}^{(k)}$. 

\item First, we perform a $U$-type transformation:  
\be
\mat{c}{{}^1\cQ_I\\\hline{}^1\cP_I}
=\mat{c|c}{\bfU&{\bf 0}\\\hline {\bf 0} & (\bfU^{\top})^{-1}}
\mat{c}{{}^0\cQ_I\\\hline{}^0\cP_I}
\,,\quad\text{where }
\bfU=\mat{c|c|c|c|c|c}
{{\bf 1}_{10\times 10}&{\bf 0} & {\bf 0} & {\bf 0} & {\bf 0}& {\bf 0}\\
\hline
{\bf 0} & {\bf 0} & 1 & {\bf 0} & {\bf 0} & 0\\
\hline
{\bf 0} & {\bf 0} & 0 & {\bf 0} & {\bf 0} & 1\\
\hline
{\bf 1}_{10\times 10}&{\bf 0} & {\bf 0} & {\bf 1}_{10\times 10} & {\bf 0} & {\bf 0}\\
\hline
{\bf 0} &{\bf 1}_{4\times 4} & {\bf 0} & {\bf 0} & {\bf 1}_{4\times 4} & {\bf 0}\\
\hline
{\bf 0} &{\bf 1}_{4\times 4} & {\bf 0} & {\bf 0} & {\bf 0}& {\bf 0} 
} \,.
\ee
The new positions $\{^1\cQ_I\}$ and momenta $\{^1\cP_I\}$ read
\be\ba{lccccccl}
^1\cQ_I=\{
&2L_{ab},&\cX_5,&\cX'_5,&2L_{ab}+2L'_{ab},&\{\cX_a+\cX'_a\}_{a=1}^4,&\{\cX_a\}_{a=1}^4&\}\,,\\ [0.15cm]
^1\cP_I=\{
&\cT_{ab}-\cT'_{ab},&\cY_5,&\cY'_5,&\cT'_{ab},&\{\cY'_a\}_{a=1}^4,&\{\cY_a-\cY'_a\}_{a=1}^4&\}\,.
\ea\nn\ee
Since $\det(\bfU)=1$, the amplitude is transformed to 
\be
\cZ_1(\vec{\mu}|\vec{m})=(\bfU\triangleright \cZ_\times)(\vec{\mu}|\vec{m})
=\cZ_\times (\bfU^{-1}\vec{\mu}|\bfU^{-1}\vec{m})\,.
\ee
When $(\vec{\alpha},\vec{\beta})\in\fP_{M_+}\otimes \fP_{M_-}$, $(\bfU\vec{\alpha},(\bfU^{-1})^\top\vec{\beta})\in \fP_1$ with $\fP_1=\bfU\circ \lb\fP_{M_+}\otimes \fP_{M_-}\rb$, and 
\be
e^{-\f{2\pi}{k}\vec{\mu}\cdot(\bfU^{-1})^\top\vec{\beta}}\cZ_\times(\bfU^{-1}(\vec{\mu}+i\bfU\vec{\alpha})|\vec{m})
\ee
is a Schwartz function. Therefore, $Z_1\in \cF^{(k)}_{\fP_1}$.

\item We then perform a partial $S$-type transformation on the last 4 positions of $^1\cQ_I$ and the last 4 momenta of $^1\cP_I$. That is,
\be
\mat{c}{{}^2\cQ_I\\\hline{}^2\cP_I}=\bfS\mat{c}{{}^1\cQ_I\\\hline{}^1\cP_I} \,,\quad
\text{where }
\bfS=\mat{c|c|c|c}{
{\bf 1}_{26\times 26} & {\bf 0} & {\bf 0} & {\bf 0} \\
\hline
{\bf 0} &  {\bf 0} &  {\bf 0} &  {\bf 1}_{4\times 4} \\
\hline
{\bf 0} & {\bf 0} & {\bf 1}_{26\times 26}  & {\bf 0} \\
\hline
{\bf 0} & -{\bf 1}_{4\times 4} & {\bf 0} & {\bf 0} 
}\,.
\ee
The new coordinates after this transformation are
\be\ba{lccccccl}
^2\cQ_I=\{
&2L_{ab},&\cX_5,&\cX'_5,&2L_{ab}+2L'_{ab},&\{\cX_a+\cX'_a\}_{a=1}^{4},&\{\cY_a-\cY'_a\}_{a=1}^{4}&\}\,,\\ [0.15cm]
^2\cP_I=\{
&\cT_{ab}-\cT'_{ab},&\cY_5,&\cY'_5,&\cT_{ab}',&\{\cY'_a\}_{a=1}^{4},&\{-\cX_a\}_{a=1}^{4}&\}\,.
\ea\ee
This partial $S$-type transformation corresponds to a Fourier transform on the amplitude to change the coordinates corresponding to the last four constraints $\{\cC_A\}_{A=15}^{18}$ while keeping the rest of the coordinates unchanged. Explicitly,
\be
\cZ_2(\vec{\mu}|\vec{m})=(\bfS\triangleright \cZ_1)(\vec{\mu}|\vec{m})
=\f{1}{k^4}\sum_{\vec{n}\in(\Z/k\Z)^{30}} \int_{\cC^{\times30}}\rd^{30}\vec{\nu}\,
\lb\prod_{I=1}^{26}\delta_{\mu_I,\nu_I}\delta_{m_I,n_I}\rb
e^{\f{2\pi i}{k}\sum_{J=27}^{30}(-\mu_J\nu_J+m_Jn_J)}\cZ_1(\vec{\nu}|\vec{n})\,.
\ee
Define $(\vec{\alpha}',\vec{\beta}')$ such that 
\be
\begin{cases}
	\alpha'_I=-\beta_I\,,\quad \beta'_I=\alpha_I \quad & \forall I=1,\cdots,26\\
	\alpha'_I=\alpha_I\,,\quad\beta'_I=\beta_I\quad &\forall I=27,\cdots,30
\end{cases}\,,
\label{eq:alpha_beta_prime}
\ee
and set $\vec{\alpha}=\im(\vec{\mu}),\vec{\beta}=\im(\vec{\nu})$.
Then when $(-\vec{\beta},\vec{\alpha})\in\fP_1$, or equivalently $(\vec{\alpha}',\vec{\beta}')\in \fP_2:=\bfS\circ \bfU\circ\lb\fP_{M_+}\otimes \fP_{M_-}\rb$,
\be
e^{-\f{2\pi i}{k}\sum_{I=27}^{30}\mu_I\nu_I}\cZ_1(\{\mu_I\}_{I=1}^{26},\{\nu_I\}_{I=27}^{30}|\{m_I\}_{I=1}^{26},\{n_I\}_{I=27}^{30})
\ee
is a Schwartz function in $\{\re(\nu_I)\}_{I=27}^{30}$. When $\{\im(\mu_{I})=\alpha_I\}_{I=27}^{30}$ and the integration contour $\cC^{\times30}$ is defined such that $\{\im(\nu_I)=\beta_I\}_{I=27}^{30}$, $\cZ_2$ converges absolutely hence $\cZ_2\in \cF_{\fP_2}^{(k)}$.

\item Finally, we perform an affine shift $\sigma_{\vec{t}}$ to arrive at the final coordinates $(\fQ_I,\fP_I)$ defined in \eqref{eq:final_QP}--\eqref{eq:final_T}. The symplectic transformation is
\be
\mat{c}{\fQ_I\\\hline \fP_I} 
=  \mat{c}{{}^2\cQ_I +i\pi \vec{t}\\\hline{}^2\cP_I}
\,,\label{QIPI}
\ee
where the vector $\vec{t}$ of length 30 is composed with integer elements. According to the constraints definitions \eqref{eq:constraint_FN} and \eqref{eq:constraint_XY}, there are only 8 non-zero elements in $\vec{t}$, which are (note a different sign in $t_{29}$)
\be
	t_{23}=t_{24}=t_{25}=t_{26}=t_{27}=t_{28}=-t_{29}=t_{30}=-2
	\,.\nn
\ee
Therefore, one can write down the conjugate variables $\{\Gamma_A\}$ of the constraints $\{\cC_A\}$ \eqref{eq:constraint_FN} and \eqref{eq:constraint_XY} in terms of $(^0\cQ_I,^0\cP_I)$: 
\be\ba{lllll}
\Gamma_1 = \cT'_{12}\,, 
&\Gamma_2 = \cT'_{13}\,,
&\Gamma_3 = \cT'_{14}\,,
&\Gamma_4 = \cT'_{15}\,,
&\Gamma_5 = \cT'_{23}\,,\\
\Gamma_6 = \cT'_{24}\,,
&\Gamma_7 = \cT'_{26}\,,
&\Gamma_8 = \cT'_{34}\,,
&\Gamma_9 = \cT'_{36}\,,
&\Gamma_{10} = \cT'_{46}\,,\\
\Gamma_{11} = \cY'_1\,,
&\Gamma_{12} = \cY'_2\,,
&\Gamma_{13} = \cY'_3\,,
&\Gamma_{14} = \cY'_4\,,\\
\Gamma_{15} = -\cX_1\,,
&\Gamma_{16} = -\cX_2\,,
&\Gamma_{17} = -\cX_3\,,
&\Gamma_{18} = -\cX_4\,.
\ea
\label{eq:Gamma_A}
\ee
Define a translation map to the positive angle variables
\be
\sigma'_{\vec{t}}:\fP_2\rightarrow \fP_{\text{new}}\,,\quad
\mat{c}{\vec{\alpha}\\\vec{\beta}}\mapsto 
\mat{c}{\vec{\alpha}+\f{Q}{2}\vec{t}\\\vec{\beta}}\,.
\nn\ee
The final positive angle structure is
\be
\fP_{\text{new}}=\sigma'_{\vec{t}}\circ \bfS\circ \bfU \circ(\fP_{M_+}\otimes \fP_{M_-})\,.
\nn\ee
\end{enumerate}
In order to write the final amplitude in a simple way, we pick out parts of the non-zero elements in $\vec{t}$ and define a length-30 vector $\vec{t}'$ whose only non-zero elements are $t'_{23}=t'_{24}=t'_{25}=t'_{26}=-2$.
The resulting amplitude is written as
\be\begin{split}
\cZ^0_{M_{+\cup -}}(\vec{\mu}|\vec{m})
&= (( \sigma_{\vec{t}}\circ \bfS \circ \bfU)\triangleright \cZ_\times)(\vec{\mu}|\vec{m})\\
&=\f{1}{k^4}\sum_{\vec{n}\in(\Z/k\Z)^{30}} \int_{\cC^{\times 30}}{\rd^{30}\vec{\nu}}\,
\lb\prod_{I=1}^{26}\delta_{\mu_I,\nu_I}\delta_{m_I,n_I}\rb
e^{\f{2\pi i}{k}\sum\limits_{J=27}^{30}(-\mu_J+i\f{Q}{2}t_J)\nu_J+m_Jn_J}\cZ_{\times}(\bfU^{-1}\vec{\nu}-\f{iQ}{2}\vec{t}'\mid\bfU^{-1}\vec{n})\,.
\label{eq:final_amplitude}
\end{split}\ee

Let us also write out the positive angle structure. If $(\vec{\alpha},\vec{\beta})\in \fP_{M_+}\otimes \fP_{M_-}$, then 
\be
\lb\vec{\alpha}_{\text{new}},\vec{\beta}_{\text{new}}\rb:=\lb\lb(\bfU^{-1})^\top \vec{\beta}\rb'+\f{Q}{2}\vec{t},-\lb\bfU\vec{\alpha}\rb' \rb\in \fP_{\text{new}}\,,
\label{eq:positive_angle_new12}
\ee
where the prime variables are defined in the same way as in \eqref{eq:alpha_beta_prime}. 
Therefore, when $(\vec{\alpha}_{\text{new}},\vec{\beta}_{\text{new}})\in \fP_{\text{new}}$, $Z^0_{M_{+\cup-}}\in \cF^{(k)}_{\fP_{\text{new}}}$.

The 18 constraints $\{\cC_A\}$ require that the corresponding elements in $\vec{\alpha}_{\text{new}}$ are zero. These requirements impose further constraints on the initial positive angle structures for ideal octahedra on top of \eqref{eq:positive_angle_oct}. One can show that the positive angle is still {\it non-empty} through examples, some of which are collected in Appendix \ref{app:positive_angle}. 

Let us now fix at once the notations of the parameterizations for the new symplectic coordinates $(\vec{\fQ},\vec{\fP})$. Label the constraints $\{\cC_{A}\}_{A=1}^{10}$ for the FN coordinates \eqref{eq:constraint_FN} by $\cC_{ab}:=2L_{ab}+2L'_{ab}$, and the constraints $\{\cC_{A}\}_{A=11}^{18}$ for the FG coordinates \eqref{eq:constraint_XY} by $\cC_{\cX_a}:=\cX_a+\cX'_a-2\pi i$ and $\cC_{\cY_a}:=\cY_a-\cY'_a-2\pi is_a\,,a=1\cdots,4$ where $\vec{s}=\{1,1,-1,1\}$ is a vector of signs. We parameterize
\begin{subequations}\label{conventionmunu000}
\begin{align}
&2L_{ab}=\f{2\pi i}{k}(-ib\mu_{ab}-m_{ab})\,,\quad 
2L'_{ab}=\f{2\pi i}{k}(-ib\mu'_{ab}-m'_{ab})\,,\quad\\
&\cT_{ab}=\f{2\pi i}{k}(-ib\nu_{ab}-n_{ab})\,,\quad 
\cT'_{ab}=\f{2\pi i}{k}(-ib\nu'_{ab}-n'_{ab})\,,\quad\\
&\cX_a=\f{2\pi i}{k}(-ib\mu_{a}-m_{a})\,,\quad
\cX'_a=\f{2\pi i}{k}(-ib\mu'_{a}-m'_{a})\,,\quad a=1,\cdots,5\,,
\label{eq:cX_param}\\
&\cY_a=\f{2\pi i}{k}(-ib\nu_{a}-n_{a})\,,\quad
\cY'_a=\f{2\pi i}{k}(-ib\nu'_{a}-n'_{a})\,,\quad a=1,\cdots,5\,,
\label{eq:cY_param}\\
&\cC_{ab}=\f{2\pi i}{k}(-ib\mu_{\cC_{ab}}-m_{\cC_{ab}})\,,\quad
\cC_{\cX_{a}}=\f{2\pi i}{k}(-ib\mu_{\cX_{a}}-m_{\cX_{a}})\,,\quad
\cC_{\cY_{a}}=\f{2\pi i}{k}(-ib\mu_{\cY_{a}}-m_{\cY_{a}})\,,\quad a=1,\cdots,4\,,\\
&\Gamma_{ab}=\f{2\pi i}{k}(-ib\nu_{\cC_{ab}}-n_{\cC_{ab}})\,,\quad
\Gamma_{\cX_{a}}=\f{2\pi i}{k}(-ib\nu_{\cX_{a}}-n_{\cX_{a}})\,,\quad
\Gamma_{\cY_{a}}=\f{2\pi i}{k}(-ib\nu_{\cY_{a}}-n_{\cY_{a}})\,,\quad a=1,\cdots,4\,.
\end{align}
\end{subequations}
Combine the parameters on the right-hand sides into vectors $\vec{\mu},\vec{\nu},\vec{m},\vec{n}$ with elements
\begin{subequations}
\begin{align}
&\vec{\mu}=\{\mu_{ab},\mu_5,\mu'_5,\mu_{\cC_{ab}},\mu_{\cX_a},\mu_{\cY_a}\}\,,\quad
\vec{m}=\{m_{ab},m_5,m'_5,m_{\cC_{ab}},m_{\cX_a},m_{\cY_a}\}\,,\\
&\vec{\nu}=\{\nu_{ab}-\nu'_{ab},\nu_5,\nu'_5,\nu_{\cC_{ab}},\nu_{\cX_{a}},\nu_{\cY_a}\}\,,\quad
\vec{n}=\{n_{ab}-n'_{ab},n_5,n'_5,n_{\cC_{ab}},n_{\cX_a},n_{\cY_a}\}\,.
\end{align}
\end{subequations}
Then constraints $\{\cC_A\}_{A=1}^{18}$ and their conjugate momenta $\{\Gamma_A\}_{A=1}^{18}$ give the following relations \footnote{
In general, the quantization of the constraints is implied by $e^{\chi_{I}+\chi_{I}^{\prime}}\equiv c_{I}=\exp\left[\frac{2\pi i}{k}(-ib\mu_{I}-m_{I})\right]= q$ and therefore $\mu_I= iQ$ and $m_I= 0$. In the case of $\cC_{\cX_a}$ and $\cC_{\cY_a}$, due to the addition of $\pm 2\pi i$ in the definitions \eqref{eq:constraint_XY}, it implies $q^{-1}e^{\cX_a+\cX'_a}= 1,q^{-s_a}e^{s_a(\cY_a+\cY'_a)}= 1 $ hence $\mu_{\cX_a}=\mu_{\cY_a}=0$ and $m_{\cX_a}=m_{\cY_a}=0$. In the case of $\cC_{ab}$ which involving FN coordinates, we use the relation $2L=\chi_{1}+\chi_{2}+\chi_{3}-3\pi i,\ 2L^{\prime}=\chi_{1}^{\prime}+\chi_{2}^{\prime}+\chi_{3}^{\prime}-3\pi i$ and derive the constraint for FN coordinates: $e^{2L+2L^{\prime}}=q^{-3}c_{1}c_{2}c_{3}= 1$. Then we also obtain $\mu_{\cC_{ab}}=0$ and $m_{\cC_{ab}}=0$ from $\cC_{ab}$. }
\begin{subequations}
\begin{align}
&\mu'_{ab}=\mu_{\cC_{ab}}-\mu_{ab}\,,\quad
m'_{ab}={\rm mod}(m_{\cC_{ab}}-m_{ab},k)\,,\quad
\nu_{\cC_{ab}}=\nu'_{ab}\,,\quad
n_{\cC_{ab}}=n'_{ab}\,,
\\
&\mu'_a=\mu_{\cX_a}+iQ-\mu_a\,,\quad
m'_a={\rm mod}(m_{\cX_a}-m_a,k)\,,
\quad\forall a=1,\cdots,4\,, \\
&\nu'_a=-\mu_{\cY_a}-iQs_a+\nu_a\,,\quad
n'_a= {\rm mod}(-m_{\cY_a}+n_a,k)\,,
\quad\forall a=1,\cdots,4\,,
\label{eq:Cya}\\
&\nu_{\cX_a}=\nu'_a\,,\quad
\nu_{\cY_a}=-\mu_a\,,\quad
n_{\cX_a}=n'_a\,,\quad
n_{\cY_a}=-m_a\,,\quad \forall a=1,\cdots,4\,.
\end{align}
\end{subequations}
Also denote the imaginary parts $\alpha_{ab}=\im(\mu_{ab}),\alpha'_{ab}=\im(\mu'_{ab}),\alpha_{a}=\im(\mu_{a}),\alpha'_{a}=\im(\mu'_{a})$ and $\beta_{ab}=\im(\nu_{ab}),\beta'_{ab}=\im(\nu'_{ab}),\beta_{a}=\im(\nu_{a}),\beta'_{a}=\im(\nu'_{a})$.

Apply these notations, the amplitude \eqref{eq:final_amplitude} can be written more explicitly. 
To shorten the notation, we denote 
\be\ba{ll}
\mu_+=\left\{\mu_{ab},\{\mu_a\}_{a=1}^5\right\}\,,\quad &
\mu_-=\left\{\mu_{\cC_{ab}}-\mu_{ab},\{\mu_{\cX_a}+iQ-\mu_{a}\}_{a=1}^{4},\mu'_5\right\}\,,\\[0.15cm]
m_+=\left\{m_{ab},\{m_a\}_{a=1}^5\right\}\,,\quad &
m_-=\left\{m_{\cC_{ab}}-m_{ab},\{m_{\cX_a}-m_{a}\}_{a=1}^{4},m'_5\right\}\,.
\ea\ee
Then
\be
\cZ^0_{M_{+\cup-}}(\vec{\mu}|\vec{m})
=\f{1}{k^4}\sum_{\{m_{a}\}\in(\Z/k\Z)^4} \int_{\cC^{\times4}}[\rd\mu_a]\,
e^{\f{2\pi i}{k}\lb\sum_{a=1}^4(\mu_{\cY_a}+iQs_a)\mu_{a}-m_{\cY_a} m_{a}\rb}\,
\cZ_{M_+} (\mu_+|m_+)\cZ_{M_-}(\mu_-|m_-)\,,
\label{eq:final_amplitude_2}
\ee
where $[\rd\mu_a]:=\rd\mu_1\rd\mu_2\rd\mu_3\rd\mu_4$ denotes four copies of measures for $\mu_a$. The integration contour $\mathcal{C}$ is along $\mu_a=\mathrm{Re}(\mu_a)+i\alpha_a$ with fixed $\alpha_a$.
When constraints are imposed, \ie
\be
\cC_A=0 \,,\quad \forall A=1,\cdots,18\quad \Longleftrightarrow \quad\mu_{\cC_{ab}}=\mu_{\cX_a}=\mu_{\cY_a}=0= m_{\cC_{ab}}=m_{\cX_a}=m_{\cY_a}
\,,\ee
we obtain the partition function of Chern-Simons theory on $M_{+\cup -}$:
\be
\cZ_{M_{+\cup-}}(\vec{\mu}|\vec{m})
=
\f{1}{k^4}\sum_{\{m_{a}\}\in(\Z/k\Z)^4} \int_{\cC^{\times4}}[\rd\mu_a]\,
\exp\lb-\f{2\pi Q }{k}\sum_{a=1}^4s_a\mu_{a}\rb\,
\cZ_{M_+}(\mu^{\cC}_+\mid m^{\cC}_+)\,
\cZ_{M_-}(\mu^{\cC}_-\mid m^{\cC}_-)\,,
\label{eq:final_amplitude_2}
\ee
where the following notations are used 
\be\ba{ll}
\mu^{\cC}_+=\left\{\mu_{ab},\{\mu_{a}\}_{a=1}^{4},\mu_5\right\}\,,\quad &
\mu^{\cC}_-=\left\{-\mu_{ab},\{iQ-\mu_{a}\}_{a=1}^{4},\mu'_5\right\}\,,\\[0.15cm]
m^{\cC}_+=\left\{m_{ab},\{m_{a}\}_{a=1}^{4},m_5\right\}\,,\quad &
m^{\cC}_-=\left\{-m_{ab},\{-m_{a}\}_{a=1}^{4},m'_5\right\}\,.
\ea\ee

$\cZ_{M_\pm}\in\cF^{(k)}_{\mathfrak{P}(M_{\pm})}$ implies that the following two functions
\be
f_{+}\left(\mu^{\cC}_+\mid m^{\cC}_+\right)=e^{-\frac{2\pi }{k}\sum\limits_{a=1}^4\beta_a\mu_a}\cZ_{M_+}\left(\mu^{\cC}_+\mid m^{\cC}_+\right)\,,\quad
{f_{-}\left(\mu^{\cC}_-\mid m^{\cC}_-\right)=e^{\frac{2\pi }{k}\sum\limits_{a=1}^4\beta'_a\mu_a}\cZ_{M_-}\left(\mu^{\cC}_-\mid m^{\cC}_-\right)}
\ee
are Schwartz functions on $\cC^{\times 4}$. With the following constraint on the positive angle structure resulting from $\cC_{\mathcal{Y}_a}=0$ (see \eqref{eq:Cya})
\be
\beta_a-\beta_a'=s_aQ\,, 
\ee
the partition function \eqref{eq:final_amplitude_2} can be rewritten as 
\be
\cZ_{M_{+\cup-}}(\vec{\mu}|\vec{m})=\f{1}{k^4}\sum_{\{m_{a}\}\in(\Z/k\Z)^4} \int_{\cC^{\times4}}[\rd\mu_a]\,
f_{+}\left(\mu^{\cC}_+\mid m^{\cC}_+\right)\, f_{-}\left(\mu^{\cC}_-\mid m^{\cC}_-\right)\,.
\ee
It is manifest that $\cZ_{M_{+\cup-}}(\vec{\mu}|\vec{m})$ is absolutely convergent.

\subsection{Coherent state representation}
\label{subsec:separate_amplitudes}

The amplitude \eqref{eq:final_amplitude_2} is now written in terms of coordinates shared by the two manifolds $M_+$ and $M_-$ due to the gluing constraints. We would like to separate the variables from $\cZ_{M_+}$ and $\cZ_{M_-}$ so that it is easier to relate to vertex amplitudes of spinfoam. We make use of the overcompleteness relation \eqref{eq:coherent_over-completeness_combined} of the coherent states. Then we apply the procedure as in Section \ref{subsubsec:second_simplicity} to impose the simplicity constraint to the coherent state labels. 

To shorten the notation, we denote $\re(\mu)$ simply by $\mu\in\R$ and specify its imaginary part by $\alpha=\im(\mu)$ if any in this subsection. 
For each gluing 4-holed sphere $\cS_a$, we need to use the relation \eqref{eq:coherent_over-completeness_combined}. First, we rewrite the amplitude \eqref{eq:final_amplitude_2} as (we omit here the labels not relevant to $\{\cS_a\}_{a=1}^4$ for conciseness) 
\begin{multline}
\cZ_{M_{+\cup-}}(\vec{\mu}|\vec{m})
= \f{1}{k^4}\sum_{\{m_a,m'_a\}\in(\Z/k\Z)^{8}} \int_{\R^{8}}[\rd\mu_a][\rd\mu'_a]\\
\prod_{a=1}^4 {\delta_{\mu'_a,-\mu_a}\delta_{e^{\f{2\pi i}{k}(m_a+m'_a)},1}}\,  f_{+}\left(\{\mu_{a}+i\alpha_a \mid m_{a}\}_{a=1}^{4}\right)\, 
{f_{-}\left(\{\mu'_{a}+i(Q-\alpha_a)\mid m'_{a}\}_{a=1}^{4}\right)\,.}
\end{multline}
Then we express the delta distributions by the coherent states through \eqref{eq:coherent_over-completeness_combined}
\be
{\delta_{\mu'_a,-\mu_a}\delta_{e^{\f{2\pi i}{k}(m_a+m'_a)},1}}
=\lb\f{k}{4\pi^2} \rb^2 
\int_{\bC\times\bT^2}\rd\rho_a\, \Psi^0_{\rho_a}(\mu_a|m_a)\bar{\Psi}^0_{\rho_a}(-\mu'_a|-m'_a)
{\equiv}\int_{\bC\times\bT^2}\rd\rho_a\, \Psi^0_{\rho_a}(\mu_a|m_a)\Psi^0_{\trho_a}(\mu'_a|m'_a) \,,
\label{eq:delta_span}
\ee
where $\tilde{\rho}_a=(-\bar{z}_a,-x_a,y_a)$ given $\rho_a=(z_a,x_a,y_a)$ and we have identified $\Psi^0_{\tilde{\rho}_a}(\mu|m)\equiv\bar{\Psi}^0_{\rho_a}(-\mu|-m)$. 

Lastly, take the inner product of $\cZ_{M_{+\cup-}}(\vec{\mu}|\vec{m})$ with coherent states $\Psi^0_{\eta_5}(\mu_5|m_5)$ on $\cS_5$ and  ${\Psi^0_{\tilde{\rho}_5}(\mu'_5|m'_5)}$ on $\cS_6$ that are not glued. 
The full partition function for $M_{+\cup -}$ can now be written as
\be
\cZ_{\trho_5,\eta_5}(\{\mu_{ab}+i\alpha_{ab}|m_{ab}\}_{(ab)})=\lb\f{k}{4\pi^2} \rb^8 
\int_{(\bC\times\bT^2)^{\times 4}}[\rd\rho_a]\,
\cZ_{M_+}(\vec{\rho},\eta_5)\,\cZ_{M_-}(\vec{\rho},\trho_5)\,,
\label{eq:total_amplitude}
\ee
 where $\vec{\rho}=\{\rho_a\}_{a=1}^4$ with $\rho_a=(x_a,y_a,z_a)\in \bT^2\times\bC$ and similar for $\trho_5,\eta_5$. $\cZ_{M_\pm}$ in \eqref{eq:total_amplitude} read
\begin{subequations}
\label{eq:split_amplitude}
\begin{align}
\cZ_{M_+}(\vec{\rho},\eta_5)
&=\sum_{\{m_a\}\in(\Z/k\Z)^5} \int_{\R^5}\{\rd\mu_a\}\prod_{a=1}^5\Psi^0_{\rho_a}(\mu_a|m_a)\big|_{\rho_5\to\eta_5}\,
f_{+}\left(\{\mu_{a}+i\alpha_a\}_{a=1}^4,\mu_5+i\alpha_5|\{m_a\}_{a=1}^4,m_5\right)\,,\\
\cZ_{M_-}(\vec{\rho},\trho_5)
&=\sum_{\{m'_a\}\in(\Z/k\Z)^5}\int_{\R^5}\{\rd\mu'_a\}\prod_{a=1}^5
{\Psi^0_{\trho_a}(\mu'_a|m'_a)}\, 
f_{-}(\{\mu'_a+i(Q-\alpha_a)\}_{a=1}^4,\mu_5'+i\alpha_5'|\{m'_a\}_{a=1}^4,m_5')\,,
\label{eq:edge_amplitude_def}
\end{align}
\end{subequations}
where $\{\rd\mu_a\}=\rd\mu_1\cdots\rd\mu_5$ denotes 5 copies of measure for $\mu_a$ and similarly for $\{\rd\mu'_a\}$.

\begin{lemma}\label{boundedness}

Both $|\cZ_{M_+}(\vec{\rho},\eta_5)|$ and $|\cZ_{M_-}(\vec{\rho},\trho_5)|$ are bounded from above on $(\bC\times\bT^2)^{\times 4}$ for any given boundary data $(\lambda_{ab},\eta_5,\trho_5)$.

\end{lemma}

\begin{proof}
Recall the expression of the coherent state $\Psi^0_{\rho}(\mu\mid m)=\psi^0_{z}(\mu)\xi_{(x,y)}(m)$. $|\xi_{(x,y)}(m)|$ is bounded on $\mathbb{T}^2$ for all $m$'s, because $\xi_{(x,y)}(m)$ relates to the Jacobi theta function by {$\xi_{(x, y)}(m)={\sqrt[4]{2} }\,{k^{-3 / 4}}e^{-\frac{k y(y-i x)}{4 \pi}} \vartheta_3\left(\frac{1}{2}\left(-\frac{2 \pi m}{k}+x+i y\right), e^{-\frac{\pi}{k}}\right)$}, and $\vartheta_3\left(z, e^{-\frac{\pi}{k}}\right)$ is analytic on $\mathbb{C}$ and $x,y$ are bounded. On the other hand, $|f_\pm|$ is bounded on $\mathbb{R}^5$ for all $\vec{\mu}$, since they are Schwartz functions. Therefore,
\be
|\cZ_{M_+}(\vec{\rho},\eta_5)|\leq \sum_{\{m_a\}\in(\Z/k\Z)^5} \int_{\R^5}\{\rd\mu_a\}|f_+|\prod_{a=1}^5|\Psi^0_{\rho_a}|
\leq C_+\prod_{a=1}^5 \lb\int_{\mathbb{R}^5} \mathrm{~d} \mu_a \,e^{-\frac{\pi}{k} \sum_a\left(\mu_a-\frac{k}{\pi \sqrt{2}} \operatorname{Re}\left(z_a\right)\right)^2}\rb=C_+ k^{5/2}
\ee
for some $0<C_+<\infty$. The same argument also holds for $\cZ_{M_-}$ which leads to $|\cZ_{M_-}(\vec{\rho},\trho_5)|\leq C_-k^{5/2}$ for some $0<C_-<\infty$\,.

\end{proof}

\subsection{The face amplitude and the full amplitude for $M_{+\cup -}$}
\label{subsec:full_amplitude}

 After obtaining the Chern-Simons partition function in the coherent state representation, we are left to impose the simplicity constraints as described in Section \ref{subsec:simplicity_constraint} to define the spinfoam amplitude for $M_{+\cup -}$. That is, to impose the first-class constraints, we require $\mathrm{Re}(\mu_{ab})=0\,,\forall(ab)$ and that $m_{ab}$ depends on $j_{ab}$ in the way of \eqref{eq:m_to_j}. The second-class constraints, on the other hand, are imposed by requiring that the coherent state labels $\rho_a=(z_a,x_a,y_a)\,,\forall a=1,\cdots,5$ are parameterized by $(\htheta_a,\hphi_a)\in [0,\pi]\times [0,\pi]$ satisfying the triangle inequality \eqref{eq:range_theta_phi}. We denote these coherent state labels as $\hat{\rho}_a$. 

One last ingredient to include for completing the amplitude for $M_{+\cup -}$ is the face amplitude, since there are torus cusp boundaries in the manifold $M_{+\cup -}$ and each torus cusp corresponds to an internal face in the spinfoam melon graph. There are in total 6 torus cusps, each of which contributes a face amplitude depending on a spin (from the lesson on 3D spinfoams and the EPRL-FK model). 

Denote $\vec{j}\equiv\{j_{ab}\}=\{\vec{j}_f,\vec{j}_b\}$ with $\vec{j}_b$ being
the spins for annuli connected to the boundary \ie for $(ab)=\{(15),(25),(35),(45)\}$ and $\vec{j}_f$ for the internal tori \ie for $(ab)=\{(12),(13),(14),(23),(24),(34)\}$. 
The form of the face amplitude should relate to the boundary Hilbert space and the amplitude behaviour under the decomposition \cite{Bianchi:2010fj}. 
According to the combinatorial quantization of the Chern-Simons theory \cite{Alekseev:1994au,Alekseev:1994pa,Alekseev:1995rn}, the quantum states of Chern-Simons theory at level $k$ is described by the quantum group deformation of the gauge group. After imposing the simplicity constraints, the gauge group is reduced to $\SU(2)$ (as we impose the reality conditions on the trace coordinates). 
Therefore, we expect that the boundary states are $\fq$-deformed spin network states of the quantum group $\SU_\fq(2)$ with $\fq=e^{2\pi i/k}$ a root of unity depending on the Chern-Simons level $k$. 
We postulate
a face amplitude 
\be
\cA_f(2j_f)=[2j_f+1]_\fq^\mu \, {e^{\f{ik}{2\pi}\cF_f(-\f{2\pi i}{k}2j_f)}}
\,,\quad\mu\in \R\,,\quad
j_f=0,\f12,\cdots,\f{k-1}{2}
\label{eq:face}
\ee
with an undetermined power $\mu$, where $[n]_\fq:=\f{\fq^n-{\fq}^{-n}}{\fq-\fq^{-1}}$ is a $\fq$-number. The limit $[n]_\fq\xrightarrow{k\rightarrow\infty} n$ relates $[2j_f+1]_\fq^\mu$ to $(2j_f+1)^\mu$ used in the EPRL-FK model. $\cF_f$ is a real function that is determined in a moment. The reason of including $e^{\frac{ik}{2\pi}\cF_f}$ is that the Chern-Simons partition function is a wave function (of position variables $\fQ_I$), which is determined up to a phase. 

The full spinfoam amplitude for the melon graph then reads
\be
\cZ_{\heta_5,\htrho_5}(\vec{\alpha}|\vec{j}_b)
    =\sum_{j_f=0}^{(k-1)/2}\prod_{f=1}^6\mathcal{A}_f(2j_f)
\int_{\overline{\cM}_{\vec j}}[\rd\hat{\rho}_a ]
\cA_{v,+}(\vec{\alpha},\vec{j},\vec{\hat{\rho}},\heta_5)\,
\cA_{v,-}(\vec{\alpha},\vec{j},\vec{\hat{\rho}},\htrho_5)\,,
\label{eq:full_amplitude}
\ee
where $\vec{\alpha}=\{\{\alpha_{ab},\beta_{ab}\}_{(ab)},\{\alpha_a,\beta_a\}_{a=1}^4,\alpha_5,\alpha'_5,\beta_5,\beta'_5\}$ are all the positive angle dependence of the full amplitude. Each integral  $\int\rd\hat{\rho}_a $ is over $\overline{\cM}_{\vec j}$ satisfying the simplicity constraints on $\cS_a$. The vertex amplitude $\cA_{v,\pm}$ is obtained by restricting the variables in $\cZ_{M_\pm}$ to satisfy the simplicity constraints.

\begin{theorem}

The melonic spinfoam amplitude $\cZ_{\heta_5,\htrho_5}(\vec{\alpha}|\vec{j}_b)$ is finite for any given boundary data $\{\heta_5,\htrho_5,\vec{j}_b\}$.

\end{theorem}

\begin{proof}
Both $|\cA_{v,\pm}|$ are bounded in the integration domain, since $|\cZ_{M_\pm}|$ are bounded by Lemma \ref{boundedness}. Then the integral is absolutely convergent since the domain of $\hat{\rho}_a$ is compact. Moreover, the sum over $j_f$ is a finite sum. We then conclude that $\cZ_{\heta_5,\htrho_5}(\vec{\alpha}|\vec{j}_b)$ is finite. 
\end{proof}

The sums over different $j_f$'s in \eqref{eq:full_amplitude} are independent. However, the range of $(\hat{\theta}_a,\hphi_a)_{a=1,\cdots,4}$, which has been denoted by $\overline{\cM}_{\vec j}$, is constrained by the triangle inequality \eqref{eq:range_theta_phi} (thus $\overline{\cM}_{\vec j}$ depends on both $j_f$ and the boundary data $j_b$). For certain $j_f$ in the sum, $\overline{\cM}_{\vec j}$ may become measure-zero, then the integral vanishes. For instance, it happens for $j_f$'s violating the triangle inequality or $j_f=0$ at some $f$.

\section{The large-$k$ behavior of the melonic amplitude }
\label{sec:radiative_correction}

 In this section, we use stationary phase analysis to analyze the large-$k$ (equivalently $\Lambda\to0$) behaviour of the melonic amplitude \eqref{eq:full_amplitude}. The sum over $j_f$'s is subject to the triangle inequality. Recall the relation $2j_{ab}=m_{ab}$, we have
\begin{multline}
\cZ_{\heta_5,\htrho_5}(\vec{\alpha}|\vec{m})
=\sum_{\{m_f\in \mathbb{Z}/k\mathbb{Z}\}} \prod_{f=1}^6[m_f+1]_\fq^\mu e^{\frac{ik}{2\pi}\cF_f(-\f{2\pi i}{k}m_f)} \\
\int_{\overline{\cM}_{\vec m}}\left[\rd \hrho_a\right]\,
\cA_{v,+}\left(\{m_{ab}\}_{(ab)},\{\hrho_a\}_{a=1}^4,\heta_5\right)
\cA_{v,-}\left(\{m_{ab}\}_{(ab)},\{\hrho_a\}_{a=1}^4,\htrho_5\right)
\,,
\label{eq:final_amplitude_3}
\end{multline}
where $\overline{\cM}_{\vec m}\equiv{\overline{\cM}_{\vec j}}$. The vertex amplitudes for $M_\pm$ are explicitly given by 
\begin{subequations}
\label{eq:split_amplitude}
\begin{multline}
\cA_{v,+}\left(\{m_{ab}\}_{(ab)},\{\hrho_a\}_{a=1}^4,\heta_5\right)
=\sum_{\{m_a\}\in(\Z/k\Z)^5} \int_{\R^5}\{\rd\mu_a\}\prod_{a=1}^5\Psi^0_{\hrho_a}(\mu_a|m_a)\big|_{\htrho_5\to\heta_5}\,e^{-\frac{2\pi}{k}\sum_{a=1}^4\beta_a(\mu_a+i\alpha_a)}\\
\mathcal{Z}_{M_+}\left(i\alpha_{ab},\{\mu_{a}+i\alpha_a\}_{a=1}^4,\mu_5+i\alpha_5|m_{ab},\{m_a\}_{a=1}^4,m_5\right)\,,
\end{multline}
\begin{multline}
\cA_{v,-}\left(\{m_{ab}\}_{(ab)},\{\hrho_a\}_{a=1}^4,\htrho_5\right)
=\sum_{\{m'_a\}\in(\Z/k\Z)^5}\int_{\R^5}\{\rd\mu'_a\}\prod_{a=1}^5
\Psi^0_{\htrho_a}(\mu'_a|m'_a)\, e^{-\frac{2\pi}{k}\sum_{a=1}^4\beta'_a(\mu'_a-i\alpha_a)} \\
\cZ_{M_-}\left(-i\alpha_{ab},\{\mu'_a+i(Q-\alpha_a)\}_{a=1}^4,\mu_5'+i\alpha_5'|-m_{ab},\{m'_a\}_{a=1}^4,m_5'\right)\,,
\label{eq:edge_amplitude_def}
\end{multline}
\end{subequations}

We are interested in the scaling behaviour of the amplitude \eqref{eq:final_amplitude_3} when $k\to\infty$, while the boundary data is fixed. Here, the boundary data includes $\vec{j_b}=\{j_{ab}=\frac{m_{ab}}{2}\}_{(ab)}$ for $(ab)={(15),(25),(35),(45)}$ and the coherent state labels $\hat{\eta}_5$ and  $\htrho_5$ (corresponding to $\mathcal{S}_5$ and $\mathcal{S}_6$ in fig.\ref{fig:glue}). However, the parameters $\{\mu_I,\nu_I,m_I,n_I\}$ of $\{\fQ_I,\fP_I\}$ involved in the integrals and sums all scale linearly in $k$ as can be seen from their definitions \eqref{eq:Z_to_mu_m}. This motivates us to change variables to the scale-invariant ones so that the large-$k$ approximation can be analyzed by the stationary phase method. In Section \ref{subsec:large_k}, we first make such a change of coordinates, with which we rewrite the amplitude \eqref{eq:final_amplitude_3} for the melon graph. At the large-$k$ regime, an effective action of the amplitude can be formulated. 
In Section \ref{subsec:stationary_analysis}, we apply the stationary analysis on the effective action to find the critical points which dominate the contributions to the amplitude. 
The effective action at the critical points turns out to be a pure phase as analyzed in Section \ref{subsec:critical action}.
The scaling of the amplitude in $k$ partially comes from (the determinant of) the Hessian matrix of the effective action, which we analyze in detail in Section \ref{subsec:Hessian}.

\subsection{Change coordinates and take the large $k$ approximation}
\label{subsec:large_k}

We convert the parameters $\{\mu_I,\nu_I,m_I,n_I\}_{I=1}^{12}$ into the coordinates $\{\fQ_I,\fP_I\}_{I=1}^{12}$ by the relations 
\begin{subequations}
\begin{align}
\mu_I=\f{kb}{2\pi(b^2+1)}\lb\fQ_I+\tfQ_I\rb\,,\quad &
m_I=\f{ik}{2\pi(b^2+1)}\lb\fQ_I-b^2\tfQ_I\rb\,,
\label{eq:m_mu-FN_FG}\\
\nu_I=\f{kb}{2\pi(b^2+1)}\lb\fP_I+\tfP_I\rb\,,\quad &
n_I=\f{ik}{2\pi(b^2+1)}\lb\fP_I-b^2\tfP_I\rb\,,
\label{eq:n_nu-FN_FG}
\end{align}
\label{eq:mn_munu-FN_FG}
\end{subequations}
which are the generalization of those in \eqref{eq:Z_to_mu_m} by allowing analytic continuation of $\mu_I,\nu_I$ to $\bC$ (hence $\tfQ_I$ (\resp $\tfP_I$) is {\it not} the complex conjugate of $\fQ_I$ (\resp $\fP_I$) in general). The constraints $\mu'_a+i\alpha'_a=iQ-(\mu_a+i\alpha_a)\,,m'_a=-m_a$ ($\mu_a,\mu_a'\in\mathbb{R}$) is translated to constraints on $\cX'_a$ and $\tilde{\cX}'_a$ as (recall the definitions \eqref{eq:q_qt_def} of $h$ and $\tilde{h}$)
\be
\cX'_a=-\cX_a+2\pi i\lb1+\f{b^2+1}{k}\rb\,,\quad
\tilde{\cX}'_a=-\tilde{\cX}_a+2\pi i\lb1+\f{b^{-2}+1}{k}\rb\,.
\ee

When the first-class simplicity constraints are imposed, we demand $\mathrm{Re}(\mu_{ab})=0$ for all annuli $(ab)$'s. 
Each of these constraints is translated into a constraint between the annulus variables $L_{ab}$ and $\tilde{L}_{ab}$:  
 \be
 2\tilde{L}_{ab}=-2L_{ab}+O(k^{-1})\,.
 \label{eq:Lt_from_L}
 \ee
 where $O(k^{-1})$ relates to $\alpha_{ab}$.

Recall that the boundary data $j_{ab}=\frac{m_{ab}}{2}$ are held fixed for $(ab)={(15),(25),(35),(45)}$ when we take $k\to\infty$. It implies that we scale some boundary $L_{ab}$'s to zero at the same time, \ie
\be
L_{ab}=O(k^{-1}),\quad\text{for}\ (ab)={(15),(25),(35),(45)}.\label{Labscaling}
\ee

In order to deal with $\cZ_{M_\pm}$ at large-$k$ in a uniform way, it is convenient to define the  following FN and FG coordinates 
\begin{subequations}
\begin{align}
&\vec{\fQ}_+=\{\{2L_{ab}\}_{(ab)},\{\cX_a\}_{a=1}^{5}\}\equiv\vec{\cQ}^+\,,\quad &&
\vec{\tfQ}_+=\{\{-2L_{ab}\}_{(ab)},\{\tilde{\cX}_a\}_{a=1}^{5}\}\equiv\vec{\tcQ}^+|_{2\tilde{L}_{ab}\rightarrow-2L_{ab}}\,, 
\label{eq:def_fQ_tfQ_+}\\
&{\vec{\fQ}_-=\{\{-2L_{ab}\}_{(ab)},\{\cX'_a\}_{a=1}^{5}\}\equiv \vec{\cQ}^-|_{2L'_{ab}\rightarrow-2L_{ab}}\,,}&&
{\vec{\tfQ}_-=\{\{2L_{ab}\}_{(ab)},\{\tilde{\cX}'_a\}_{a=1}^{5}\}\equiv \vec{\tcQ}^-|_{2\tilde{L}'_{ab}\rightarrow2L_{ab}}\,, }
\label{eq:def_fQ_tfQ_-}\\
&\vec{\fP}_+=\{\{\cT_{ab}\}_{(ab)},\{\cY_a\}_{a=1}^{5}\}\equiv \vec{\cP}^+\,,\quad &&
\vec{\tfP}_+=\{\{\widetilde{\cT}_{ab}\}_{(ab)},\{\tilde{\cY}_a\}_{a=1}^{5}\}\equiv \vec{\tcP}^+\,, \\
&
{\vec{\fP}_-=\{\{\cT'_{ab}\}_{(ab)},\{\cY'_a\}_{a=1}^{5}\}}\equiv \vec{\cP}^-\,,\quad &&
{\vec{\tfP}_-=\{\{\widetilde{\cT}'_{ab}\}_{(ab)},\{\tilde{\cY}_a'\}_{a=1}^{5}\} \equiv \vec{\tcP}^-\,.}
\end{align}
\label{eq:fQ-fPt_def}
\end{subequations}
One can then define the parameter vectors $\vec{\mu}_\pm,\vec{\nu}_\pm,\vec{m}_\pm$ and $\vec{n}_\pm$ of these coordinates accordingly. We will also extensively use the notation $a^\times:=a+10$ in the rest of the paper. 

 The amplitude $\cZ_{\heta_5,\htrho_5}(\vec{\alpha}|\vec{m})$ involves some sums $\sum_{n\in\mathbb{Z}/k\mathbb{Z}}\cdots $ where $n\in\{m_f,m_a,m_a',\vec{n}_\pm\}$. We need to relate the sums to integrals in order to apply the method of stationary phase. The trick is choosing a representation of the sum followed by the Poisson resummation \footnote{We use the Poisson resummation formula in \cite{BJBCrowley_1979}: $\sum_{n=0}^{k-1}f(n)
=\sum_{p\in\Z}e^{2 \pi ip(\alpha-\f12)}\int_{0}^{k}\rd n \,f(n+\alpha-\f12)\,e^{2\pi i p n}$ for any $\alpha\in\R$ satisfying $|\alpha|<\f12$. Take $\alpha=\f12-\delta$ with
$\delta>0$ being arbitrarily small. Eq.\eqref{poisum1} is obtained by a change of variable. 
{The sum of $m'_a$ after imposing the gluing constraint $m'_a=-m_a$ becomes $\sum_{m'_a=-k+1}^0$. However, $f(m'_a)$ involved in our discussion are periodic in $m'_a$ \ie $f(m'_a+k)=f(m'_a)$, hence it does no harm to shift to $\sum_{m'_a=0}^{k-1}$ then \eqref{poisum1} can still be applied.}} 
\be
\sum_{{n} \in \mathbb{Z} / k \mathbb{Z}} f({n} )=\sum_{n =0}^{k-1} f(n)=\sum_{p \in \mathbb{Z}} \int_{-\delta}^{k-\delta} \mathrm{d} n \, f({n} ) e^{2 \pi i {p} n } 
=\frac{k}{2 \pi} \sum_{p \in \mathbb{Z}} \int_{-\delta/k}^{2 \pi-\delta/k} \mathrm{d} \mathcal{J}\, f\left(\frac{k}{2\pi}\mathcal{J} \right) e^{i k {p} {\mathcal{J}} }\,,
\label{poisum1}
\ee
where $\mathcal{J}=2\pi n/k$ and $\delta>0$ is arbitrarily small. The  application of this formula to the sums of ${n}_{\pm,I}$ and $m_{\pm,a^\times}$ in $\cA_{v,\pm}$ and combining the Lebesgue measure $\rd\nu_{\pm,I}$ or $\rd\mu_{\pm,a^\times}$, we obtain for all $I=1,\cdots,15\,,a=1,\cdots,5$ that
\be
\mathrm{d} \nu_{\epsilon,I} \wedge \mathrm{~d} \mathcal{J}_{\epsilon,I} =\frac{k}{2 \pi  Q} \left(-i\, \mathrm{d} \mathfrak{P}_{\epsilon,I} \wedge \mathrm{d} \widetilde{\mathfrak{P}}_{\epsilon,I}\right) \,,\quad
\mathrm{d} \mu_{\epsilon,a^\times} \wedge \mathrm{~d} \mathcal{K}_{\epsilon,a} =\frac{k}{2 \pi  Q} \left(-i\, \mathrm{d} \mathfrak{Q}_{\epsilon,a^\times} \wedge \mathrm{d} \widetilde{\mathfrak{Q}}_{\epsilon,a^\times}\right)\,,\quad \epsilon=\pm
\,,
\ee
where $\cJ_{\epsilon,I}=2\pi n_{\epsilon,I}/k$ and $\cK_{\epsilon,a}=2\pi m_{\epsilon,a^\times}/k$. 
Similarly, the sum over $m_f$ becomes
\be
\sum_{m_f \in \mathbb{Z} / k \mathbb{Z}} \cdots =\frac{k}{2\pi}\sum_{u_f\in\mathbb{Z}}\int_{-\delta/k}^{2 \pi-\delta/k} \mathrm{d} (i\mathfrak{Q}_f)e^{ -k u_f \fQ_f}\cdots,\qquad f=1,\cdots,6
\label{eq:poisson_resum_Qf}
\ee
where $i\mathfrak{Q}_f=2\pi m_f/{k}$. This procedure makes choices of the lift from $e^{{\mathfrak{P}}_{\epsilon,I}},\,e^{\widetilde{\mathfrak{P}}_{\epsilon,I}},\,e^{{\mathfrak{Q}}_{\epsilon,a^\times}},\,e^{\widetilde{\mathfrak{Q}}_{\epsilon,a^\times}},\,e^{\mathfrak{Q}_f}$ to ${\mathfrak{P}}_{\epsilon,I},\,\widetilde{\mathfrak{P}}_{\epsilon,I},\,{\mathfrak{Q}}_{\epsilon,a^\times},\,\widetilde{\mathfrak{Q}}_{\epsilon,a^\times},\,\mathfrak{Q}_f$. The integration domain $\overline{\cM}_{\vec m}$ of $\{\htheta_a,\hphi_a\}$ is well-defined with continuous ${\vec m}$.

\subsubsection{The large-$k$ approximation of the vertex amplitudes}
\label{subsubsec:large_k_vertex}

Let us first consider the large-$k$ approximation of the vertex amplitudes $\cA_{v,\pm}$. 
We apply the result in \cite{Han:2021tzw} and write the partition functions
$\cZ_{M_\pm}$ in the form of path integrals at large $k$:
\be
\cZ_{M_\epsilon}=\cN_0\sum_{\vec{p}_\epsilon\in\Z^{15}}\int_{\cC^{\times 30}_{\fP_\epsilon\times\tfP_\epsilon}}\bigwedge_{I=1}^{15}\left(-i\, \mathrm{d} \mathfrak{P}_{\epsilon,I} \wedge \mathrm{d} \widetilde{\mathfrak{P}}_{\epsilon,I}\right)
e^{{k}S_{\vec{p}_\epsilon}(\vec{\fP}_\epsilon,\vec{\tfP}_\epsilon,\vec{\fQ}_\epsilon,\vec{\tfQ}_\epsilon)}\left[1+O(1/k)\right]\,,\quad\forall \epsilon=\pm\,,
\ee
The overall constant $\cN_0=\f{4k^{15}}{(2\pi)^{30}Q^{15}}$ and the effective action can be separated into four parts as 
\be
S_{\vec{p}_\epsilon}
=S^\epsilon_0(\vec{\fP}_\epsilon,\vec{\tfP}_\epsilon,\vec{\fQ}_\epsilon,\vec{\tfQ}_\epsilon)
+S_1^\epsilon(-\bB_\epsilon^\top\cdot \vec{\fP}_\epsilon)
+\tS^\epsilon_1(-\bB_\epsilon^\top \cdot\vec{\tfP}_\epsilon)
-\f{1}{b^2+1}\vec{p}_\epsilon\cdot(\vec{\fP}_\epsilon-b^2\vec{\tfP}_\epsilon)\,.
\label{eq:effective_action_all}
\ee
The vector $\vec{p}_\epsilon$ comes from the Poisson resummation of $\vec{n}_\epsilon$ (recall the expression \eqref{eq:partition_S3G5}).  
The first three terms in \eqref{eq:effective_action_all} are explicitly \cite{Han:2021tzw} 
\begin{subequations}
\begin{align}
S_0^\epsilon(\vec{\fP}_\epsilon,\vec{\tfP}_\epsilon,\vec{\fQ}_\epsilon,\vec{\tfQ}_\epsilon)
=&-\f{1}{k}\vec{t}_\epsilon\cdot\lb\vec{\fP}_\epsilon+\vec{\tfP}_\epsilon\rb 
-\f{i}{4\pi(b^2+1)}\left[\vec{\fP}_\epsilon\cdot\lb\bA_\epsilon\bB_\epsilon^\top\cdot\vec{\fP}_\epsilon+2\vec{\fQ}_\epsilon \rb 
+b^2\vec{\tfP}_\epsilon\cdot\lb\bA_\epsilon\bB_\epsilon^\top\cdot\vec{\tfP}_\epsilon
+2\vec{\tfQ}_\epsilon \rb\right]\nn\\
&-\f{1}{2(b^2+1)}\vec{t}_\epsilon\cdot \lb\vec{\fP}_\epsilon-b^2\vec{\tfP}_\epsilon \rb\,,
\label{eq:S0}\\
S_1^\epsilon(-\bB_\epsilon^\top \cdot\vec{\fP}_\epsilon)
=&-\f{i}{2\pi(b^2+1)}\sum_{i=1}^5
\left[\Li_2(e^{-X^\epsilon_i})+\Li_2(e^{-Y^\epsilon_i})+\Li_2(e^{-Z^\epsilon_i})+\Li_2(e^{-W^\epsilon_i}) \right]\,,
\label{eq:S1}\\
\tS^\epsilon_1(-\bB_\epsilon^\top \cdot\vec{\tfP}_\epsilon)
=&-\f{i}{2\pi(b^{-2}+1)}\sum_{i=1}^5
\left[\Li_2(e^{-\Xt^\epsilon_a})+\Li_2(e^{-\Yt^\epsilon_i})+\Li_2(e^{-\Zt^\epsilon_i})+\Li_2(e^{-\Wt^\epsilon_i}) \right]\,,
\label{eq:S1t}
\end{align}
\label{eq:S0-S1t}
\end{subequations}
where $-\bB_\epsilon^\top\cdot \vec{\fP}_\epsilon=(X^\epsilon_i,Y^\epsilon_i,Z^\epsilon_i)_{i=1}^5$ with notations $(X^+_i,Y^+_i,Z^+_i)\equiv(X_i,Y_i,Z_i)$ and $(X^-_i,Y^-_i,Z^-_i)\equiv(X_{i+5},Y_{i+5},Z_{i+5})$. $i$ here denotes the octahedron $\Oct(i)$. Similarly for the tilde sectors. 
$\Li_2$ appearing in \eqref{eq:S1} and \eqref{eq:S1t} is the dilogarithm function defined as 
\be
\Li_2(z):=-\int_0^z\f{\ln(1-u)}{u}\rd u
\ee
 for $z\in\bC$.
$\bB^\top_\epsilon$ transforms the momenta $\vec{\fP}_\epsilon$ and $\vec{\tfP}_\epsilon$ on the 3-manifold $M_\epsilon$ to position variables on the octahedra $\{\Oct(i)\}$ of $M_\epsilon$. This is the reversed version of the coordinate transformation \eqref{eq:change_coordinate}. 
$W^\epsilon_i$ and $\Wt^\epsilon_i$ are obtained from the constraints on an octahedron:
\be
X^\epsilon_i+Y^\epsilon_i+Z^\epsilon_i+W^\epsilon_i=2\pi i +\f{2\pi i}{k}(b^2+1)\,,\quad
\Xt^\epsilon_i+\Yt^\epsilon_i+\Zt^\epsilon_i+\Wt^\epsilon_i=2\pi i +\f{2\pi i}{k}(b^{-2}+1)\,.
\label{eq:Qconstraint_XYZW}
\ee
The effective actions $S_1^\epsilon$ and $\tS_1^\epsilon$ as in \eqref{eq:S1} -- \eqref{eq:S1t} are obtained by taking the large-$k$ approximation of all the quantum dilogarithm functions within $\cZ_{M_\epsilon}$. As an example,
\be
\Psi_\triangle(\mu_{X_i}|m_{X_i})
=\exp\left[-\f{ik}{2\pi(b^2+1)}\Li_2(e^{-X_i})-\f{ik}{2\pi(b^{-2}+1)}\Li_2(e^{-\Xt_i})\right][1+O(1/k)]\,.
\ee

Let us consider now the inner product of $\cZ_{M_+}$ (\resp$\cZ_{M_-}$) with coherent states. 
Firstly, we use the change of variables \eqref{eq:m_mu-FN_FG} to express the coherent states $\Psi^0_{\pm,a}:=\Psi^0_{\hrho_{\pm,a}}(\fQ_{\pm,a^\times}, \tfQ_{\pm,a^\times})$ 
into functions of the new variables, where
\be\ba{c}
\begin{cases}
\hrho_{+,a}\equiv(\zh_{+,a},\xh_{+,a},\yh_{+,a}):=(\zh_a,\xh_a,\yh_a)\\
\hrho_{-,a}\equiv (\zh_{-,a},\xh_{-,a},\yh_{-,a}):=(-\hat{\bar{z}}_a,-\xh_a,\yh_a)
\end{cases}
\quad \text{for }a=1,\cdots,5\,,\\[0.15cm]
\hrho_{+,5}\equiv(\zh_{+,5},\xh_{+,5},\yh_{+,5}):=(\zh_5,\xh_5,\yh_5)\,,\quad
\hrho_{-,5}\equiv(\zh_{-,5},\xh_{-,5},\yh_{-,5}):=(-\hat{\bar{z}}'_5,-\xh'_5,\yh'_5)\,.
\ea
\label{eq:rho_notation}
\ee
We then perform the Poisson resummation for the sums over $m_{+,a}:= m_a$ and $m_{-,a}:= m'_a$ in the inner product. We also denote $\mu_{+,a}:=\mu_a\,,\,\mu_{-,a}:=\mu'_a$ for $a=1,\cdots,5$.

As a result, the inner product takes the form (we omit the factor $e^{-\f{2\pi}{k}\sum_{a=1}^4\beta_{\epsilon,a}\mu_{\epsilon,a}}$ as it is subleading at large $k$) 
\be\begin{split}
&
\cA_{v,\epsilon}(\vec{\mu}_\epsilon,\vec{m}_\epsilon,\vec{\hrho}_\epsilon)
=\sum_{\{m_{\epsilon,a}\}\in (\Z/k\Z)^5}\int_{\R^5}\{\rd\mu_{\epsilon,a}\}\,
{\cZ_{M_\epsilon}\lb\mu_{\epsilon,a}| m_{\epsilon,a}\rb}\lb\prod\limits_{a=1}^5\Psi^0_{\epsilon,a}(\mu_{\epsilon,a}|m_{\epsilon,a})\rb\\
=&\cN_1\sum_{\vec{p}_{\epsilon}\in\Z^{15}}\sum_{\vec{u}_\epsilon\in\Z^5}
\int_{\cC^{\times40}_{M_\epsilon}}\rd M_\epsilon
\,
\exp\left[ k
S^{\cohe}_{\vec{p}_\epsilon,\vec{u}_\epsilon,\vec{\hrho}_\epsilon}
(\vec{\fP}_\epsilon,\vec{\tfP}_\epsilon,\vec{\fQ}_\epsilon,\vec{\tfQ}_\epsilon)
\right]
\,,
\end{split}
\label{eq:Ztot_cohe}
\ee
where $\cN_1=\cN_0\lb\f{k^2}{4\pi^2Q}\sqrt{\f{2}{k}}\rb^5=\f{16\sqrt{2}}{(2\pi)^{40}Q^{20}}k^{45/2}$ and
\be
\int_{\cC^{\times40}_{M_\epsilon}}\rd M_\epsilon=
\int_{\cC^{\times 30}_{\fP_\epsilon\times\tfP_\epsilon}}
\bigwedge_{I=1}^{15}\left(-i\, \mathrm{d} \mathfrak{P}_{\epsilon,I} \wedge \mathrm{d} \widetilde{\mathfrak{P}}_{\epsilon,I}\right)
\int_{\cC^{\times10}_{\fQ_{\epsilon,a^\times}\times\tfQ_{\epsilon,a^\times}}}
\bigwedge_{a=1}^{5}\lb-i\, \mathrm{d} \mathfrak{Q}_{\epsilon,a^\times} \wedge \mathrm{d} \widetilde{\mathfrak{Q}}_{\epsilon,a^\times}\rb\,.
\label{eq:dM}
\ee
The effective action in the exponent is 
\be\begin{split}
&S^\cohe_{\vec{p}_\epsilon,\vec{u}_\epsilon,\vec{\hrho}_\epsilon}(\vec{\fP}_\epsilon,\vec{\tfP}_\epsilon,\vec{\fQ}_\epsilon,\vec{\tfQ}_\epsilon)\\
=&S_{\vec{p}_\epsilon}(\vec{\fP}_\epsilon,\vec{\tfP}_\epsilon,\vec{\fQ}_\epsilon,\vec{\tfQ}_\epsilon)
+\sum\limits_{a=1}^5
\left[S_{\zh_{\epsilon,a}}(\fQ_{\epsilon,a^{\times}},\tfQ_{\epsilon,a^{\times}})
+S_{(\xh_{\epsilon,a},\yh_{\epsilon,a})}(\fQ_{\epsilon,a^{\times}},\tfQ_{\epsilon,a^{\times}})
-\f{1}{b^2+1}u_{\epsilon,a}\lb\fQ_{\epsilon,a^{\times}}-b^2\tfQ_{\epsilon,a^{\times}}\rb\right]\,,
\end{split}
\label{eq:vertex_action}
\ee 
where $u_{\epsilon,a}$ in the last term comes from the Poisson resummation of $m_{\epsilon,a}$. 
Then the first two effective actions in the summation $\sum_{a=1}^5$ of \eqref{eq:vertex_action} takes the form 
\begin{subequations}
\begin{align}
&S_{\zh_{\epsilon,a}}(\fQ_{\epsilon,a^{\times}},\tfQ_{\epsilon,a^{\times}})=
-\f{ b}{2\pi(b^2+1)}\lb\fQ_{\epsilon,a^{\times}}+\tfQ_{\epsilon,a^{\times}}\rb
\left[\f{b\lb\fQ_{\epsilon,a^{\times}}+\tfQ_{\epsilon,a^{\times}}\rb}{2(b^2+1)}-\sqrt{2}\hat{\bar{z}}_{\epsilon,a}\right]-\f{1}{2\pi}\re(\zh_{\epsilon,a})^2\,,
\\
&S_{(\xh_{\epsilon,a},\yh_{\epsilon,a})}(\fQ_{\epsilon,a^{\times}},\tfQ_{\epsilon,a^{\times}})
=- \f{i\xh_{\epsilon,a}\yh_{\epsilon,a}}{4\pi}-\f{1}{4\pi}\left[\f{i\lb\fQ_{\epsilon,a^{\times}}-b^2\tfQ_{\epsilon,a^{\times}}\rb}{b^2+1}-\xh_{\epsilon,a}\right]^2-\f{1}{2\pi}\f{\lb\fQ_{\epsilon,a^{\times}}-b^2\tfQ_{\epsilon,a^{\times}}\rb  \yh_{\epsilon,a}}{b^2+1} 
\,.
\label{eq:coherent_QQt_1}
\end{align}
\label{eq:coherent_QQt}
\end{subequations}
The expression \eqref{eq:coherent_QQt_1} comes from the simplified version of $\xi_{(\xh_{\epsilon,a},\yh_{\epsilon,a})}$ when restricting $m_{\epsilon,a}=0,\cdots,k-1\,, (\epsilon\xh_{\epsilon,a}, \yh_{\epsilon,a})\in[0,2\pi]$ and neglecting the exponentially decaying contribution at large $k$. 
We keep in mind that, when converting the variables $\{\fQ_{\epsilon,a^\times},\tfQ_{\epsilon,a^\times}\}$ in \eqref{eq:coherent_QQt} back to $\mu_{\epsilon,a}$ using \eqref{eq:m_mu-FN_FG}, we should replace $\mu_{\epsilon,a}$ by $\re(\mu_{\epsilon,a})$ as these actions are form the coherent state $\Psi_{\hrho_{\epsilon,a}}^0(\re(\mu_{\epsilon,a})| m_{\epsilon,a})$. (See \eqref{eq:eom_pos_cohe} below.) 
The same expressions as in \eqref{eq:Ztot_cohe} -- \eqref{eq:coherent_QQt} have been obtained in \cite{Han:2021tzw} when considering one 3-manifold $S^3\backslash\Gamma_5$, to which we refer for detailed derivation.

\subsubsection{The full amplitudes}
\label{subsubsec:large_k_full}

Lastly, we take into account the Poisson resummation \eqref{eq:poisson_resum_Qf} for $\sum_{m_{f}}$ and change the variables from $m_{f}$ to $\fQ_f\in\{2L_{12},2L_{13},2L_{14},2L_{23},2L_{24},2L_{34}\}$. In the large-$k$ regime, the $\fq$-number $[m_f+1]_\fq$ is approximated by the integer $m_f+1$. With the simplicity constraints \eqref{eq:Lt_from_L} imposed, they are related by 
\be
m_{f}=\f{ik}{2\pi}\fQ_f\,,\quad\forall f=1,\cdots,6
\quad\Longrightarrow\quad
m_f+1=\f{ik}{2\pi}\lb\fQ_f +\f{2\pi}{ik}\rb\overset{k\rightarrow\infty}{\sim}\f{ik}{2\pi}\fQ_f \,.
\label{eq:constraint_fQ}
\ee
Since the discussion is in the large-$k$ limit, we often identify $2 L_{ab}\in i\mathbb{R}$ and do not distinguish it with $\fQ_f$.

The total amplitude can be written as 
\be
\cZ_{\heta_5,\htrho_5}(\vec{\mu}|\vec{m})
=\cN
\sum_{\substack{\vec{u}_f\in\Z^6\\\vec{u}_{\pm}\in\Z^5\\\vec{p}_\pm\in\Z^{15}}}
\int\limits_{\cC^{\times86}_{\fQ\times\fP}}
\rd \cM_{\fQ\times\fP}\,
\int_{\overline{\cM}_{\vec m}}[\rd \hrho_a]
\lb\prod_{f=1}^6(i\fQ_f)^\mu\rb e^{k S_{\tot}(\vec{\fQ}_\tot,\vec{\fP}_\tot,\vec{\tfQ}_\tot,\vec{\tfP}_\tot)}[1+O(k^{-1})]\,,
\label{eq:total_amplitude_action}
\ee
where $\cC^{\times86}_{\fQ\times\fP}$ is the integration contour for all the $86=30\times2+10\times 2+6$ momentum and position integration variables (recall \eqref{eq:dM}):
\be
\rd \cM_{\fQ\times\fP}:=
\rd M_+\wedge
\rd M_-\wedge
\lb\bigwedge_{f=1}^6\rd \lb i\fQ_f\rb\rb\,.
\ee
The prefactor $\cN$ has three sources and reads 
\be
\cN
=\lb\f{16\sqrt{2} k^{45/2}}{(2\pi)^{40}Q^{20}}\rb^2 \times\lb\f{k^2}{(2\pi)^4}\rb^4
\times \lb\f{k}{2\pi}\rb^{6+6\mu}
=\f{512\, k^{59+6\mu}}{(2\pi)^{102+6\mu}Q^{40}}\,,
\label{eq:cN}
\ee
where the first term comes from the two vertex amplitudes $\cA_{v,\pm}$ (see \eqref{eq:Ztot_cohe}), the second term comes from the delta distributions for the four glued 4-holed spheres (see \eqref{eq:delta_span}) and the last term comes from the six face amplitudes and the change of integration variables $\rd m_f=\f{k}{2\pi} \rd \lb i\fQ_f\rb$. The total effective action $S_\tot$ is a function of 30 position variables $(\vec{\fQ}_\tot,\vec{\tfQ}_\tot)$, 60 momentum variables $(\vec{\fP}_\tot,\vec{\tfP}_\tot)$ and six sets of coherent state labels $\{\{\hrho_a\}_{a=1}^4
,\heta_5,\htrho_5\}$ where (recall \eqref{eq:fQ-fPt_def})
\be
\vec{\fQ}_\tot=\{\vec{\fQ}_+,\vec{\fQ}_-\}\,,\quad
\vec{\tfQ}_\tot=\{\vec{\tfQ}_+,\vec{\tfQ}_-\}\,,\quad
\vec{\fP}_\tot=\{\vec{\fP}_+,\vec{\fP}_-\}\,,\quad
\vec{\tfP}_\tot=\{\vec{\tfP}_+,\vec{\tfP}_-\}\quad
\text{requiring }2\tilde{L}_{ab}\equiv -2L_{ab}\,,\,\, \forall(ab)\,.
\label{eq:all_fQ_fP}
\ee
 
The total action is 
\be
S_\tot(\vec{\fQ}_\tot,\vec{\fP}_\tot,\vec{\tfQ}_\tot,\vec{\tfP}_\tot)
=\sum_{\epsilon=\pm} S^\cohe_{\vec{p}_\epsilon,\vec{u}_\epsilon,\vec{\hrho}_\epsilon}(\vec{\fP}_\epsilon,\vec{\tfP}_\epsilon,\vec{\fQ}_\epsilon,\vec{\tfQ}_\epsilon)
+ \sum_{f=1}^6\lb \f{i}{2\pi}\cF_f(2L_f)- u_f\fQ_f\rb\,,
\label{eq:total_action}
\ee
where $S^\cohe_{\vec{p}_\epsilon,\vec{u}_\epsilon,\vec{\hrho}_\epsilon}(\vec{\fP}_\epsilon,\vec{\tfP}_\epsilon,\vec{\fQ}_\epsilon,\vec{\tfQ}_\epsilon)$ is defined in \eqref{eq:vertex_action}
and $\{u_f\}_{f=1}^6\in\Z^6$ come from the Poisson resummations of $\{m_f\}_{f=1}^6$. 

\medskip

Note that although the leading order behavior of $S_{\rm tot}$ is linear in $k$, it does not scale uniformly as $k\to \infty$, because \eqref{Labscaling} and the term $-\frac{1}{2} \vec{t}_\epsilon \cdot(\vec{\mathfrak{P}}_\epsilon+\vec{\widetilde{\mathfrak{P}}}_\epsilon)$ in $S_0^\epsilon$ result in some terms in $S_{\rm tot}$ not scaling in $k$. In performing the stationary phase analysis, one may firstly extract the terms in $S_{\rm tot}$ that is linear in $k$, denoted by $S'_{\rm tot}$ and derive the critical equation $\partial S'_{\rm tot}=0$, whose solutions denoted by $x_c'$ make dominant contribution to the integral \eqref{eq:total_amplitude_action}. However, we can also use $S_{\rm tot}$ and the express \eqref{eq:total_amplitude_action} for the stationary phase analysis.  The critical equation $\partial S_{\rm tot}=0$ will contain some terms of $O(k^{-1})$. The solution to $\partial S_{\rm tot}=0$ is denoted by $x_c$. The difference between $x_c$ and $x_c'$ is of $O(k^{-1})$. Therefore, the dominant contributions of \eqref{eq:total_amplitude_action} computed respectively from $x_c$ and $x_c'$ are different only by some subleading contributions of $O(k^{-1})$, which does not affect our discussion since we focus on the leading asymptotic behaviour.

\subsection{Stationary phase analysis of the effective action}
\label{subsec:stationary_analysis}

Now that we have written the total amplitude $\cZ_{\heta_5,\htrho_5}$ for $M_{+\cup-}$ in terms of the scale-invariant variables $\{\vec{\fP}_\epsilon,\vec{\tfP}_\epsilon,\vec{\fQ}_\epsilon,\vec{\tfQ}_\epsilon\}$, stationary analysis can be performed on the effective action \eqref{eq:total_action}. 
Denote for short $S^{\cohe}_\epsilon=S^\cohe_{\vec{p}_\epsilon,\vec{u}_\epsilon,\vec{\hrho}_\epsilon}$ and $S_{0,f}=S_0^++S_0^-{+\f{i}{2\pi}\cF_f(2L_f)}-u_f\fQ_f$.  
Notice that the dependence of $S_\tot$ on $(\vec{\fP}_{\epsilon},\vec{\tfP}_\epsilon)$ is all in $S_{\vec{p}_\epsilon}$ defined in \eqref{eq:effective_action_all}, 
the dependence on $\fQ_f$ is in $S_{0,f}$ and
the dependence on $(\fQ_{\epsilon,a^\times},\tfQ_{\epsilon,a^\times})$ is in $S^\cohe_\epsilon$,
one can simplify the critical equations to be
\begin{subequations}
\begin{align}
\f{\partial S_{\vec{p}_+}}{\partial \fP_{+,I}}
=\f{\partial S_{\vec{p}_-}}{\partial \fP_{-,I}}
=\f{\partial S_{\vec{p}_+}}{\partial \tfP_{+,I}}
=\f{\partial S_{\vec{p}_-}}{\partial \tfP_{-,I}}
&=0\,,\quad \forall I=1,\cdots,15\,,
\label{eq:critical_equations_1}\\
\f{\partial S_{0,f}}{\partial \fQ_{f}}
&=0\,,\quad \forall f=1,\cdots,6\,,
\label{eq:critical_equations_2}\\
\f{\partial S^{\cohe}_+}{\partial \fQ_{+,a^\times}}
=\f{\partial S^{\cohe}_-}{\partial \fQ_{-,a^\times}}
=\f{\partial S^{\cohe}_+}{\partial \tfQ_{+,a^\times}}
=\f{\partial S^{\cohe}_-}{\partial \tfQ_{-,a^\times}}
&=0\,,\quad  \forall a=1,\cdots,5\,.
\label{eq:critical_equations_3}
\end{align}
\label{eq:critical_equations}
\end{subequations}

\subsubsection{momentum aspects}

We first analyze the derivatives \eqref{eq:critical_equations_1} of $\fP_{\epsilon,I}$ and $\tfP_{\epsilon,I}$ for all $I=1,\cdots,15$. (Recall the explicit expressions \eqref{eq:effective_action_all}--\eqref{eq:S0-S1t} of the action.) 
\begin{subequations}
\begin{align}
\f{\partial S_{\vec{p}_\epsilon}}{\partial \fP_{\epsilon,I}}
=&-\f{t_{\epsilon,I}}{2k}-\f{i}{2\pi(b^2+1)}
\left[\lb\bA_\epsilon\bB_\epsilon^\top\cdot\vec{\fP}_\epsilon\rb_I +\fQ_{\epsilon,I}\right] 
-\f{t_{\epsilon,I}}{2(b^2+1)}
+\f{i}{2\pi(b^2+1)}\lb\bB_{\epsilon}\cdot \vec{P}_{\epsilon}\rb_I
-\f{p_{\epsilon,I}}{b^2+1}\,,
\label{eq:critical_fP_+R}\\
\f{\partial S_{\vec{p}_\epsilon}}{\partial \tfP_{\epsilon,I}}
=&-\f{t_{\epsilon,I}}{2k}-\f{i}{2\pi(b^{-2}+1)}
\left[\lb\bA_\epsilon\bB_\epsilon^\top\cdot\vec{\tfP}_\epsilon\rb_I +\tfQ_{\epsilon,I}\right]
+\f{t_{\epsilon,I}}{2(b^{-2}+1)}
+\f{i}{2\pi(b^{-2}+1)}\lb\bB_{\epsilon}\cdot \vec{\widetilde{P}}_{\epsilon}\rb_I
+\f{p_{\epsilon,I}}{b^{-2}+1}\,,
\label{eq:critical_tfP_+R}
\end{align}
\label{eq:critical-P}
\end{subequations}
where we have used the fact that $\bA_\epsilon\bB_\epsilon^\top$ is a symmetric matrix for both $\epsilon=\pm$ and that 
\begin{subequations}
\begin{align}
\vec{P}_\epsilon &:= 
\left\{\log\lb\f{1-e^{-X_i}}{1-e^{-W_i}}\rb\,,
\log\lb\f{1-e^{-Y_i}}{1-e^{-W_i}}\rb\,,
\log\lb\f{1-e^{-Z_i}}{1-e^{-W_i}}\rb\right\}_{\substack{i=1,\cdots,5\,\text{ if }\epsilon=+\\i=6,\cdots,10\,\text{ if }\epsilon=-}}\,,
\label{eq:P_+}\\
\vec{\widetilde{P}}_\epsilon &:= 
\left\{\log\lb\f{1-e^{-\Xt_i}}{1-e^{-\Wt_i}}\rb\,,
\log\lb\f{1-e^{-\Yt_i}}{1-e^{-\Wt_i}}\rb\,,
\log\lb\f{1-e^{-\Zt_i}}{1-e^{-\Wt_i}}\rb\right\}_{\substack{i=1,\cdots,5\,\text{ if }\epsilon=+\\i=6,\cdots,10\,\text{ if }\epsilon=-}}\,,
\label{eq:P_-}
\end{align}
\end{subequations}
which comes from the derivative of the dilogarithm function $\f{\rd \Li_2(\fz_i)}{\rd \fz_i}=-\f{1}{\fz_i}\log(1-\fz_i)$ for $\fz_i=x_i,y_i,z_i,w_i$ and $i$ labels the octahedron $\Oct(i)$. Here the imaginary part of $\log(x)$ is fixed to be in $[0,2\pi)$. $W_i,\Wt_i$ are defined as
\be
W_i:=2\pi i +\f{2\pi i}{k}\lb b^2+1\rb-X_i-Y_i-Z_i\,,\quad 
\Wt_i:=2\pi i  +\f{2\pi i}{k}\lb b^{-2}+1\rb-\Xt_i -\Yt_i-\Zt_i\,.
\nn\ee 
The critical equations \eqref{eq:critical-P} look complicated at first sight. However, as explained below, they are simply the reformulation of the algebraic curve equations \eqref{eq:A-polynomial} for ideal tetrahedra
\begin{subequations}
\begin{align}
&\fz^{-1}_i+\fz''_i-1=0 \quad \Longleftrightarrow \quad \fZ''_i=\log(1-e^{-\fZ_i})\,,\quad 
\text{ with } \fz_i=e^{\fZ_i}\,,\, \fz''_i\equiv e^{\fZ''_i}\,, \quad\forall i=1,\cdots,5\,,\\
&\fz^{-1}_i+\fz'_i-1=0 \quad \Longleftrightarrow \quad \fZ'_i=\log(1-e^{-\fZ_i})\,,\quad 
\text{ with } \fz_i=e^{\fZ_i}\,,\, \fz'_i\equiv e^{\fZ'_i}\,, \quad\forall i=6,\cdots,10\,.
\end{align}
\label{eq:A-polynomial_-eform}
\end{subequations}

For notational simplicity, we define $\vec{\cW}=\{\cW_i\}_{i=1}^{10}$ and $\vec{\cWt}=\{\cWt_i\}_{i=1}^{10}$ such that $\cW_i=W''_i\,,\, \cWt_i=\Wt''_i$ if $i=1,\cdots,5$ while $\cW_i=W'_i\,,\,\cWt_i=\Wt'_i$ if $i=6,\cdots,10$. Replacing the logarithm function $\log(1-e^{-W_i})$ by $\cW_i$ in \eqref{eq:P_+} and $\log(1-e^{-\Wt_i})$ by $\cWt_i$ in \eqref{eq:P_-}, we rewrite $\vec{P}_\epsilon$ and $\vec{\widetilde{P}}_\epsilon$ into 
\begin{subequations}
\begin{align}
\vec{P}_\epsilon&=\{\log(1-e^{-X_i})-\cW_i\,,\,\log(1-e^{-Y_i})-\cW_i\,,\,\log(1-e^{-Z_i})-\cW_i\}_{\substack{i=1,\cdots,5\,\text{ if }\epsilon=+\\i=6,\cdots,10\,\text{ if }\epsilon=-}}
\,,\\
\vec{\widetilde{P}}_\epsilon&=\{\log(1-e^{-\Xt_i})-\cWt_i\,,\,\log(1-e^{-\Yt_a})-\cWt_i\,,\,\log(1-e^{-\Zt_i})-\cWt_i\}_{\substack{i=1,\cdots,5\,\text{ if }\epsilon=+\\i=6,\cdots,10\,\text{ if }\epsilon=-}}
\,.
\end{align}
\label{eq:def_P}
\end{subequations}
Denote the original position and momentum coordinates for the ideal tetrahedra by
\begin{subequations}
\begin{align}
\vec{\Phi}_\epsilon &= \{X_i,Y_i,Z_i\}_{\substack{i=1,\cdots,5\,\text{ if }\epsilon=+\\i=6,\cdots,10\,\text{ if }\epsilon=-}}\,,
\quad \text{and}\quad
\vec{\Pi}_\epsilon = \{P_{X_i},P_{Y_i},P_{Z_i}\}_{\substack{i=1,\cdots,5\,\text{ if }\epsilon=+\\i=6,\cdots,10\,\text{ if }\epsilon=-}}\,,\\
\vec{\widetilde{\Phi}}_\epsilon &= \{\Xt_i,\Yt_i,\Zt_i\}_{\substack{i=1,\cdots,5\,\text{ if }\epsilon=+\\i=6,\cdots,10\,\text{ if }\epsilon=-}}\,,
\quad \text{and}\quad
\vec{\widetilde{\Pi}}_\epsilon = \{P_{\Xt_i},P_{\Yt_i},P_{\Zt_i}\}_{\substack{i=1,\cdots,5\,\text{ if }\epsilon=+\\i=6,\cdots,10\,\text{ if }\epsilon=-}}\,.
\end{align}
\end{subequations}
They are related to the new coordinates $\{\vec{\fQ}_\epsilon, \vec{\fP}_{\epsilon}\}$ and $\{\vec{\tfQ}_\epsilon, \vec{\tfP}_{\epsilon}\}$ by linear transformations that can be formulated neatly by the following matrix multiplications \cite{Han:2021tzw}
\be
\mat{c}{\vec{\fQ}_\epsilon -i\pi \vec{t}_\epsilon \\ \vec{\fP}_{\epsilon}}
=\mat{cc}{\bA_\epsilon & \bB_\epsilon \\ -(\bB_\epsilon^\top)^{-1} & 0}
\mat{cc}{\vec{\Phi}_\epsilon \\ \vec{\Pi}_\epsilon}\,,\quad
\mat{c}{\vec{\tfQ}_\epsilon+i\pi \vec{t}_\epsilon \\ \vec{\tfP}_\epsilon}
=\mat{cc}{\bA_\epsilon & \bB_\epsilon \\ -(\bB_\epsilon^\top)^{-1} & 0}
\mat{cc}{\vec{\widetilde{\Phi}}_\epsilon \\ \vec{\widetilde{\Pi}}_\epsilon}\,.
\ee
Or inversely,
\be
\mat{cc}{\vec{\Phi}_\epsilon \\ \vec{\Pi}_\epsilon}
=\mat{cc}{0 & -\bB_\epsilon^\top \\\bB_\epsilon^{-1} & \bA^\top}
\mat{c}{\vec{\fQ}_\epsilon -i\pi \vec{t}_\epsilon \\ \vec{\fP}_{\epsilon}}\,,\quad
\mat{cc}{\vec{\widetilde{\Phi}}_\epsilon \\ \vec{\widetilde{\Pi}}_\epsilon}
=\mat{cc}{0 & -\bB_\epsilon^\top \\\bB_\epsilon^{-1} & \bA^\top}
\mat{c}{\vec{\tfQ}_\epsilon +i\pi \vec{t}_\epsilon\\ \vec{\tfP}_\epsilon}\,.
\label{eq:QP_to_PhiPi}
\ee
Therefore,
\begin{subequations}
\be\begin{split}
-2\pi i(b^2+1)\f{\partial S_{\vec{p}_\epsilon}}{\partial \fP_{\epsilon,I}}
&=-\lb\bA_\epsilon\bB_\epsilon^\top\cdot \vec{\fP}_\epsilon\rb_I
-\lb\vec{\fQ}_{\epsilon} 
- i\pi \vec{t}_{\epsilon}\rb_I +2\pi ip_{\epsilon,I}
+\lb\bB_\epsilon\cdot\vec{P}_\epsilon\rb_I \\
&=\lb\bA_\epsilon\cdot\vec{\Phi}_\epsilon\rb_I 
-\lb\bA_\epsilon\cdot \vec{\Phi}_\epsilon\rb_I
-\lb\bB_\epsilon\cdot \vec{\Pi}_\epsilon\rb_I 
+\lb\bB_\epsilon\cdot\vec{P}_\epsilon\rb_I 
+2\pi ip_{\epsilon,I}\\
&=\left[\bB_\epsilon\cdot\lb\vec{P}_\epsilon-\vec{\Pi}_\epsilon \rb \right]_I +2\pi ip_{\epsilon,I}
\equiv 0\,,
\end{split}
\label{eq:critial_P}
\ee
\be\begin{split}
-2\pi i(b^{-2}+1)\f{\partial S_{\vec{p}_\epsilon}}{\partial \tfP_{\epsilon,I}}
&=-\lb\bA_\epsilon\bB_\epsilon^\top\cdot \vec{\tfP}_\epsilon\rb_I
-\lb \vec{\tfQ}_{\epsilon} +i \pi \vec{t}_{\epsilon}\rb_I
-2\pi i p_{\epsilon,I}
+\lb\bB_\epsilon\cdot\vec{\widetilde{P}}_\epsilon\rb_I \\
&=\lb\bA_\epsilon\cdot\vec{\widetilde{\Phi}}_\epsilon\rb_I 
- \lb\bA_\epsilon\cdot \vec{\widetilde{\Phi}}_\epsilon\rb_I
-\lb\bB_\epsilon\cdot \vec{\widetilde{\Pi}}_\epsilon\rb_I
+\lb\bB_\epsilon\cdot\vec{\widetilde{P}}_\epsilon\rb_I 
-2\pi i p_{\epsilon,I} \\
&=\left[\bB_\epsilon\cdot\lb\vec{\widetilde{P}}_\epsilon-\vec{\widetilde{\Pi}}_\epsilon \rb \right]_I -2\pi i p_{\epsilon,I} \equiv 0\,,
\end{split}
\label{eq:critial_Pt}
\ee
\label{eq:critical_P_Pt}
\end{subequations}
where we have omitted the first term $-\f{t_{\epsilon,I}}{2k}$ in \eqref{eq:critical-P} at large $k$. The critical equations \eqref{eq:critical_P_Pt} are then equivalent to the following equations. 
\be
e^{P_{\fZ_i}+\cW_i} = 1-e^{-\fZ_i}\,,\quad
e^{P_{\widetilde{\fZ}_i}+\cWt_i} = 1-e^{-\widetilde{\fZ}_i}\,,\quad
\fZ_i\in\{X_i,Y_i,Z_i\}\,,\,\,\,
\widetilde{\fZ}_i\in\{\Xt_i,\Yt_i,\Zt_i\}\,.
\label{eq:critial_simplify}
\ee
If we defined $\fZ''_i:=P_{\fZ_i}+\cW_i$ and $\widetilde{\fZ}''_i:=P_{\widetilde{\fZ}_i}+\cWt_i$, then these equations are nothing but the algebraic curve equations \eqref{eq:A-polynomial_-eform} for ideal tetrahedra. It is clear that $\vec{p}_\epsilon$ relates to different lifts from $e^{P_{\fZ_i}},e^{\cW_i}$ to the logarithmic variables $P_{\fZ_i},\cW_i$. By the procedure in \eqref{poisum1}, we have fixed the lift ambiguities of all $e^{\mathfrak{Q}_{\epsilon,I}},e^{\mathfrak{P}_{\epsilon,I}}$, so the lifts of $e^{P_{\fZ_i}},e^{\cW_i}$ have already been fixed in the integral representation of the amplitude. Therefore \eqref{eq:critial_P} and \eqref{eq:critial_Pt} uniquely determine the values of $\vec{p}_\epsilon$.

\subsubsection{position aspects}
\label{subsubsec:position}

Let us now move on to consider the derivative of $S_{0,f}$ \wrt the position variables. We first consider those \eqref{eq:critical_equations_2} \wrt the positions on the torus cusps $\{\fQ_f\}$. 
The critical equations are  
\be
\f{\partial S_{0,f}}{\partial \fQ_f}
=\f{i}{2\pi(b^2+1)}\left[\lb\fP_{-,f}-b^2\tfP_{-,f}\rb-\lb\fP_{+,f}-b^2\tfP_{+,f}\rb\right]
{+\f{i}{2\pi}\cF'_f(\fQ_f)}-u_f=0\,,\label{dSdQ}
\ee
where $\fP_{\pm,f}$ and $\tfP_{\pm,f}$ correspond to the momenta conjugate to $\fQ_f$ {and $\cF'_f(\fQ_f):=\rd\cF_f(\fQ_f)/\rd \fQ_f$}. 
The critical equation solves
\be
n_{-,f}-n_{+,f}+\f{ik}{2\pi}\cF'_f(2L_f)-ku_f=0\,,
\quad \forall f=1,\cdots,6\,.
\label{eq:stationary_Qf1}
\ee

The conjugate momenta of $\{\mathfrak{Q}_I\}_{I=1}^{10}=\{2L_{ab}\}_{(ab)}$ are $\{\mathfrak{P}_{+,I}-\mathfrak{P}_{-,I}\}_{I=1}^{10}=\{\cT_{ab}-\cT'_{ab}\}_{(ab)}$ as shown in \eqref{eq:final_T}. 10 pairs of conjugate variables $(2L_{ab},\cT_{ab}-\cT'_{ab})$ associate respectively to 6 torus cusps $f$'s and 4 annulus cusps $b$'s of $M_{+\cup -}$ (blue lines in fig.\ref{fig:glue}). $\{\cT_{ab}-\cT'_{ab}\}$ equal to six B-cycle holonomy eigenvalues on $f$'s and 4 FN twists on $b$'s calculated by the snake rule for cusp boundaries up to a constant $n \pi i$ with $n\in\Z$ (see Appendix \ref{app:twist} for details and for a generalized argument):
\be
\cT_{ab}-\cT_{ab}'=T_{ab}+\zeta_{ab},\qquad \widetilde{\cT}_{ab}-\widetilde{\cT}_{ab}'=\widetilde{T}_{ab}-\zeta_{ab},
\label{cTcTT}
\ee
where
\be
\zeta_{12}=0\,,\,\,
\zeta_{13}=\pi i\,,\,\,
\zeta_{14}=\pi i\,,\,\,
\zeta_{15}=\pi i\,,\,\,
\zeta_{23}=0\,,\,\,
\zeta_{24}=-2\pi i\,,\,\,
\zeta_{25}=0\,,\,\,
\zeta_{34}=-\pi i\,,\,\,
\zeta_{35}=0\,,\,\,
\zeta_{45}=-\pi i\,.
\label{eq:zeta0}
\ee 
Note that $(ab)$'s involving $5$ label the annuli, while others label the tori.

When we parameterize  
\be
T_f=\frac{2\pi i}{k}(-ib\nu_f -n_f)\,,\quad 
\tilde{T}_f=\frac{2\pi i}{k}(-ib^{-1}\nu_f +n_f)\,,\qquad \nu_f\in \mathbb{R},\quad n_f\in [0,k)\,,
\label{eq:Tf_nf}
\ee 
and by \eqref{cTcTT}, we can rewrite \eqref{eq:stationary_Qf1} to be
\be
-n_f-\f{ik}{2\pi}\zeta_f{+\f{ik}{2\pi}\cF'_f(2L_f)}-ku_f=0\,,
\label{nfzetaf}
\ee
where $\frac{i}{2\pi}\zeta_f\in\frac{1}{2}\mathbb{Z}$. We may set $\cF_f$ for  $f=(13),(14),(34)$ such that
\be
\frac{i}{2\pi}\left[\zeta_f-\cF'_f(2L_f)\right]\in \mathbb{Z},\qquad \ie \qquad \f{\cF'_f(2L_f)}{i\pi}\text{ is odd}\,,
\ee
and $\cF_f=0$ for $f=(12),(23),(24)$.
Absorbing this integer into $u_f$, the critical equation \eqref{eq:stationary_Qf1} becomes
\be
-n_f-ku_f=0.
\ee
The solutions of $n_f$ and $u_f\in\Z$ are both {\it unique}, since $n_{\pm,f}$ and $n_f$ have been restricted into a single period $[0,k)$:
\be
n_f=0\,,\quad u_f=0\,.
\label{nf=0}
\ee
Setting $\cF'=i\pi \mathbb{Z}$ and a vanishing constant term in $\cF$ leads to a sign factor
\be
e^{\frac{ik}{2\pi}\cF_f}=(-1)^{2j_f}
\ee
in the face amplitude for $f=(13),(14),(34)$.

\medskip 

For the remaining position variables $\{\fQ_{\epsilon,a^\times},\tfQ_{\epsilon,a^\times}\}_{a=1}^5$ defined in \eqref{eq:def_fQ_tfQ_+} and \eqref{eq:def_fQ_tfQ_-}, 
the critical equations \eqref{eq:critical_equations_3} give 
\begin{subequations}
\begin{align}
\f{\partial S_{\epsilon}^\cohe}{\partial \fQ_{\epsilon,a^\times}}
&=-\f{i\,\fP_{\epsilon,a^\times}}{2\pi(b^2+1)}
-\f{bk^{-1}}{b^2+1}\left[\re(\mu_{\epsilon,a^\times})-\f{k}{\sqrt{2}\pi}\hat{\bar{z}}_{\epsilon,a}\right]
-\f{ik^{-1}}{b^2+1}\left[m_{\epsilon,a^\times}-\f{k}{2\pi}\xh_{\epsilon,a}\right]
-\f{\yh_{\epsilon,a}+2\pi u_{\epsilon,a}}{2\pi(b^2+1)}\simeq 0\,,\\
\f{\partial S_{\epsilon}^\cohe}{\partial \tfQ_{\epsilon,a^\times}}
&=-\f{i\,\tfP_{\epsilon,a^\times}}{2\pi(b^{-2}+1)}
-\f{bk^{-1}}{b^2+1}\left[\re(\mu_{\epsilon,a^\times})-\f{k}{\sqrt{2}\pi}\hat{\bar{z}}_{\epsilon,a}\right]
+\f{ik^{-1}}{b^{-2}+1}\left[m_{\epsilon,a^\times}-\f{k}{2\pi}\xh_{\epsilon,a}\right]
+\f{\yh_{\epsilon,a}+2\pi u_{\epsilon,a}}{2\pi(b^{-2}+1)}\simeq 0\,.
\end{align}
\label{eq:eom_pos_cohe}
\end{subequations}
$u_{\epsilon,a}$ only shifts $n_{\epsilon,a}$ by multiple of $k$. By the same argument below Eq.\eqref{eq:stationary_Qf1},  $n_{\epsilon,a}\in[-\delta ,k-\delta]$ fixed in \eqref{poisum1} uniquely determines $u_{\epsilon,a}=0$. 
Use the notations in \eqref{eq:cX_param} and \eqref{eq:cY_param}, and recall the notation \eqref{eq:rho_notation} for $(\zh_{\epsilon,a},\xh_{\epsilon,a},\yh_{\epsilon,a})$, the solution is given by  
\begin{subequations}
\begin{align}
&\re(\mu_{a})=\f{k}{\sqrt{2}\pi}\re(\zh_{a})\,,\quad
\re(\nu_{a})=-\f{k}{\sqrt{2}\pi}\im(\zh_{a})\,,\quad
m_{a}=\f{k}{2\pi}\xh_{a}\,,\quad
n_{a}=-\f{k}{2\pi}\yh_{a}\,,\quad
\forall a=1,\cdots, 5\,,
\label{eq:eom_position_1}\\
&\re(\mu'_{a})=-\f{k}{\sqrt{2}\pi}\re(\zh_{a})\,,\quad
\re(\nu'_{a})=-\f{k }{\sqrt{2}\pi}\im(\zh_{a})\,,\quad
m'_{a}=-\f{k}{2\pi}\xh_{a}\,,\quad
n'_{a}=-\f{k}{2\pi}\yh_{a}\,,\quad 
\forall a=1,\cdots, 4\,,
\label{eq:eom_position_2}\\
&\re(\mu'_{5})=-\f{k}{\sqrt{2}\pi}\re(\zh'_{5})\,,\quad
\re(\nu'_{5})=-\f{k}{\sqrt{2}\pi}\im(\zh'_{5})\,,\quad
m'_{5}=-\f{k}{2\pi}\xh'_{5}\,,\quad
n'_{5}=-\f{k}{2\pi}\yh'_{5}\,.
\label{eq:eom_position_3}
\end{align}
\label{eq:critial_spheres}
\end{subequations}
{The critical points \eqref{eq:eom_position_1} -- \eqref{eq:eom_position_2} immediately reproduce the gluing constraints on the position parameters and also match the momentum parameters on $\{\cS_a\}_{a=1}^4$ from $M_+$ and $M_-$:
\be
\re(\mu'_a)=-\re(\mu_a)\,,\quad m_a'=-m_a\,,\quad \re(\nu'_a)=\re(\nu_a)\,,\quad n'_a=n_a\,,\quad
a=1,\cdots,4\,.
\ee}
{Note that the coherent state $\bar{\xi}_{(x,y)}(m'_a)$ is invariant under a shift $m'_a\rightarrow m'_a-k$, which corresponds to shifting $\fQ_{-,a^\times}\rightarrow\fQ_{-,a^\times}+2\pi i$ while $\tfQ_{-,a^\times}\rightarrow\tfQ_{-,a^\times}-2\pi i
$. Performing this shift for $\{m'_a\}_{a=1}^4$ would shift the solution to $m'_a$ in \eqref{eq:eom_position_2} to $m'_a=k-\f{k}{2\pi }x'_a$ hence critical solutions $m_a$ and $m'_a=k-m_a$ can both be taken to be in the range of $[-\delta,k-\delta]$ so that the Poisson resummation \eqref{poisum1} can be applied with no ambiguity. }

\medskip

To summarize, the above discussion shows that a part of the critical equations recovers the algebraic curve equation for ideal tetrahedra under the octahedron constraints and recovers the gluing constraints between $M_+$ and $M_-$. These critical equations indicate that the critical points of the amplitude are $\mathrm{SL}(2,\mathbb{C})$ flat connections on $M_{+\cup-}$ satisfying the simplicity constraints. Moreover, because of the sum over $j_f$, the variation of $\mathfrak{Q}_f$ imposes an addition constraint \eqref{nf=0} to the flat connection. This additional constraint is an analog of the ``flatness constraint'' \cite{Han:2013hna,Bonzom:2009hw,Hellmann:2013gva} in the EPRL-KF spinfoam model for the following reason. $m_f$ encodes the area of an internal triangle dual to the spinfoam face $f$. Its conjugate variable $n_f$ then encodes the deficit angle around this triangle. The solution to $n_f$ is interpreted as bulk simplices being glued such that the bulk curvature is a constant. Given $m_f$, the solution to $n_f$ is unique. It is a feature different from the case in EPRL-FK model, where there may be infinitely many critical solutions to the deficit angles separated by $4\pi \Z$. This ambiguity seems to able to be resolved by adding a non-vanishing cosmological constant, from the experience of the spinfoam model we study in this paper.

\subsection{Amplitude at the critical points}
\label{subsec:critical action}

At the critical points \eqref{eq:critial_spheres} solved from the derivative \wrt the position variables of the actions $S_{\zh_{\epsilon,a}}$ and $S_{(\xh_{\epsilon,a},\yh_{\epsilon,a})}$ \eqref{eq:coherent_QQt}, we obtain the critical actions 
\be
S^0_{\zh_{\epsilon,a}}=\f{i}\pi\re(\zh_{\epsilon,a})\im(\zh_{\epsilon,a})\,,\quad
S^0_{(\xh_{\epsilon,a},\yh_{\epsilon,a})}=\f{i}{4\pi}\xh_{\epsilon,a}\yh_{\epsilon,a}\,,
\label{eq:critial_action_coherent}
\ee
which sum to zero when considering both $\epsilon=\pm$ by definition \eqref{eq:rho_notation}.

On the other hand, $e^{k S_{\vec{p}_\epsilon}}$ is a pure phase at large $k$ for the following reasons. 
Firstly, the imaginary parts $\im(\mu)$'s for all $\mu$'s are not seen at large $k$ and $b^{-1}$ is the complex conjugate of $b$ by definition. Therefore, $\tilde{\fz}_i^\epsilon$ is the complex conjugate of $\fz_i^{\epsilon}$ for $\fz_i^\epsilon\in\{x_i^\epsilon,y_i^\epsilon,z_i^\epsilon,w_i^\epsilon\}$ in $\Oct(i)$. We then conclude that the sum $S^\epsilon_1+\tS^\epsilon_1$ is pure imaginary for both $\epsilon=\pm$ from the expressions \eqref{eq:S1} and \eqref{eq:S1t}. For the rest of $S_{\vec{p}_\epsilon}$, we rewrite them as 
\begin{multline}
k S_0^\epsilon(\vec{\mu}_\epsilon,\vec{\nu}_\epsilon,\vec{m}_\epsilon,\vec{n}_\epsilon)
-2\pi i \vec{p}_\epsilon\cdot\vec{n}_\epsilon \\
=\f{\pi i}{k}\left[-2\left(\vec{\mu}_\epsilon-\frac{i Q}{2} \vec{t}_\epsilon\right) \cdot \vec{\nu}_\epsilon+2 \vec{m}_\epsilon \cdot \vec{n}_\epsilon-\vec{\nu}_\epsilon \cdot \bA_\epsilon \bB_\epsilon^\top \cdot \vec{\nu}_\epsilon+ \vec{n}_\epsilon \cdot \bA_\epsilon \bB_\epsilon^\top \cdot \vec{n}_\epsilon+k\vec{n}_\epsilon\cdot( \vec{t}_\epsilon+
2\vec{p}_\epsilon)\right]\,.
\label{eq:S0_old_coord}
\end{multline}
As the imaginary parts of $\vec{\mu}_\epsilon$ and $\vec{\nu}_\epsilon$ do not scale with $k$, \eqref{eq:S0_old_coord} is also imaginary at large $k$. Therefore, $e^{k S_{\vec{p}_\epsilon}}$ contributes to the amplitude only a phase at the critical points. 

Lastly, 
$e^{\f{ik}{2\pi}\cF_f}$ only contribute a sign to the total amplitude.

The above stationary phase analysis is carried out for all integrals except the integrals of $\hat{\rho}_a$. This means we study $\cA_{\heta_5,\htrho_5}(\{\hrho_a\})$ with $\{\hrho_a\}$ as parameters and we can write 
\be
\cZ_{\heta_5,\htrho_5}=\int_{\overline{\cM}_{\vec{m}}}[\rd\hrho_a]\, \cA_{\heta_5,\htrho_5}(\{\hrho_a\})\,.
\ee
Note that we can interchange the order of integrations since $\cZ_{\heta_5,\htrho_5}$ is absolutely convergent. The above analysis assumes the existence of critical point(s) at certain $\{\hrho_a=\hrho_a^{(0)}\}$. At $\{\hrho_a^{(0)}\}$, we have the purely imaginary critical action being $S_{\rm tot}$ evaluated at the critical point. We denote the critical point by $\alpha$ and the critical action by $ S_{\tot}^{\alpha}$. The critical action $S_{\tot}^{\alpha}\in i\mathbb{R}$ is scaleless in $k$. Each $\alpha$ is associated with a unique $\hrho_a^{(0)}$. The asymptotic of $\cA_{\heta_5,\htrho_5}$ at $\hrho_a^{(0)}$ is 
\be
\cA_{\heta_5,\htrho_5}(\{\hrho^{(0)}_a\})=\sum_{\alpha\  \text{associated with } \{\hrho^{(0)}_a\}}\f{\cN}{\sqrt{\det(-H_\alpha/(2\pi))}}
e^{k S^{\alpha}_{\tot}}[1+O(1/k)]
\label{asympform}
\ee
where {$\left.H_\alpha=\partial^2\lb k S_{\tot}\rb\right|_\alpha$} is the Hessian matrix evaluated at the critical point $\alpha$ and $\cN$ is given in \eqref{eq:cN}. We have also assumed that $\alpha$ are isolated and $H_\alpha$ are nondegenerate. Then the sum of $\alpha$ is finite because all critical equations are polynomial equations of certain degree in terms of exponential coordinates $e^{\mathfrak{P}},e^{\mathfrak{Q}},e^{\tilde{\mathfrak{P}}},e^{\tilde{\mathfrak{Q}}}$. Other situations are going to be discussed in a moment in Section \ref{subsec:Hessian}. We have removed the summations for $\vec{p}_\pm,\vec{u}_\pm,\vec{u}_f$ which come from the Poisson resummations, because at the stationary points, the following conditions must be satisfied.
\be
\vec{u}_\pm=\vec{0}\,,\quad
\vec{u}_f\,,\vec{p}_\pm \ \text{are unique}\,.
\ee
The conditions pick up only one term in the sums of $\vec{p}_\pm,\vec{u}_\pm,\vec{u}_f$.

However, it is generally possible that for some $\hat{\rho}_a$, the critical point does not exist in the integration domain. In this case, the asymptotics becomes
\be
\cA_{\heta_5,\htrho_5}(\{\hrho_a\})=O(1/k^N),\qquad \forall\ N>0\,,
\ee
\emph{i.e.} it suppresses faster than any polynomial of $k^{-1}$. Then we can generalize the formula \eqref{asympform} for $\{\hrho_a\}$ in a neighborhood of $\{\hrho_a^{(0)}\}$ \cite{hormander2015analysis}
\be
\cA_{\heta_5,\htrho_5}(\{\hrho_a\})=\sum_{\alpha\  \text{associated with}\ \{\hrho^{(0)}_a\}}\f{\cN}{\sqrt{\det(-H_\alpha/(2\pi))}}
e^{kS^{\alpha}_{\tot}}\cB^{\alpha}(\{\hrho_a\})[1+O(1/k)]\,,
\label{asympform1}
\ee
where $\cB^{\alpha}$ satisfies that $\cB^{\alpha}=1$ at $\{\hrho_a=\hrho^{(0)}_a\}$ and of $O(k^{-N})$ for any $N>0$ elsewhere.  $\cB^{\alpha}$ is smooth and bounded on $\overline{\cM}_{\vec {m}}$. 

The asymptotics of $\cZ_{\heta_5,\htrho_5}$ can be expressed as 
\be
\cZ_{\heta_5,\htrho_5}(\vec{\mu}|\vec{m})= \sum_\alpha \f{\cN}{\sqrt{\det(-H_\alpha/2\pi)}}
e^{kS^{\alpha}_{\tot}}{\int_{\overline{\cM}_{\vec{m}}}[\rd\hrho_a]}\, \cB^{\alpha}(\{\hrho_a\})[1+O(1/k)]\,.
\ee
The $\hrho_a$-integrals are dominated by the contributions from the neighbourhoods of $\{\hrho_a^{(0)}\}$'s and is bounded. Here we remind that
\be
\cN=\f{512\, k^{59+6\mu}}{(2\pi)^{102+6\mu}Q^{40}}\,.
\ee

Note that the above formula clearly assumes that the critical point $\alpha$ exists at some $\{\hrho_a^{(0)}\}$. 

\subsection{The Hessian matrix}
\label{subsec:Hessian}
To obtain the total scaling of $\cZ_{\heta_5,\htrho_5}(\vec{\mu}|\vec{m})$ with $k$, we are left to calculate the scaling of the Hessian $H_\alpha$ at the critical points. 
Let us first determine the dimension of the Hessian matrix. This is given by the number of integration variables in the expression \eqref{eq:total_amplitude_action} of the amplitude after imposing all the simplicity constraints. These variables are summarized as follows. 
\be
\ba{ll}
\{\fQ_f\}_{f=1}^6=\{2L_{12},2L_{13},2L_{14},2L_{23},2L_{24},2L_{34}\}\,,\quad
\\[0.15cm]
\{\fQ_{+,a^\times}\}_{a=1}^5=\{\cX_1,\cX_2,\cX_3,\cX_4,\cX_5\}\,,\quad &
\{\tfQ_{+,a^\times}\}_{a=1}^5=\{\tilde{\cX}_1,\tilde{\cX}_2,\tilde{\cX}_3,\tilde{\cX}_4,\tilde{\cX}_5\}\,,\\[0.15cm]
\{\fQ_{-,a^\times}\}_{a=1}^5=\{\cX'_1,\cX'_2,\cX'_3,\cX'_4,\cX'_5\}\,,\quad &
{\{\tfQ_{-,a^\times}\}_{a=1}^5=\{\tilde{\cX}'_1,\tilde{\cX}'_2,\tilde{\cX}'_3,\tilde{\cX}'_4,\tilde{\cX}'_5\}}\,,\\[0.15cm]
\{\fP_{+,I}\}_{I=1}^{15}=\{\{\cT_{ab}\}_{(ab)},\{\cY_a\}_{a=1}^{5}\}\,,\quad &
\{\tfP_{+,I}\}_{I=1}^{15}=\{\{\widetilde{\cT}_{ab}\}_{(ab)},\{\tilde{\cY}_a\}_{a=1}^{5}\}\,, \\[0.15cm]
{\{\fP_{-,I}\}_{I=1}^{15}=\{\{\cT'_{ab}\}_{(ab)},\{\cY'_a\}_{a=1}^{5}\}}\,,\quad &
{\{\tfP_{-,I}\}_{I=1}^{15}=\{\{\widetilde{\cT}'_{ab}\}_{(ab)},\{\tilde{\cY}_a'\}_{a=1}^{5}\}} \,.
\ea
\label{eq:all_variables}
\ee
Therefore, $H_\alpha$ is an {$86\times 86$} matrix. The entries are given by the second derivatives of the effective action $S_\tot$ with the simplicity constraints imposed. We now calculate the Hessian entries. 

\medskip
\noindent{\it Second derivatives \wrt momenta. }
\medskip

Firstly, consider the second derivatives of $S_{\tot}$ \wrt $\{\fP_{\epsilon,I},\tfP_{\epsilon,I}\}$. The nontrivial results are all from the action $S_{\vec{p}_\epsilon}$. Since $S_{\vec{p}_+}$ and $S_{\vec{p}-}$ are not entangled, we can consider them separately. From \eqref{eq:S0-S1t}, we get, for all $I,J=1,\cdots,15$, 
\begin{subequations}
\begin{align}
\f{\partial^2 S_{\vec{p}_\epsilon}}{\partial \fP_{\epsilon,I}\partial \fP_{\epsilon,J}}
&= \f{i}{2\pi(b^2+1)}\lb \bB_{\epsilon}\cdot\f{\partial \vec{P}_{\epsilon}}{\partial \vec{\fP}_{\epsilon}}-\bA_\epsilon\bB_\epsilon^\top \rb_{IJ}\,,\quad
\f{\partial^2 S_{\vec{p}_\epsilon}}{\partial \tfP_{\epsilon,I}\partial \tfP_{\epsilon,J}}
= \f{i}{2\pi(b^{-2}+1)}\lb \bB_{\epsilon}\cdot\f{\partial \vec{\widetilde{P}}_{\epsilon}}{\partial \vec{\tfP}_{\epsilon}}-\bA_\epsilon\bB_\epsilon^\top \rb_{IJ}\,,\\
\f{\partial^2 S_{\vec{p}_\epsilon}}{\partial \fP_{\epsilon,I}\partial \tfP_{\epsilon,J}}
&=0\,,
\end{align}
\label{eq:second_derivative_momentum}
\end{subequations}
where $\f{\partial \vec{P}_{\epsilon}}{\partial \vec{\fP}_{\epsilon}}$ (\resp $\f{\partial \vec{\widetilde{P}}_{\epsilon}}{\partial \vec{\tfP}_{\epsilon}}$) is a block diagonal matrix in terms of $\vec{\Phi}_\epsilon$ (\resp $\vec{\tilde{\Phi}}$) and $\vec{P}_\epsilon\,,\,\vec{\widetilde{P}}_\epsilon$ are defined in \eqref{eq:def_P}. Explicitly,
\be
\f{\partial \vec{P}_{\epsilon}}{\partial \vec{\fP}_{\epsilon}}(\{X_i,Y_i,Z_i\})
=-\diag\lb\{{\bf E}_i\}\rb\cdot\bB_\epsilon^{\top}\,,\quad
\f{\partial \vec{\widetilde{P}}_{\epsilon}}{\partial \vec{\tfP}_{\epsilon}}(\Xt_i,\Yt_i,\Zt_i)
=-\diag\lb\{\widetilde{\bf E}_i\}\rb\cdot\bB_\epsilon^{\top}\,,\quad
i=\begin{cases}
1,\cdots,5\,\text{ if }\epsilon=+\\
6,\cdots,10\,\text{ if }\epsilon=-
\end{cases}\,,
 \label{eq:pPipP}
\ee
where
\begin{subequations}
\begin{align}
{\bf E}_i
&=\diag\lb\f{e^{-X_i}}{1-e^{-X_i}},\f{e^{-Y_i}}{1-e^{-Y_i}},\f{e^{-Z_i}}
{1-e^{-Z_i}}\rb +\f{e^{X_i+Y_i+Z_i}}{1-e^{X_i+Y_i+Z_i}}\mat{ccc}{1&1&1\\1&1&1\\1&1&1}\,,\\
\widetilde{\bf E}_i
&=\diag\lb\f{e^{-\Xt_i}}{1-e^{-\Xt_i}},\f{e^{-\Yt_i}}{1-e^{-\Yt_i}},\f{e^{-\Zt_i}}
{1-e^{-\Zt_i}}\rb +\f{e^{\Xt_i+\Yt_i+\Zt_i}}{1-e^{\Xt_i+\Yt_i+\Zt_i}}\mat{ccc}{1&1&1\\1&1&1\\1&1&1}\,.
\end{align}
\end{subequations}

\medskip
\noindent{\it Second derivatives \wrt positions. }
\medskip

Secondly, we calculate the second derivatives of $S_\tot$ \wrt the position variables $\{\fQ_{\epsilon,a^\times},\tfQ_{\epsilon,a^\times}\}$. (There are no nontrivial second derivatives \wrt $\{\fQ_f\}$.) The nontrivial results are all from the action $S_{\hrho_{\epsilon,a}} = S_{\zh_{\epsilon,a}}+S_{(\xh_{\epsilon,a},\yh_{\epsilon,a})}$ and there are no entanglement between different $a$'s or different $\epsilon$'s. From the definitions \eqref{eq:coherent_QQt},   
\be
\f{\partial^2 S_{\hrho_{\epsilon,a}}}{\partial \fQ_{\epsilon,a^\times}^2}
=\f{1-b^2}{2\pi(b^2+1)^2}\,,\quad
\f{\partial^2 S_{\hrho_{\epsilon,a}}}{\partial \tfQ_{\epsilon,a^\times}^2}
=\f{b^2(b^2-1)}{2\pi(b^2+1)^2}\,,\quad
\f{\partial^2 S_{\hrho_{\epsilon,a}}}{\partial \fQ_{\epsilon,a}\partial \tfQ_{\epsilon,a^\times}}
=-\f{b^2}{\pi(b^2+1)^2}\,,\quad \forall a=1,\cdots,5\,.
\ee

\medskip
\noindent{\it Second derivatives \wrt positions and momenta. }
\medskip

There are cross-terms of positions and momenta in $S_0^\epsilon$. Therefore, they contribute nontrivial second derivatives of $S_\tot$ \wrt both the momentum variables $\{\fP_{\epsilon,I},\tfP_{\epsilon,I}\}$ and the position variables $\{\fQ_f\}\,,\,\{\fQ_{\epsilon,a^\times},\tfQ_{\epsilon,a^\times}\}$. 
The momenta corresponding to the internal tori are the elements with indices $I=1,2,3,5,6,8$ while others correspond to the boundary annuli, hence it is convenient to define indices $f^\times=\{1,2,3,5,6,8\}$ one-to-one corresponding to the indices $f=\{1,2,3,4,5,6\}$. 
There is no correlation between $\epsilon=+$ and $\epsilon=-$ hence we consider them separately:
\begin{subequations}
\begin{align}
&\f{\partial^2 S_0^\epsilon}{\partial\fP_{\epsilon,I}\partial \fQ_f}=\delta_{I,f^\times}\f{-i\epsilon }{2\pi(b^2+1)}\,,\quad 
\f{\partial^2 S_0^\epsilon}{\partial\tfP_{\epsilon,I}\partial \fQ_f}=\delta_{I,f^\times}\f{i\epsilon }{2\pi(b^{-2}+1)}\,,\\
&\f{\partial^2 S_0^\epsilon}{\partial\fP_{\epsilon,I}\partial \fQ_{\epsilon,a^\times}}=\delta_{I,a^\times}\f{i}{2\pi(b^2+1)}\,,\quad 
\f{\partial^2 S_0^\epsilon}{\partial\tfP_{\epsilon,I}\partial \tfQ_{\epsilon,a^\times}}=\delta_{I,a^\times}\f{i
}{2\pi(b^{-2}+1)}\,.
\end{align}
\label{eq:2derivativePQ}
\end{subequations}

 We observe that all the non-zero second derivatives \eqref{eq:second_derivative_momentum}--\eqref{eq:2derivativePQ} are scaleless with $k$. Assume that the Hessian is non-degenerate, then 
  the power of its determinant in $k$ must be the same as the dimension of the integration, which is the sum of 30 momenta $\{\fP_{\pm,I},\tfP_{\pm,I}\}_{I=1}^{15}$ and 10 positions $\{\fQ_{\pm,a^\times},\tfQ_{\pm,a^\times}\}_{a=1}^{5}$ from both vertex amplitudes 
as well as 6 position variables $\{\fQ_f\}_{f=1}^6$ from the face amplitudes. That is,
 \be
{\det(H_\alpha)\propto k^{86}}\,.
\ee

Combining the power of $k$ in $\cN$ \eqref{eq:cN}, we conclude that 
\be
\cZ_{\heta_5,\htrho_5}= k^{16+6\mu} \sum_\alpha C_\alpha
e^{kS^{\alpha}_{\tot}}{\int_{\overline{\cM}_{\vec{m}}}[\rd\hrho_a]}\, \cB^{\alpha}(\{\hrho_a\})[1+O(1/k)]\,,
\label{eq:final_scaling}
\ee
where $C_\alpha$ does not scale with $k$

 The result \eqref{eq:final_scaling} is based on the assumption of the nondegeneracy of the Hessian $H_\alpha$. In the case of $\det(H_\alpha)=0$ for some $\alpha$, one needs to separate the part of the integral corresponding to the degenerate directions and only applies the stationary phase analysis to the rest of the integral, where the Hessian is nondegenerate. Let us first consider a simpler case when there exists a degenerate critical point $\alpha$ associated with $\{\hrho_a^{(0)}\}$, whereas the degeneracy is caused by continuously many critical points in a neighbourhood of $\alpha$. From the argument of the geometrical interpretation, as will be discussed in Section \ref{app:geometry}, this case happens in our model. Let us focus on a neighbourhood $U_\alpha$ of $\alpha$ in the integration domain. We change the integration variables into two subsets $(\vec{x},\vec{t})$ where $\vec{t}=(t_1,\cdots,t_d)$ satisfying $\vec{t}|_\alpha=\vec{0}$ corresponds to all the degenerate directions of Hessian, \ie $\partial_{t_i}\partial_{t_j}S_{\rm tot}|_\alpha=0$ and the submatrix $h_\alpha$, whose entries are $(h_{\alpha})_{ij}\equiv \partial_{x_i}\partial_{x_j}S_{\rm tot}|_\alpha$, is nondegenerate. By the assumption that there are continuously many critical points in the $\vec{t}$-directions, we have $\partial_{t_i} S_{\rm tot}(\alpha +\vec{t})=0\,,\forall i=1,\cdots,d$ so $S_{\rm tot}(\alpha +\vec{t})=S_{\rm tot}(\alpha)=S_{\rm tot}^{(\alpha)}$ is constant on the submanifold $V_\alpha$ of $\vec{t}$ in $U_\alpha$. Performing the stationary phase approximation for the $\vec{x}$-integral results in that the contribution from $U_\alpha$ to the partial amplitude $\cA_{\heta_5,\htrho_5}(\{\hrho^{(0)}_a\})$ is given by the following integral:  
\be
\cN e^{kS_{\rm tot}^{\alpha}}\int_{V_\alpha} \frac{ \mathrm{d}^d t}{\sqrt{\det(-h_{\alpha}/2\pi)}}\left[1+O(k^{-1})\right]\sim k^{16+6\mu+d/2}\int_{V_\alpha}\mathrm{d}^d t \, u_\alpha(\vec{t})\left[1+O(k^{-1})\right]\,,
\label{degecritpt}
\ee
where $u_\alpha(\vec{t})$ does not scale with $k$. The last integral must be finite since the amplitude is finite at any $k$. As a result, the degeneracy of the critical point increases the exponent of $k$ by $d/2$:
\be
\cZ_{\htrho_5,\heta_5}(\vec{\mu}|\vec{m})\sim  k^{16+6\mu+d/2}\,e^{kS_{\tot}^{\alpha}}\int_{\overline{\cM}_{\vec{m}}}[\rd\hrho_a]\, \cB^{\alpha}(\{\hrho_a\}),
\label{eq:final_scaling1}
\ee
where $d$ is the maximal number of degenerate directions in the set of critical points. We discuss the possible existence of such degenerate directions from the geometrical point of view in Section \ref{app:geometry}.

\section{Integration over coherent state labels}
\label{sec:coherent_integral}

We now discuss the integration of the coherent state labels satisfying the simplicity constraints $\int_{\overline{\cM}_{\vec{m}}}[\rd\hrho_a]\,\cB^{\alpha}(\{\hrho_a\})$. 
The integral of the coherent state label $\hat{\rho}_a=(\htheta_a,\hphi_a)$ is over the subdomain $\overline{\cM}_{\vec{m}}\subset [0,\pi]\times[0,\pi]$ constrained by the triangle inequality. Recall that fixing the boundary configuration $m_{a5}$ at large $k$ leads to $L_{a5}=O(k^{-1})$ for all annulus connections on $\cS_5$ and $\cS_6$. It results in that one of the holes of each $\cS_a$, $a=1,\cdots,4$, has the monodromy nearly trivial up to $O(k^{-1})$, \ie one of $\{a_i\}_{i=1}^4$ in the triangle inequality \eqref{eq:range_theta_phi} is of $O(k^{-1})$. Let $a_1=c k^{-1}$ without loss of generality. It indicates that 
\be
a_2-c k^{-1}\leq \htheta\leq a_2+c k^{-1}\,,
\label{a2thetaineq}
\ee
where $c=2\pi j_1$ for certain fixed $j_1$. It means that $\overline{\cM}_{\vec{m}}$ is very narrow in $\htheta$-direction and shrinks to measure-zero as $k\to\infty$. The $\hphi$-integral is not constrained. The integration domain $\overline{\cM}_{\vec{m}}$ depends on $k$, so the $\hrho_a$-integrals cannot be studied by using stationary phase analysis. 

Let us first consider the $\htheta_a$-integrals. The critical point only associates with $\htheta=a_2$ in \eqref{a2thetaineq}. The integral is confined to an arbitrarily small neighbourhood of $\htheta_a$. Recall that $\cB^{\alpha}$ is smooth in the neighbourhood, so, by the mean value theorem,
\be
\cB^{\alpha}(\{\htheta_a,\hphi_a\})=\cB^{\alpha}(\{\htheta^{(0)}_a,\hphi_a\})+\sum_{a}\left[\theta_a-\theta^{(0)}_a\right]F_a(\{\theta_a\})=\cB^{\alpha}(\{\htheta^{(0)}_a,\hphi_a\})+O(k^{-1})\,,
\ee
where $\htheta^{(0)}_a$ is the one associating to the critical point ($a_2$ in \eqref{a2thetaineq}), and $F_a(\{\theta_a\})$ equals to the derivative $\partial\cB^{\alpha}/\partial{\htheta_a}$ evaluated in the interval $[\theta_a,\theta_a^{(0)}]$. Therefore, the $\htheta_a$-integrals behave as
\be
\int_{ \theta_1^{(0)}-c_1k^{-1}}^{\theta_1^{(0)}+c_1k^{-1}}\rd \htheta_1\cdots\int_{ \theta_4^{(0)}-c_4k^{-1}}^{\theta_4^{(0)}+c_4k^{-1}}\rd \htheta_4\, \cB^{\alpha}(\{\htheta_a,\hphi_a\})=Ck^{-4}\cB^{\alpha}(\{\htheta^{(0)}_a,\hphi_a\})[1+O(k^{-1})] \,.
\ee
where $C=2^4\prod_{a=1}^4 c_a$.

When $a_1\to 0$ and $\hat{\theta}$ is fixed to be $a_2$, all $\hphi\in[0,\pi]$ correspond to the same $\SU(2)$ flat connection on a 3-holed sphere (see the discussion in Section \ref{subsubsec:second_simplicity}). Therefore, the flat connections on all $\cS_a$'s and thus the flat connection on $M'_{+\cup -}$ are independent of $\hat{\phi}_a$. In other words, the critical point exists for any $\hphi_a$. Then the situation is the same as the case of continuously many critical points. The $\hphi_a$-integrals $\int [\rd\hphi_a]\,  \cB^{\alpha}(\{\htheta^{(0)}_a,\hphi_a\})$ does not scale with $k$.

We conclude that the integral of $(\htheta_a,\hphi_a)_{a=1}^4$ contributes $k^{-4}$ scaling. Inserting this result to \eqref{eq:final_scaling1} gives
\be
\cZ_{\htrho_5,\heta_5}(\vec{\mu}|\vec{m})\sim {k^{12+6\mu+d/2}}\, e^{k S_{\rm tot}^\alpha}\,.
\label{eq:final_scaling_final}
\ee
As $k\to\infty$, the amplitude for $M_{+\cup -}$ diverges when {$\mu>-2$}. This conclusion is drawn with the assumption that the degeneracy of the critical points taken into account above is maximal. If any additional degeneracy exists, the power of the amplitude in $k$ may increase as discussed in Section \ref{subsec:Hessian}.

\section{Geometrical interpretation of critical points}
\label{app:geometry}

In this section,
we discuss the existence and properties of the critical points from the geometrical point of view. It is complementary to the possible degeneracy of the Hessian matrix discussed in Section \ref{subsec:Hessian}. 

The critical points of the partial amplitude $\cA_{\heta_5,\htrho_5}(\{\hrho^{(0)}_a\})$ are framed $\mathrm{SL}(2,\mathbb{C})$ flat connections in the open patch defined by the triangulation of $M_{+\cup -}$. A framed flat connection in the patch is an irreducible flat connection \footnote{A flat connection is called
irreducible if the only elements of $\mathrm{SL}(2,\mathbb{C})$ that commute with all holonomies are the central elements $\pm\id$.} together with a choice of flat section of an associated flag bundle at every cusps boundary \cite{Dimofte:2013iv}. A generic choice of flat section always exists for every torus or annulus cusp, since the fundamental group is abelian. Therefore, finding a framed flat connection boils down to finding an irreducible flat connection on $M_{+\cup -}$. The closure of the open patch covers all framed flat connections (with fixed boundary triangulation) when the 3D ideal triangulation is sufficiently refined, which we assume to be true. 

Recall that at large $k$, the boundary configurations $\{\f{j_b}{k}\}\equiv\{\f{j_{a5}}{k}=\f{j_{a6}}{k}\}\rightarrow 0\,,\forall a=1,\cdots,4$
as $j_{a5}=j_{a6}$ are kept fixed. This means that the eigenvalues of holonomies around all holes of $\cS_5$ and $\cS_6$ equal 1 up to $O(1/k)$ correction. Resulting from this, the holes of $\{\cS_a\}_{a=1}^4$ connected to $\cS_5$ by annuli (which are the same ones connected to $\cS_6$ as can be seen from fig.\ref{fig:glue}) all have approximately trivial holonomy eigenvalues when the connection is flat. It implies that the holonomy around one hole of each $\cS_a (a=1,\cdots,4)$ becomes approximately 
\be\label{nipotent}
\begin{pmatrix}
1 & z \\
0 & 1 
\end{pmatrix}
\ee
for some $z\in\mathbb{C}$. However, the simplicity constraints restrict that the flat connection on every $\cS_a$ is SU(2), so $z\equiv 0$. Therefore, the holonomy around any annulus cusps connecting $\cS_5$ and $\cS_6$ is trivial, since it relates \eqref{nipotent} by conjugation. The critical point can then be approximated by a flat connection with trivial holonomies around all annulus cusps (connecting $\cS_5$ and $\cS_6$). The error of the approximation is of $O(1/k)$. In contrast, for the internal torus cusps, $\{\f{j_f}{k}\}\equiv\{\f{j_{ab}}{k}\}_{(ab)=(12),(13),(14),(23),(24),(34)}$ are finite at the critical points, so their A-cycle holonomies are nontrivial.

In this approximation, we can remove the annulus cusps with trivial holonomies, and we remove the boundaries $\cS_5$ and $\cS_6$ since they have no holes so the flat connections on them are trivial. Then on each side, $M_\pm$ effectively becomes a different graph complement of $S^3$: $M_\pm'=S^3\backslash \Gamma_4$, where $\Gamma_4$ is the tetrahedron graph. The fundamental group $\pi_1(S^3\backslash \Gamma_4)$ of $M'_\pm$ is generated by 
\be
{\left\{\ell_{ab} \Bigg|\ a,b=1,\cdots,4\,, a\neq b ;\ \ell_{ab}=\ell_{ba}^{-1};\ \prod_{b}\ell_{ab}=1\,,\forall a\right\}}\,,
\ee
where each $\ell_{ab}$ is a closed loop in $S^3\backslash \Gamma_4$ around an annulus cusp connecting the holes of $\mathcal{S}_a$ and $\cS_b$. Given a Lie group $G$, the irreducible $G$ flat connection on $S^3\backslash \Gamma_4$ is an irreducible $G$-representation of $\pi_1(S^3\backslash \Gamma_4)$ modulo conjugation. The flat connection on $M_\pm'$ can be identified as a flat connection on the original $M_\pm$, simply by adding trivial holonomies as the representatives of the loops around the annulus connecting $\cS_5$ or $\cS_6$.

Here we are not aiming at a full classification of the critical points but showing their existence. We focus on the $\SU(2)$ flat connections, which at least cover a subset of critical points.

\begin{lemma}\label{bijectionV1}

There is a bijection from the set of 4 points on the unit $S^3$ modulo the $\SU(2)$ left and right translations to $\SU(2)$ flat connections on $S^3\backslash \Gamma_4$.

\end{lemma}

\begin{proof}
Given 4 points on $S^3$ denoted by $v_1,v_2,v_3,v_4$. Each $v_i$ can be uniquely represented by an SU(2) matrix. 
We define $h_{ij}=v_iv_j^{-1}$ for any pair of $i,j=1,\cdots,4$ ($i\neq j$). The set of $h_{ij}$ satisfies $h_{ij}h_{jk}h_{ki}=1$. The data $\{h_{ij}\}$ defines a representation of the generators in $\pi_1(S^3\backslash\Gamma_4)$ by the relation $h_{ij}=h(\ell_{ab})$ {for some $i,j,a,b\in\{1,2,3,4\}$}. The left translation acting on $v_1,v_2,v_3,v_4$ gives the conjugation $h_{ij}\to gh_{ij}g^{-1}$ for any $g\in\mathrm{SU(2)}$, and the right translation leaves $h_{ij}$ invariant. 

By the right translation, we fix $v_1=1$ then we have 1-to-1 correspondence between $\{v_1,v_2,v_3,v_4\}$ and $\{h_{1j}=v_j^{-1}\}$, while other $h_{ij}$ are completely fixed by $\{h_{1j}\}$ via $h_{1i}h_{ij}h_{j1}=1$. Therefore, modulo the left and right translation, the map from the set of 4 points on $S^3$ to SU(2) flat connections on $S^3\backslash \Gamma_4$ is bijective.
\end{proof}

The flat connection on $M_\pm'$ maps each $\ell_{ab}$ to the SU(2) holonomy $h(\ell_{ab})=h_{ij}$ with {$i,j,a, b\in\{1,2,3,4\}$}. The flat connection is irreducible for generic 4 points. The eigenvalue of $h(\ell_{ab})$ is the FN variables $\lambda_{ab}=e^{\frac{2\pi i}{k}m_{ab}}$ associated to the annulus cusp connecting $\cS_a$ and $\cS_b$. If we write $\mathrm{Tr}\lb h(\ell_{ab})\rb=2\cos \theta_{ij}$, $\theta_{ij}\in[0,\pi]$, then $\theta_{ij}$ is the geodesic length connecting $v_i$ and $v_j$ on $S^3$. To see this, we use the relation that the inner product $\langle X,Y\rangle=-\frac{1}{2}[\mathrm{Tr}(XY)-\mathrm{Tr}X\mathrm{Tr}Y]$ of any two $\SU(2)$ elements $X$ and $Y$ is identical to the Euclidean inner product $\sum_{i=0}^3x^iy^i$ of two vectors $\vec{x}$ and $\vec{y}$ on $\mathbb{R}^4$, when we parameterize $X=x^0 I+i\sum_{i=1}^3x^i\bm{\sigma }^i$ and $Y=y^0 I+i\sum_{i=1}^3y^i\bm{\sigma }^i$ with $\{\bm{\sigma }^i\}$ being the Pauli matrices. We then obtain 
\be
\langle v_i,v_j\rangle=\frac{1}{2}\mathrm{Tr}\lb h_{ij}\rb=\cos\theta_{ij}\,,
\ee
where the identity of $\SL(2,\bC)$ matrices -- $\mathrm{Tr}(g)\mathrm{Tr}(h)=\mathrm{Tr}(gh)+\mathrm{Tr}(gh^{-1})\,,\forall g,h\in\SL(2,\bC)$ -- is used. The geodesic distances $\{\theta_{ij}\}$ between all pairs of points uniquely fix the positions of $v_1,\cdots,v_4$ on $S^3$ up to a global $\SU(2)$ left or right translation. We are led to the following result. 

\begin{lemma}
\label{lamma:unique_flat_connection}
The FN variables $\{\lambda_{ab}\}$ uniquely determine an $\SU(2)$ flat connection on $M_\pm'$.
    
\end{lemma}

\begin{proof}
The positions of 4 points $v_1,\cdots,v_4$ on $S^3$ have $3\times 4=12$ degrees of freedom, $2\times 3=6$ of which are gauges of left and right translations. Therefore, the relative positions of all points have 6 degrees of freedom and they are fixed by the set of 6 geodesic distances $\{\theta_{ij}\}$. When we restrict $\theta_{ij}\in[0,\pi]$, they are uniquely determined by $\{\lambda_{ab}\}$ through the relations $2\cos\theta_{ij}=\lambda_{ab}+\lambda_{ab}^{-1}$. The $\SU(2)$ flat connection is determined by Lemma \ref{bijectionV1}.
\end{proof}

Given that $\{\lambda_{ab}\}$ are shared by $M'_+$ and $M'_-$, an $\SU(2)$ flat connection on $M_{+\cup-}'=M'_+\cup M'_-$ is constructed from two identical  $\SU(2)$ flat connections on $M_+'$ and $M_-'$ respectively: We denote the representations of the flat connections on $M_+$ and $M_-$ by $\SU(2)$ group elements $\{h_{ij}\}$ and $\{h_{ij}'\}$ (modulo their conjugations) respectively. 
They being identical means that there exists a $g\in \mathrm{SU(2)}$ such that $h'_{ij}=gh_{ij}g^{-1}$ for all pairs $(i,j)$'s. The fundamental group $\pi_1(M_{+\cup-}')$ adds 3 generators $\ell_{I=1,2,3}$ to $\pi_1(S^3\backslash\Gamma_4)$, which come from the 3 non-contractible cycles of the ambient space in fig.\ref{fig:glue}. The SU(2) element $g$ is understood as the parallel transport from the base point $p_+\in M_+'$ of $\{h_{ij}\}$ to the base point $p_-\in M_-'$ of $\{h_{ij}'\}$. But since $g$ is obtained by identifying $\{h_{ij}\}$ and $\{h_{ij}'\}$, $g$ does not depend on the path connecting $p_+$ and $p_-$ (in particular, it does not depend on which $\cS_a$ the path goes across). It indicates that the holonomy along each $\ell_{I}$ has to be trivial. As a result, $\{h_{ij}\}$ and $h(\ell_I)=1$ define an $\SU(2)$ representation of {$\pi_1(M'_{+\cup-})$}. The $\SU(2)$ flat connection on {$M_{+\cup-}'$} is obtained by the gauge equivalence class of this representation. $h(\ell_I)=1$ implies that all B-cycle holonomies of the torus cusps are trivial,  consistent with the critical equation $n_f=0$ from \eqref{nf=0}. The resulting flat connection is irreducible if $\{h_{ij}\}$ are generic, and it can be identified as an irreducible flat connection on $M_{+\cup-}$ by adding trivial holonomies in the representation. So we obtain a critical point of the partial ampltiude $\cA_{\heta_5,\htrho_5}(\{\hrho^{(0)}_a\})$ up to $O(1/k)$ correction, for any boundary data $j_b,\heta_5,\htrho_5$. 
This is based on the assumption that the 3D ideal triangulation is sufficiently refined so that the flat connection is covered by the closure of the open patch.

By the above argument, different sets of $\{\lambda_{ab}\}$ determine different SU(2) flat connections. On $M_{+\cup-}'$, the FN variables $\{\lambda_{ab}\}$ associated to the torus cusps are integrated in the amplitude. There are continuously many SU(2) flat connections on $M_{+\cup-}'$ labelled by different $\{\lambda_{ab}\}$. All of these flat connections should correspond to the critical points of the amplitude. Then it suggests that there should exist degeneracy of the Hessian $H_\alpha$ caused by continuously many critical points. Recall the discussion at the end of the last section. Here $d=|\{\lambda_{ab}\}|=6$ indicates that the scaling of the amplitude should increase by $k^{3}$ compared to \eqref{eq:final_scaling_final} hence
\be
\cZ_{\htrho_5,\heta_5}(\vec{\mu}|\vec{m})\sim k^{15+6\mu}\,.
\label{eq:final_scaling3}
\ee
This formula is valid if $d=6$ is the maximal number of degenerate directions in the set of critical points. Rigorously speaking, the power in the formula is a lower bound, since any additional degeneracy, if exists, may increase the power.

\section{Conclusion and outlook}
\label{sec:conclusion}

In this work, we analyze the radiative correction corresponding to the melon graph with the spinfoam model introduced in \cite{Han:2021tzw}, which describes 4D quantum gravity with a non-vanishing cosmological constant $\Lambda$. The melon graph represents that two 4-simplices are glued by identifying four tetrahedra. We first construct the Chern-Simons partition function with the state-integral model for the 3-manifold $M_{+\cup -}$ corresponding to the melon graph then separate it into partition functions for two 3-manifolds, each corresponding to a spinfoam vertice, by using the over-completeness of the Chern-Simons coherent states. The spinfoam amplitude for $M_{+\cup -}$ is obtained by imposing simplicity constraints properly on the Chern-Simons partition function followed by coupling with 6 face amplitudes, each for a torus cusp. 
We propose a face amplitude as a $q$-deformed version of that in EPRL-FK model combined with a sign factor depending on the FN coordinate on the torus cusp. A key point of the paper is to show that the amplitude of the melon graph is finite. There is no infrared divergence in the radiative correction.

We study the scaling behavior of the melon graph amplitude at small $|\Lambda|$. It scales as $|\Lambda|^{-6\mu-15}$ provided that the face amplitude is a degree-$\mu$ polynomial of internal spins at small $|\Lambda|$ approximation. This provides the first-order correction of a spinfoam edge amplitude. At the $|\Lambda|\rightarrow 0$ limit, such a radiative correction diverges when {$\mu>-5/2$}. 

This work is the first application of this newly constructed spinfoam model. Compared to the original model \cite{Han:2021tzw}, we modify the imposition of the second-class simplicity constraints by using the trace coordinates as described in Section \ref {subsubsec:second_simplicity} so that degenerate simplices are also included in the expression of the spinfoam amplitude. A certain exponential suppressing factor in the edge amplitude in \cite{Han:2021tzw} is removed in the construction here, while the finiteness still holds. Our analysis also establishes that this spinfoam model is as computable as the EPRL-FK model and can be easily generalized to a general triangulation. 

It is natural to compare our result with that of the melonic radiative correction of the EPRL-FK model, which was recently found to scale as $|\Lambda|^{-1}$ at $\mu=1$ by numerical analysis \cite{Frisoni:2021uwx}. Even though the amplitude in our case scales differently (at least as $|\Lambda|^{-21}$ at $\mu=1$), a contradiction is {\it not} immediately drawn. This is because the coherent states that define the two spinfoam models differ. The coherent states in the EPRL-FK model are defined based on the holonomy-flux algebra of loop quantum gravity while the coherent states we use here are defined from the Chern-Simons phase space variables. It is nevertheless interesting to investigate the relation of these two coherent states which then relates the two spinfoam models. We expect that it will explain the different scalings of the melonic radiative corrections.

It is also interesting to study the divergence of the spinfoam amplitude corresponding to a more complicated spinfoam graph or even a general spinfoam graph. This may be systematically analyzed by developing a GFT formalism of the spinfoam model. Such a ``group field'' should encode the information of the cosmological constant and a consistent GFT should reproduce the divergent power of the melonic radiative correction we discover in this paper. 
When such a GFT is formulated and the relation of coherent states in this spinfoam model and those in the EPRL-FK model mentioned above is made clear, one can compare the divergences for other spinfoam graphs in this spinfoam model and the EPRL-FK model (see \eg \cite{BenGeloun:2010qkf,Krajewski:2010yq,Bahr:2015gxa}).

\begin{acknowledgements}
This work receives support from the National Science Foundation through grants PHY-2207763, PHY-2110234, the Blaumann Foundation and the Jumpstart Postdoctoral Program at Florida Atlantic University. The authors acknowledge IQG at FAU Erlangen-N\"urnberg, Perimeter Institute for Theoretical Physics, and University of Western Ontario for the hospitality during their visits.   

\end{acknowledgements}

\appendix
\renewcommand\thesection{\Alph{section}}

\section{Construction of Chern-Simons partition function on $S^3\backslash\Gamma_5$}
\label{app:partition}

In this appendix, we sketch the necessary steps to derive the partition function \eqref{eq:partition_S3G5_1} used in \cite{Han:2021tzw}. We refer to \eg \cite{Dimofte:2011gm,Dimofte:2013lba,Dimofte:2010wxa} for more details of the construction.

\subsection{Ideal tetrahedron partition function}
\label{subsec:ideal_tetra}

The phase space of $\PSL(2,\bC)$ Chern-Simons theory on the boundary $\partial\triangle$ of an ideal tetrahedron $\triangle$ is the moduli space $\cM_{\Flat}(\partial\triangle,\PSL(2,\bC))$ of {\it framed} \footnote{The moduli space we describe in this paper is for framed flat connection because the edge coordinates are defined as the cross-ratios of the {\it framing flags} at the disc cusps of each ideal tetrahedron. See \cite{Han:2021tzw} and reference therein for more details.} flat ${\rm PSL}(2,\bC)$ connection on $\partial\triangle$. Each edge $E$ belonging to the geodesic boundary of $\partial\triangle$ is dressed with an {\it edge coordinate} \cite{Fock:2003alg} $x_E$ which is a coordinate in $\cM_{\Flat}(\partial\triangle,\PSL(2,\bC))$. An edge coordinate can also be lifted to its logarithmic coordinate by choosing a branch such that $x_E=e^{\chi_E}$. 
The $\PSL(2,\bC)$ holonomies on $\partial\triangle$ can be written as $2\times 2$ matrices whose matrix elements are in terms of the edge coordinates dressing the edges they cross. This is called the ``snake rule''. We refer to  \cite{Dimofte:2013lba,Gaiotto:2009hg} for a detailed description of the snake rules. 
For a holonomy along a disc cusp with eigenvalue $\lambda\equiv e^L$, the snake rule gives
\be
\prod_{E \text{ around disc cusp}} (-x_E) =\lambda^2
\quad\Longleftrightarrow\quad
\sum_{E \text{ around disc cusp}} (\chi_E-i\pi)=2L\,.
\label{eq:snake_rule_annulus}
\ee
One immediately realizes that the edge coordinates are not sensitive to the sign of the eigenvalue $\lambda$. This reflects the fact that the gauge group is $\PSL(2,\bC)$ rather than $\SL(2,\bC)$. One can easily choose a lift $\sqrt{-x_E}$ or $-\sqrt{-x_E}$ of the edge coordinates, in which case the gauge group is lifted to $\SL(2,\bC)$. 
(We will choose the former lift for all the edges when constructing the discrete simplicity constraints. See Section \ref{subsubsec:second_simplicity}.)
When the eigenvalues are all fixed for holonomies around the four disc cusps of $\partial\triangle$, the moduli space of flat connection on $\partial\triangle$ is a symplectic space with the Poisson structure given by 
\be
\{\chi_E,\chi_{E'}\}=\epsilon_{EE'}\,,
\label{eq:ZZ'_Poisson}
\ee
where $\epsilon_{EE'}=0,\pm1$ counts the oriented triangles shared by $E,E'$ and $\epsilon_{EE'}=1$ if $E'$ occurs to the left of $E$ in the triangle. 

As shown in fig.\ref{fig:triangulation_All}, the disc cusps of $\triangle$ are not pierced by $\Gamma_5$ hence holonomies are trivial around each disc cusp. In other words, the connection is flat on $\triangle$. The Chern-Simons phase space $\cP_{\partial \triangle}$ on the boundary $\partial\triangle$ is given by three {\it edge coordinates} $\{z,z',z''\}\in\bC^*$ each labelling a pair of opposite edges of $\triangle$ as shown in fig.\ref{fig:ideal_tetra} and it is defined as
\be
\cP_{\partial\triangle}=\{z,z',z''\in\bC^*|zz'z''=-1\}\in(\bC^*)^2\,.
\label{eq:boundary_phase_space}
\ee
It comes from requiring that the holonomy $h$ around (any) one disc cusp of $\triangle$ defined by the snake rule 
\be
h=
\mat{cc}{1&0\\0&-z'} 
\mat{cc}{1&0\\1&1}
\mat{cc}{1&0\\0&-z''}
\mat{cc}{1&0\\1&1}
\mat{cc}{1&0\\0&-z}
\mat{cc}{1&0\\1&1}
\equiv 
\mat{cc}{1&0\\ zz' (z^{-1}+z''-1) & -zz'z''}
\label{eq:snake_rule_h}
\ee
is an $\SL(2,\bC)$ element hence $\det(h)=1$.
The constraint $zz'z''=-1$ eliminates one edge coordinate, say $z'$, then the holomorphic part of the Atiyah-Bott-Goldman symplectic form can be written as
\be
\Omega=\f{\rd z''}{z''}\w \f{\rd z}{z}\,.
\label{eq:ABG_symplectic_form_holo}
\ee
Taking the anti-holomorphic coordinates into account, the symplectic form for the Chern-Simons action \eqref{eq:CS_action} is: 
\be
\omega_{k,s}=\f{t}{4\pi}\Omega+\f{\tb}{4\pi}\overline{\Omega}\,.
\label{eq:ABG_symplectic_form}
\ee
Lift these coordinates to their logarithmic correspondence, $Z:=\log (z),Z':=\log(z'),Z'':=\log(z'')$ and similarly for the anti-holomorphic counterparts, the constraint of the edge coordinates and the Poisson structure induced by \eqref{eq:ABG_symplectic_form} are
\be
Z+Z'+Z''=i\pi=\bZ+\bZ'+\bZ''\,,\quad
\{Z,Z''\}_\Omega=1=\{\bZ,\bZ''\}_{\overline{\Omega}}\,.
\ee
Therefore, $(Z,Z'')$ and $(\bZ,\bZ'')$ form two canonical pairs.  
The quantization is based on another equivalent canonical pairs $(\mu,\nu)\in \R^2$ and $(m,n)\in (\Z/k\Z)^2$ defined as
\be
Z=\f{2\pi i}{k}\lb-ib\mu-m \rb \,,\quad
Z''=\f{2\pi i}{k}\lb-ib\nu-n \rb\,,\quad
\bZ=\f{2\pi i}{k}\lb-ib^{-1}\mu+m \rb \,,\quad
\bZ''=\f{2\pi i}{k}\lb-ib^{-1}\nu+n \rb \,,
\label{eq:Z_to_mu_m}
\ee
where $k\in\Z_+$ is defined in \eqref{eq:def_t_k} and $b$ is a phase parameter related to the Barbero-Immirzi parameter:
\be
b^2=\f{1-i\gamma}{1+i\gamma}\,,\quad \re(b)>0\,,\quad \im(b)\neq0\,,\quad |b|=1
\quad \Rightarrow \quad
t=\f{2k}{1+b^2}\,,\quad \tb=\f{2k}{1+b^{-2}}\,.
\ee
Conversely, one can express $Z,Z'',\bZ,\bZ''$ in terms of $(\mu,\nu,m,n)$ as
\be
\mu=\f{k}{2\pi Q}\lb Z+\bZ \rb\,,\quad 
m=\f{ik}{2\pi bQ}\lb Z-b^2\bZ \rb\,,\quad
\nu=\f{k}{2\pi Q}\lb Z''+\bZ''\rb\,,\quad 
n=\f{ik}{2\pi bQ}\lb Z''-b^2\bZ''\rb\,,\quad
Q=b+b^{-1}\,.
\label{eq:Z_to_mu_m}
\ee
The symplectic form in terms of the new variables and the Poisson brackets it generates are
\be
\omega_{k,s}=\f{2\pi}{k}\lb\rd\nu\w\rd\mu-\rd n\w\rd m \rb\,,\quad
\{\mu,\nu\}_\omega=\{n,m\}_\omega=\f{k}{2\pi}\,,\quad
\{\mu,n\}_\omega=\{\nu,m\}_\omega=0\,.
\ee
To promote to the quantum theory, we introduce quantum parameters
\be
q=\exp\lb\f{4\pi i}{t}\rb =\exp\left[\f{2\pi i}{k}(1+b^2)\right]\equiv e^{h}\,,\quad
\qt=\exp\lb\f{4\pi i}{\tb}\rb =\exp\left[\f{2\pi i}{k}(1+b^{-2})\right]\equiv e^{\tilde{h}}\,.
\label{eq:q_qt_def}
\ee
Here, $h:=4\pi i/t$ (or equivalently $\tilde{h}:=4\pi i/\tb$) is a (non-standard) complex quantum parameter related to the Chern-Simons level whose $h\rightarrow 0$ limit corresponds to the classical limit. 
{A Poisson bracket $\{x,y\}_\omega$ is quantized to a commutator by $[\hat{x},\hat{y}]:=\widehat{\{x,y\}}_\omega/i$.} 
{We allow the analytic continuation of $\mu,\nu$ to be in $\bC$ by adding imaginary parts, and define $Z,Z'',\Zt$ and $\Zt''$ in the same way as in \eqref{eq:Z_to_mu_m} with these complex variables. Then $\Zt$ (\resp $\Zt''$) is not necessarily the complex conjugate of $Z$ (\resp $Z''$). The exponential of $\Zt$ and $\Zt''$ are denoted as $\zt$ and $\zt''$ respectively. }
The quantization of $\cP_{\partial\triangle}$ promotes $\mu,m$ (\resp $Z,\Zt$) to be multiplication operators $\bmu,\bfm$ (\resp $\Zb,\Zbt$) and $\nu,n$ (\resp $Z'',\Zt''$) to be derivative operators $\bnu,\bfn$ (\resp $\Zb'',\Zbt''$) with the commutators
\be
{[\Zb'',\Zb]=h\,,\quad
[\Zbt'',\Zbt]=\tilde{h}\quad
\Longleftrightarrow\quad
[\bmu,\bnu]=[\bfn,\bfm]=\f{k}{2\pi i}\,,\quad
[\bmu,\bfn]=[\bnu,\bfm]=0\,.}
\label{eq:commutator}
\ee
Upon quantization, we require the imaginary parts of $\mu$ and $\nu$ remain to be $c$-numbers. 
Projecting the commutators to the exponential operators $\zb,\zb'',\zbt,\zbt''$, one finds $q$-Weyl and $\qt$-Weyl algebras
\be
{\zb''\zb=q\zb\zb''\,,\quad
\zbt''\zbt=\qt\zbt\zbt''\,,\quad
\zbt''\zb=\zb\zbt''\,,\quad
\zb''\zbt=\zbt\zb''\,.}
\ee
Due to the discreteness and periodicity of $m,n$, the spectra of $\bfm,\bfn$ are discrete and bounded to be $\mathbb{Z}/k\mathbb{Z}$. On the other hand, the spectra of $\bmu,\bnu$ are real. The kinematical Hilbert space of Chern-Simons theory is hence
\be
\cH^{\text{kin}}_{k,s}=L^2(\R)\otimes_{\bC} \bC^k\,,
\ee
where $\bC^k$ is a $k$-dimensional vector space. 
The quantum operators $\zb,\zb '',\zbt,\zbt''$ act on a wave function $f(\mu|m)\in\cH^{\text{kin}}_{k,s}$ as
\be
\zb f(\mu|m)= zf(\mu|m)\,,\quad
{\zb'' f(\mu|m)=f(\mu+ib|m-1)\,,\quad}
\zbt f(\mu|m)=\bz f(\mu|m)\,,\quad
{\zbt'' f(\mu|m)=f(\mu+ib^{-1}|m+1)}
\label{eq:z_on_f}
\ee
or a re-parameterized version
\be
\zb f(z,\zt)= zf(z,\zt )\,,\quad
{\zb'' f(z,\zt)=f(qz,\zt)\,,\quad}
\zbt f(z,\zt)=\bz f(z,\zt)\,,\quad
{\zbt'' f(z,\zt)=f(z,\qt\zt)\,.}
\label{eq:z_on_f_2}
\ee

$(z,z'')$ are holomorphic coordinates on $\cP_{\partial\triangle}$. The moduli space of flat $\PSL(2,\bC)$ connection on an ideal tetrahedron, denoted as $\cL_{\triangle}$, is a holomorphic Lagrange submanifold of $\cP_{\partial\triangle}$ determined by further requiring the holonomy $h$ defined in \eqref{eq:snake_rule_h} to be trivial. In other words, $\cL_{\triangle}$ is an algebraic curve given by
\be
\cL_\triangle=\{z^{-1}+z''-1=0\}\equiv \{\zt^{-1}+\zt''-1=0\}\subset \cP_{\partial\triangle}\,.
\ee
Quantization promotes the algebraic curve to the quantum constraints whose solution $\Psi_\triangle(\mu|m)$ satisfying
\be
(\zb^{-1}+\zb''-1)\Psi_\triangle=(\zbt^{-1}+\zbt''-1)\Psi_\triangle(\mu|m)=0
\label{eq:eom_quantum}
\ee
defines the Chern-Simons wave function for the ideal triangulation, or the Chern-Simons partition function with boundary condition specified by parameters $\mu$ and $m$. 
$\Psi_\triangle(\mu|m)$ is the quantum dilogarithm function \cite{Dimofte:2014zga,Imamura:2013qxa,Faddeev:1995nb,kashaev1997hyperbolic} \footnote{As $k=\f{12\pi}{\ell_{\text{p}}^2\gamma|\Lambda|}$ is taken to be positive integer, $\gamma\in\R_+$ hence $\im(b)<0$, leading to $|q|<1$. Suppose $k\in\Z_-$ then $\gamma<0$ and $|q|>1$, the quantum dilogarithm function takes the form $\Psi_\triangle(\mu|m)=\prod\limits_{j=0}^\infty \f{1-q^{j+1}z^{-1}}{1-\qt^{-j}\zt^{-1}}$, which is still the solution to \eqref{eq:eom_quantum}. }:
\be
\Psi_\triangle(\mu|m)=
	\prod\limits_{j=0}^\infty \f{1-\qt^{j+1}\zt^{-1}}{1-q^{-j}z^{-1}}\,.
\ee
$\Psi_{\triangle}$ has poles on the real line and in the lower half-plane $\im(\mu)\leq 0$ but is holomorphic in the upper half-plane $\im(\mu)>0$. 
Let $\alpha,\beta>0$, (The absolute value of) the function $e^{-\f{2\pi}{k}\beta\mu}\Psi_\triangle(\mu+i\alpha|m)$ with $\mu\in\R$ has limits
\be
|e^{-\f{2\pi}{k}\beta\mu}\Psi_\triangle(\mu+i\alpha|m)|\rightarrow \left\{
\ba{ll}
\exp\left[-\f{2\pi}{k}\beta\mu \right]\,,\quad &\mu\rightarrow+\infty\\[0.15cm]
\exp\left[-\f{2\pi}{k}\mu(\alpha+\beta-Q/2) \right]\,,\quad &\mu\rightarrow-\infty
\ea\right..
\ee
Therefore, $e^{-\f{2\pi}{k}\beta\mu}\Psi_\triangle(\mu+i\alpha|m)$ is a Schwartz function when $(\alpha,\beta)\in\fP_\triangle$ satisfy the {\it positive angle structure} of $\triangle$, defined as
\be
\fP_\triangle=\{(\alpha,\beta )\in\R^2|\alpha,\beta>0,\alpha+\beta<Q/2 \}\,.
\ee
The positive angle structure of a 3-manifold has been extensively discussed in \eg \cite{andersen2014complex,Dimofte:2014zga} and it is useful for understanding the Fourier transform of $\Psi_\triangle$. Let $\alpha=\im(\mu),\beta=\im(\nu)$, then $\int_{\cC}\rd\mu \,e^{-\f{2\pi i}{k}\nu\mu}\Psi(\mu|m)$ is absolutely convergent when the integration contour $\cC$ is shift above the real axis while remains in $\fP_\triangle$.

\subsection{Ideal octahedron partition function}
\label{subsec:ideal_octa}

Now that we have the Chern-Simons partition function $\Psi_\triangle$ on an ideal tetrahedron as the building block, the next step is to construct the partition function on an ideal octahedron. Each ideal octahedron can be decomposed into 4 ideal tetrahedra by adding an internal edge (see fig.\ref{fig:ideal_octa}). We then have 4 copies of edge coordinates $\{x,y,z,w\}$ (or considering the logarithm coordinates $\{X,Y,Z,W\} $) subject to the constraint
\be
\ba{l}
c=xyzw=1\\
\tilde{c}=\xt\yt\zt\wt=1
\ea
\quad\Longleftrightarrow\quad
\ba{l}
C=X+Y+Z+W=2\pi i \\
\widetilde{C}=\Xt+\Yt+\Zt+\Wt=2\pi i
\ea
\quad\Longleftrightarrow\quad
\ba{l}
\mu_X+\mu_Y+\mu_Z+\mu_W=0\\
m_X+m_Y+m_Z+m_W=0
\ea\,.
\label{eq:oct_constraint}
\ee

We define a set of symplectic coordinates $(X,P_X),(Y,P_Y),(Z,P_Z),(C,\Gamma)$ where
\be
P_X=X''-W''\,,\quad
P_Y=Y''-W''\,,\quad
P_Z=Z''-W''\,,\quad
\Gamma=W''\,,
\label{eq:momentum}
\ee
and similarly for the tilde sectors. 
Performing the symplectic reduction of the four copies of phase space $\cP_{\partial\triangle}$ associated to the four ideal tetrahedra by imposing the constraint $C=2\pi i$ as well as quotient out the gauge orbit variable $\Gamma$, we obtain the phase space $\cP_{\partial\oct}$ of the boundary of the ideal octahedron with the following symplectic form and Poisson structure.
\be
\omega^{\oct}_{k,s}=\f{2\pi}{k}\sum_{i}(\rd\nu_i\w\rd\mu_i -\rd n_i\w\rd m_i)\,,\quad
\left|\ba{l}
\{\mu_i,\nu_j\}_\omega=\{n_i,m_j\}_\omega= \f{k}{2\pi}\delta_{ij}\\[0.15cm]
\{\mu_i,n_j\}_\omega=\{\nu_i,m_j\}_\omega=0
\ea\right.\,,\quad i,j=X,Y,Z\,.
\ee
Quantization of the constraint $C$ and $\widetilde{C}$ adds a quantum correction as
\be
\ba{l}
c=1
\quad\rightarrow\quad
\hat{c}=q
\quad\Longleftrightarrow\quad
 C=2\pi i
\quad\rightarrow\quad
\hat{C}=2\pi i+h\,,\\
\tilde{c}=1
\quad\rightarrow\quad
\hat{\tilde{c}}=\qt
\quad\Longleftrightarrow\quad
 \widetilde{C}=2\pi i
\quad\rightarrow\quad
\hat{\widetilde{C}}=2\pi i+\tilde{h}\,.
\ea
\label{eq:oct_quantum_constraint}
\ee
In terms of $\{\bmu_i,\bfm_i\}_{i=X,Y,Z,W}$ which are the quantization of $\{\mu_i,m_i\}_{i=X,Y,Z,W}$, the quantum constraints read
\be
\bmu_X+\bmu_Y+\bmu_Z+\bmu_W=iQ\,,\quad
\bfm_X+\bfm_Y+\bfm_Z+\bfm_W=0\,.
\ee
Each octahedron partition function can hence be written in terms of the position variables $(x,y,z;\tilde{x},\tilde{y},\zt )\equiv\exp[(X,Y,Z;\tilde{X},\tilde{Y},\tilde{Z})]$ as 
\be
\cZ_{\oct}(x,y,z;\xt,\yt,\zt ) = \prod_{i,j,k,l=0}^{\infty} 
\f{1-q^{i+1}x^{-1}}{1-\qt^{-i}\xt^{-1}}
\f{1-q^{j+1}y^{-1}}{1-\qt^{-j}\yt^{-1}}
\f{1-q^{k+1}z^{-1}}{1-\qt^{-k}\zt^{-1}}
\f{1-q^{l}xyz}{1-\qt^{-l-1}\xt\yt\zt}\,,
\ee
where we have imposed the constraint \eqref{eq:oct_quantum_constraint} to eliminate the variables $w$ and $\wt$. $e^{-\f{2\pi}{k}\vec{\beta}\cdot\vec{\mu}}\cZ_{\oct}(\{\mu_i+i\alpha_i\}|\{m_i\})$ with $\mu_i\in\R$ and $\vec{\beta}\cdot\vec{\mu}\equiv \beta_X\mu_X+\beta_Y\mu_Y+\beta_Z\mu_Z$ has the following asymptotic behavior
\be
|e^{-\f{2\pi}{k}\vec{\beta}\cdot\vec{\mu}}\cZ_{\oct}(\{\mu_i+i\alpha_i\}|\{m_i\})|
\sim \left\{
\ba{ll}
e^{\f{2\pi}{k}\mu_i(\alpha_X+\alpha_Y+\alpha_Z+\beta_i-Q/2)}\,, & \mu_i\rightarrow +\infty\\
e^{\f{2\pi}{k}\mu_i(\alpha_i+\beta_i-Q/2)}\,,&\mu_i\rightarrow -\infty
\ea
\right.,\quad\forall i=X,Y,Z\,.
\ee
This function is a Schwartz function of $\mu_{X},\mu_{Y}$ and $\mu_{Z}$ if $(\alpha_{X},\alpha_{Y},\alpha_{Z},\beta_{X},\beta_{Y},\beta_{Z})\in \R^6$ is inside the open polytope $\fP(\oct)$ defined by the following inequalities
\be
\alpha_{i}>0\,,\quad
\alpha_X+\alpha_Y+\alpha_Z<Q\,,\quad
\alpha_{i}+\beta_{i}<Q/2\,,\quad
\alpha_X+\alpha_Y+\alpha_Z+\beta_{i}>Q/2\,,\quad
\forall i=X,Y,Z\,.
\label{eq:positive_angle_oct}
\ee
$(\vec{\alpha},\vec{\beta})\in \fP(\oct)$ is the positive angle of an ideal octahedron and has been shown in \cite{Han:2021tzw} to be non-empty. 
We also define the functional space
\be
\cF_{\fP(\oct)}=\left\{ \text{holomorphic }f:\bC^3\rightarrow \bC \mid \forall(\vec{\alpha},\vec{\beta})\in\fP(\oct)\,, e^{-\f{2\pi}{k}\vec{\beta}\cdot\vec{\mu}}f(\vec{\mu}+i\vec{\alpha})\in \cS(\R^3)\, \text{ is Schwartz class} \right\}\,.
\ee
This definition of the functional space $\cF_{\fP(\oct)}$ can be generalized to the functional space $\cF_\fP$ corresponding to any given positive angle structure $\fP$. (See \eg Section \ref{subsubsec:partition_S3Gamma5} for the case of $S^3\backslash\Gamma_5$.) 
Combining the discrete representation part, we define
\be
\cF_{\fP(\oct)}^{(k)}=\cF_{\fP(\oct)}\otimes_{\bC}(\bC^k)^{\otimes 3}\,.
\label{eq:def_Fk}
\ee
We conclude that $\cZ_\oct\in \cF_{\fP(\oct)}^{(k)}$.

\subsection{Chern-Simons partition function on $S^3\backslash\Gamma_5$}
\label{subsubsec:partition_S3Gamma5}

the Chern-Simons phase space $\cP_{\partial(S^3\backslash\Gamma_5)}$ is simply the 5 copies of $\cP_{\partial\oct}$ with no more constraints to be imposed. 
To impose the simplicity constraints in a more natural way as in Section \ref{subsec:simplicity_constraint}, we change the symplectic coordinates as follows. Denote the 15 position coordinates of the phase space $\cP_{\partial(S^3\backslash\Gamma_5)}$ to be $\vec{\Phi}=(X_i,Y_i,Z_i)_{i=1,\cdots,5}$ and the 15 momentum coordinates to be $\vec{\Pi}=(P_{X_i},P_{Y_i},P_{Z_i})_{i=1,\cdots,5}$ where each triple $(P_{X_i},P_{Y_i},P_{Z_i})$ is defined in the same way as \eqref{eq:momentum}.
The change of symplectic coordinates corresponds to performing (a series of) symplectic transformations which can be summarized by the following linear equations \cite{Han:2015gma,Han:2021tzw}.
\be
\mat{c}{\vec{\cQ} \\ \vec{\cP}}
=\mat{cc}{\bA & \bB \\ -(\bB^\top)^{-1} & 0}
\mat{cc}{\vec{\Phi} \\ \vec{\Pi}}
+\mat{c}{i\pi \vec{t}\\0}\,,\quad
\left|\ba{l}
\vec{\cQ}=(\{2L_{ab}\}_{(ab)},\{\cX_a\}_{a=1}^{5})\\[0.15cm]
\vec{\cP}=(\{\cT_{ab}\}_{(ab)},\{\cY_a\}_{a=1}^{5})
\ea\right., 
\label{eq:change_coordinate}
\ee
where $\bA$ and $\bB$ are $15 \times 15$ matrices with integer entries and $\vec{t}$ is a vector with integer elements (see \eqref{eq:ABt+}).
$(\vec{\cQ},\vec{\cP})$ can also be parametrized as \eqref{eq:new_coordinate_param}.
One can check that $(\vec{\cQ},\vec{\cP})$ do form a set of symplectic coordinates of the Chern-Simons phase space $\cP_{\partial(S^3\backslash\Gamma_5)}\equiv \otimes_{i=1}^5 \cP_{\partial\Oct(i)}$ on $\partial(S^3\backslash\Gamma_5)$. The Atiyah-Bott-Goldman symplectic form and the Poisson structure are \cite{Han:2015gma,Han:2021tzw}
\be
\Omega=\sum_{I=1}^{15}\cP_I\w\cQ_I\,,\quad
\{\cQ_I,\cP_J\}_\Omega=\delta_{IJ}\,,\quad
\{\cQ_I,\cQ_J\}_\Omega=\{\cP_I,\cP_J\}_\Omega=0\,,\quad 
I,J=1,\cdots,15\,.
\ee

The Chern-Simons partition function $\cZ_{\times}$ on $S^3\backslash\Gamma_5$ written in terms of coordinates $(\vec{\Phi},\vec{\Pi})$ is indeed the product of five $\cZ_{\oct}$'s. To express it in terms of the new coordinates $(\vec{\cQ},\vec{\cP})$, one separates the transformation matrix into generator matrices of the symplectic transformations:
\be
\mat{cc}{\bA & \bB \\ -(\bB^\top)^{-1} & 0}
=\mat{cc}{0 & -\id \\ \id & 0}
\mat{cc}{\id & 0 \\ \bA\bB^\top & \id}
\mat{cc}{-(\bB^{-1})^\top & 0 \\ 0 & -\bB}\,.
\ee 
The three matrices on the right-hand side correspond to the $S$-type, $T$-type and $U$-type transformations respectively \cite{Dimofte:2014zga,Han:2021tzw}.
Combining the affine translation $\sigma_{\vec{t}}$ given by the vector $\vec{t}$ as shown in \eqref{eq:change_coordinate}, the total action on the wave function $\cZ_\times$ corresponding to these transformation defines the Chern-Simons partition $\cS_{S^3\backslash\Gamma_5}$ on $S^3\backslash\Gamma_5$ in terms of new coordinates \eqref{eq:new_coordinate_param} \cite{Han:2021tzw}\footnotemark{}: 
\be\begin{split}
\cZ'_{S^3\backslash\Gamma_5}(\vec{\mu}|\vec{m})
&=((\sigma_{\vec{t}}\circ S\circ T \circ U )\triangleright \cZ_\times)(\vec{\mu}|\vec{m})\\
&=\f{4i}{k^{15}}\sum_{\vec{n}\in(\Z/k\Z)^{15}}\int_{\cC^{\times 15}}\rd^{15}\vec{\nu}\,
(-1)^{\vec{n}\cdot \bA\bB^\top\cdot \vec{n}}
e^{\f{i\pi}{k}(-\vec{\nu}\cdot \bA\bB^\top\cdot \vec{\nu}+\vec{n}\cdot \bA\bB^\top\cdot \vec{n}})
 e^{\f{2\pi i}{k}\left[-\vec{\nu}\cdot(\vec{\mu}-\f{iQ}{2}\vec{t})+\vec{n}\cdot \vec{m}\right]}\cZ_\times(-\bB^\top\vec{\nu}|-\bB^\top\vec{n})\,.
\end{split}
\label{eq:partition_S3G5_1_app}
\ee
\footnotetext{
The factor $(-1)^{\vec{n}\cdot \bA\bB^\top\cdot \vec{n}}$ is there to keep invariant the sign of the integrand of $\cZ'_{S^3\backslash\Gamma_5}(\vec{\mu}|\vec{m})$ when $n_I\rightarrow n_I+k$ for any $I$. The sign would change when $k$ is odd as well as the $I$-th the diagonal element of $\bA\bB^\top$ is odd (which happens for some $I$'s).
}
 
The positive angle structure $\fP(S^3\backslash\Gamma_5)$ for $S^3\backslash\Gamma_5$ in terms of the new variables $(\vec{\mu},\vec{\nu})$ is \cite{Han:2021tzw} \footnotemark{}
\be\begin{split}
&\fP(S^3\backslash\Gamma_5)=\sigma'_{\vec{t}}\circ S\circ T\circ U\circ \fP(\oct)^{\times 5}\\
\quad \Rightarrow\quad
&\text{ If } \, 
(\vec{\alpha}_0,\vec{\beta}_0)\in\fP(\oct)^{\times 5}\,,\quad \text{ then }\,
(\vec{\alpha}_4,\vec{\beta}_4)=(\bA\vec{\alpha}_0+\bB\vec{\beta}_0+\f{Q}{2}\vec{t},-(\bB^{-1})^\top\vec{\alpha}_0)\in \fP(S^3\backslash\Gamma_5)\,.
\end{split}
\label{eq:positive_angle_strucutre_M3}
\ee
Inversely,
\be
(\vec{\alpha}_0,\vec{\beta}_0)=(\bB^{\top}\vec{\beta}_4,\bB^{-1}\vec{\alpha}_4+\bA^\top\vec{\beta}_4-\f{Q}{2}\vec{t})\in\fP(\oct)^{\times 5}\,.
\ee
\footnotetext{
The operator $\vec{\sigma}'_{\vec{t}}$ for the positive angle structure is different from the affine transformation $\vec{\sigma}_{\vec{t}}$ acting on the wave functions. The latter is given in \eqref{eq:change_coordinate} while the former is defined as: $\vec{\sigma}'_{\vec{t}}: (\vec{\alpha},\vec{\beta})\mapsto (\vec{\alpha}+\f{Q}{2} \vec{t},\vec{\beta})$ \cite{Han:2021tzw}.
}
The symplectic transformations ensures that $\fP(S^3\backslash\Gamma_5)$ is non-empty since $\fP(\oct)^{\times 5}$ is non-empty, which concludes that $\cZ_{S^3\backslash\Gamma_5}\in\cF^{(k)}_{\fP(S^3\backslash\Gamma_5)}=\cF_{\fP(S^3\backslash\Gamma_5)}\otimes_\bC(\bC^k)^{\otimes 15}$. 

\section{Symplectic transformation of coordinates on $M_+$ and $M_-$}
\label{app:symplectic_tranf}

In this appendix, we collect the symplectic transformation matrix and the affine translation vector used in Sections \ref{subsubsec:partition_S3Gamma5} and \ref{subsubsec:constraint_system}. The linear symplectic transformation from $(\vec{\Phi}_\epsilon,\vec{\Pi}_\epsilon)=(\{X_i,Y_i,Z_i\}_{\substack{i=1,\cdots,5\,\text{ if }\epsilon=+\\i=6,\cdots,10\,\text{ if }\epsilon=-}},\{P_{X_i},P_{Y_i},P_{Z_i}\}_{\substack{i=1,\cdots,5\,\text{ if }\epsilon=+\\i=6,\cdots,10\,\text{ if }\epsilon=-}})$ to $(\vec{Q}^\epsilon,\vec{\cP}^\epsilon)$ defined in \eqref{eq:QP_epsilon_def} is
\be
\mat{c}{\vec{\cQ}^\epsilon \\ \vec{\cP}^\epsilon}
=\mat{cc}{\bA_\epsilon & \bB_\epsilon \\ -(\bB_\epsilon^\top)^{-1} & 0}
\mat{cc}{\vec{\Phi}_\epsilon \\ \vec{\Pi}_\epsilon}
+\mat{c}{i\pi \vec{t}_\epsilon\\0}\,.
\label{eq:symplectic_tranf-app}
\ee
The $\epsilon=+$ copy of the coordinates, transformation matrix and translation vector are the same as in a single $M_3$ used in \eqref{eq:change_coordinate}.

Explicitly, $\bA_\pm,\bB_\pm$ and $\vec{t}_\pm$ read
\begin{subequations}
\begin{align}
\bA=\bA_+=&\left(
\begin{array}{ccccccccccccccc}
 0 & 0 & 0 & 0 & 0 & 0 & 1 & 1 & 0 & 1 & 1 & 0 & 1 & 1 & 0 \\
 0 & 0 & 0 & 1 & 1 & 0 & 0 & 0 & 0 & 0 & 0 & 0 & 1 & 1 & 2 \\
 0 & 0 & 0 & 0 & 0 & 0 & 1 & 1 & 2 & 0 & 0 & 0 & 0 & 0 & 0 \\
 0 & 0 & 0 & 1 & 1 & 2 & 0 & 0 & 0 & 1 & 1 & 2 & 0 & 0 & 0 \\
 0 & 0 & 0 & 0 & 0 & 0 & 0 & 0 & 0 & 1 & -1 & 0 & 0 & 0 & 0 \\
 -1 & -1 & 0 & 0 & 0 & 0 & 0 & 0 & 0 & 0 & 0 & 0 & 1 & -1 & 0 \\
 1 & -1 & 0 & 0 & 0 & 0 & 1 & -1 & 0 & 0 & 0 & 0 & 0 & 0 & 0 \\
 -1 & -1 & -2 & 1 & -1 & 0 & 0 & 0 & 0 & 0 & 0 & 0 & 1 & 1 & 0 \\
 1 & 1 & 0 & 0 & 0 & 0 & 0 & 0 & 0 & -1 & -1 & 0 & 0 & 0 & 0 \\
 0 & 0 & 0 & -1 & -1 & 0 & 1 & 1 & 0 & 0 & 0 & 0 & 0 & 0 & 0 \\
 0 & 0 & 0 & 0 & 0 & -1 & 0 & 0 & 0 & 0 & 0 & 0 & 0 & 0 & 1 \\
 -1 & 0 & 0 & 0 & 0 & 0 & 0 & 0 & 0 & 0 & 0 & 0 & 1 & 0 & 0 \\
 -1 & -1 & -1 & 0 & 0 & 0 & 0 & 0 & 0 & 0 & 0 & 0 & 1 & 1 & 1 \\
 1 & 1 & 1 & 0 & 0 & 0 & 0 & 0 & 0 & 0 & 0 & 0 & -1 & 0 & 0 \\
 -1 & 0 & 0 & 0 & 0 & 0 & 0 & 0 & 0 & 1 & 1 & 1 & 0 & 0 & 0 \\
\end{array}
\right)\,,
\label{eq:A+}\\
\bB=\bB_+=&\left(
\begin{array}{ccccccccccccccc}
 0 & 0 & 0 & 0 & 0 & 0 & 0 & 1 & 0 & 0 & 1 & 0 & 0 & 1 & 0 \\
 0 & 0 & 0 & 0 & 1 & 0 & 0 & 0 & 0 & 0 & 1 & -1 & 0 & 0 & 1 \\
 0 & 0 & 0 & 0 & 1 & -1 & 0 & 0 & 1 & 0 & 0 & 0 & 0 & 1 & -1 \\
 0 & 0 & 0 & 0 & 0 & 1 & 0 & 1 & -1 & 0 & 0 & 1 & 0 & 0 & 0 \\
 -1 & 0 & 0 & 0 & 0 & 0 & 0 & 0 & 0 & 1 & -1 & 0 & -1 & 0 & 0 \\
 0 & -1 & 0 & 0 & 0 & 0 & -1 & 0 & 0 & 0 & 0 & 0 & 1 & -1 & 0 \\
 1 & -1 & 0 & 0 & 0 & 0 & 1 & -1 & 0 & -1 & 0 & 0 & 0 & 0 & 0 \\
 0 & 0 & -1 & 1 & -1 & 0 & 0 & 0 & 0 & 0 & 0 & 0 & 1 & 0 & -1 \\
 1 & 0 & -1 & -1 & 0 & 0 & 0 & 0 & 0 & -1 & 0 & 1 & 0 & 0 & 0 \\
 0 & 1 & -1 & -1 & 0 & 1 & 1 & 0 & -1 & 0 & 0 & 0 & 0 & 0 & 0 \\
 0 & 0 & 0 & 0 & 1 & -1 & 0 & 0 & 0 & 0 & 0 & 0 & 0 & 0 & 0 \\
 -1 & 0 & 0 & 0 & 0 & 0 & 0 & 0 & 0 & 0 & 0 & 0 & 0 & 0 & 0 \\
 0 & 0 & 0 & 0 & 0 & 0 & 0 & 0 & 0 & 0 & 0 & 0 & 1 & 0 & 0 \\
 0 & 0 & 1 & 0 & 0 & 0 & 0 & 0 & 0 & 0 & 0 & 0 & -1 & 1 & 0 \\
 -1 & 1 & 0 & 0 & 0 & 0 & 0 & 0 & 0 & 1 & 0 & 0 & 0 & 0 & 0 \\
\end{array}
\right)\,,
\label{eq:B+}\\
\vec{t}=\vec{t}_+&=(-3,-3,-2,-4,0,1,0,1,0,0,1,1,1,0,0)^\top\,,
\label{eq:t+}
\end{align}
\label{eq:ABt+}
\end{subequations}
\begin{subequations}
\begin{align}
\bA_-=&\left(
\begin{array}{ccccccccccccccc}
 0 & 0 & 0 & 0 & 0 & 0 & 0 & -1 & -1 & 0 & -1 & -1 & 0 & -1 & -1 \\
 0 & 0 & 0 & 0 & -1 & -1 & 0 & 0 & 0 & 0 & -1 & 1 & 0 & 0 & 0 \\
 0 & 0 & 0 & 0 & -1 & 1 & 0 & 0 & 0 & 0 & 0 & 0 & 0 & -1 & 1 \\
 0 & 0 & 0 & 0 & 0 & 0 & 0 & -1 & 1 & 0 & 0 & 0 & 0 & 0 & 0 \\
 2 & 1 & 1 & 0 & 0 & 0 & 0 & 0 & 0 & 0 & 0 & 0 & 2 & 1 & 1 \\
 0 & 1 & 1 & 0 & 0 & 0 & 2 & 1 & 1 & 0 & 0 & 0 & 0 & 0 & 0 \\
 0 & 0 & 0 & 0 & 0 & 0 & 0 & 0 & 0 & 2 & 1 & 1 & 0 & 0 & 0 \\
 0 & 0 & 0 & 0 & 0 & 0 & 0 & 0 & 0 & 0 & 0 & 0 & 0 & 1 & 1 \\
 0 & 1 & 1 & 2 & 1 & 1 & 0 & 0 & 0 & 0 & -1 & -1 & 0 & 0 & 0 \\
 0 & -1 & 1 & 0 & -1 & -1 & 0 & 1 & 1 & 0 & 0 & 0 & 0 & 0 & 0 \\
 0 & 0 & 0 & 0 & -1 & 0 & 0 & 0 & 0 & 0 & 0 & 0 & 0 & 0 & 1 \\
 1 & 1 & 1 & 0 & 0 & 0 & 0 & 0 & 0 & 0 & 0 & 0 & 1 & 0 & 0 \\
 -1 & -1 & -1 & 0 & 0 & 0 & 0 & 0 & 0 & 0 & 0 & 0 & -1 & 0 & 0 \\
 0 & 0 & -1 & 0 & 0 & 0 & 0 & 0 & 0 & 0 & 0 & 0 & 0 & -1 & 0 \\
 0 & -1 & 0 & 0 & 0 & 0 & 0 & 0 & 0 & -1 & 0 & 0 & 0 & 0 & 0 \\
\end{array}
\right)\,,
\label{eq:A-}\\
\bB_-=&\left(
\begin{array}{ccccccccccccccc}
 0 & 0 & 0 & 0 & 0 & 0 & 0 & -1 & 0 & 0 & -1 & 0 & 0 & -1 & 0 \\
 0 & 0 & 0 & 0 & -1 & 0 & 0 & 0 & 0 & 0 & -1 & 1 & 0 & 0 & -1 \\
 0 & 0 & 0 & 0 & -1 & 1 & 0 & 0 & -1 & 0 & 0 & 0 & 0 & -1 & 1 \\
 0 & 0 & 0 & 0 & 0 & -1 & 0 & -1 & 1 & 0 & 0 & -1 & 0 & 0 & 0 \\
 1 & 0 & 0 & 0 & 0 & 0 & 0 & 0 & 0 & -1 & 1 & 0 & 1 & 0 & 0 \\
 0 & 1 & 0 & 0 & 0 & 0 & 1 & 0 & 0 & 0 & 0 & 0 & -1 & 1 & 0 \\
 -1 & 1 & 0 & 0 & 0 & 0 & -1 & 1 & 0 & 1 & 0 & 0 & 0 & 0 & 0 \\
 0 & 0 & 1 & -1 & 1 & 0 & 0 & 0 & 0 & 0 & 0 & 0 & -1 & 0 & 1 \\
 -1 & 0 & 1 & 1 & 0 & 0 & 0 & 0 & 0 & 1 & 0 & -1 & 0 & 0 & 0 \\
 0 & -1 & 1 & 1 & 0 & -1 & -1 & 0 & 1 & 0 & 0 & 0 & 0 & 0 & 0 \\
 0 & 0 & 0 & 0 & -1 & 1 & 0 & 0 & 0 & 0 & 0 & 0 & 0 & 0 & 0 \\
 1 & 0 & 0 & 0 & 0 & 0 & 0 & 0 & 0 & 0 & 0 & 0 & 0 & 0 & 0 \\
 0 & 0 & 0 & 0 & 0 & 0 & 0 & 0 & 0 & 0 & 0 & 0 & -1 & 0 & 0 \\
 0 & 0 & -1 & 0 & 0 & 0 & 0 & 0 & 0 & 0 & 0 & 0 & 1 & -1 & 0 \\
 1 & -1 & 0 & 0 & 0 & 0 & 0 & 0 & 0 & -1 & 0 & 0 & 0 & 0 & 0 \\
\end{array}
\right)\,,
\label{eq:B-}\\
\vec{t}_-&=(3,1,0,0,-4,-3,-2,-1,-2,0,1,-1,3,2,2)^\top\,.
\label{eq:t-}
\end{align}
\label{eq:ABt-}
\end{subequations}

\section{Fock-Goncharov coordinates and Fenchel–Nielsen coordinates on $M_+$ and $M_-$}
\label{app:FG_FN_coordinates}

In this appendix, we collect the explicit definitions of the FG coordinates and FN coordinates dressing the edges or annuli on $M_+$ and $M_-$ which are two copies of $S^3\backslash\Gamma_5$. We refer to fig.\ref{fig:octahedra} for the face labels $a,b,c,d,e,f,g,h,i,j$. The FG coordinates on each $\cS_a$ in terms of the edge coordinates on $\{\Oct(i)\}$ are listed in Table \ref{tab:edges}. 
\begin{table}[h!!]
\begin{center}
\begin{tabular}{|c|c|c|}
\hline
$\cS_1$: & $\Ht_2\cap \Ht_3:\ \ \chi^{(1)}_{23}=Z_2+Z_3$        
		 & $\Ht_3\cap \et_4:\ \ \chi^{(1)}_{34}=Y_3''+Z_3'+Z_4''+W_4'$\\
         & $\Ht_2\cap \et_4:\ \ \chi^{(1)}_{24}=Z_2''+W_2'+Z_4$		
         & $\Ht_3\cap \ct_5:\ \ \chi^{(1)}_{35}=Z_3''+W_3'+Y_5''+Z_5'$\\
         & $\Ht_2\cap \ct_5:\ \ \chi^{(1)}_{25}=Y_2''+Z_2'+Z_5$
         & $\et_4\cap \ct_5:\ \ \chi^{(1)}_{45}=Y''_4+Z_4'+Z_5''+W_5'$\\
\hline   
         & $\Ht_7\cap \Ht_8:\ \ \chi^{(1)}_{78}=Z_7+Z_8$        
		 & $\Ht_8\cap \et_9:\ \ \chi^{(1)}_{89}=Y_8''+Z_8'+Z_9''+W_9'$\\
         & $\Ht_7\cap \et_9:\ \ \chi^{(1)}_{79}=Z_7''+W_7'+Z_9$		
         & $\Ht_8\cap \ct_{10}:\ \ \chi^{(1)}_{8,10}=Z_8''+W_8'+Y_{10}''+Z_{10}'$\\
         & $\Ht_7\cap \ct_{10}:\ \ \chi^{(1)}_{7,10}=Y_7''+Z_7'+Z_{10}$
         & $\et_9\cap \ct_{10}:\ \ \chi^{(1)}_{9,10}=Y''_9+Z_9'+Z_{10}''+W_{10}'$\\
\hline\hline
$\cS_2$: & $\ft_1\cap \It_3:\ \ \chi^{(2)}_{13}=X_1''+Y_1'+X_3$ 
		 & $\It_3\cap \ft_4:\ \ \chi^{(2)}_{34}=X_3''+Y_3'+W_4''+X_4'$\\
         & $\ft_1\cap \ft_4:\ \ \chi^{(2)}_{14}=X_1+X_4$        
         & $\It_3\cap \bt_5:\ \ \chi^{(2)}_{35}=W_3''+X_3'+X_5''+Y_5'$\\
         & $\ft_1\cap \bt_5:\ \ \chi^{(2)}_{15}=W_1''+X_1'+X_5$ 
         & $\ft_4\cap \bt_5:\ \ \chi^{(2)}_{45}=X''_4+Y_4'+W_5''+X_5'$\\
\hline
		 & $\ft_6\cap \It_8:\ \ \chi^{(2)}_{68}=X_6''+Y_6'+X_8$ 
		 & $\It_8\cap \ft_9:\ \ \chi^{(2)}_{89}=X_8''+Y_8'+W_9''+X_9'$\\
         & $\ft_6\cap \ft_9:\ \ \chi^{(2)}_{69}=X_6+X_9$        
         & $\It_8\cap \bt_{10}:\ \ \chi^{(2)}_{8,50}=W_8''+X_8'+X_{10}''+Y_{10}'$\\
         & $\ft_6\cap \bt_{10}:\ \ \chi^{(2)}_{6,10}=W_6''+X_6'+X_{10}$ 
         & $\ft_9\cap \bt_{10}:\ \ \chi^{(2)}_{9,10}=X''_9+Y_9'+W_{10}''+X_{10}'$\\
\hline\hline
$\cS_3$: & $\bt_1\cap \at_2:\ \ \chi^{(3)}_{12}=Z_1'+W_1''+X_2$       
		 & $\at_2\cap \dt_4:\ \ \chi^{(3)}_{24}=W_2''+X_2'+Y_4'+Z_4''$\\
         & $\bt_1\cap \dt_4:\ \ \chi^{(3)}_{14}=W'_1+X_1''+X_4'+Y_4''$ 
         & $\at_2\cap \dt_5:\ \ \chi^{(3)}_{25}=X_2''+Y_2'+Z_5'+W_5''$\\
         & $\bt_1\cap \dt_5:\ \ \chi^{(3)}_{15}=W_1+W_5'+X_5''$        
         & $\dt_4\cap \dt_5:\ \ \chi^{(3)}_{45}=Y_4+W_5$\\
\hline
		 & $\bt_6\cap \at_7:\ \ \chi^{(3)}_{67}=Z_6'+W_6''+X_7$       
		 & $\at_7\cap \dt_9:\ \ \chi^{(3)}_{79}=W_7''+X_7'+Y_9'+Z_9''$\\
         & $\bt_6\cap \dt_9:\ \ \chi^{(3)}_{69}=W'_6+X_6''+X_9'+Y_9''$ 
         & $\at_7\cap \dt_{10}:\ \ \chi^{(3)}_{7,10}=X_7''+Y_7'+Z_{10}'+W_{10}''$\\
         & $\bt_6\cap \dt_{10}:\ \ \chi^{(3)}_{6,10}=W_6+W_{10}'+X_{10}''$        
         & $\dt_9\cap \dt_{10}:\ \ \chi^{(3)}_{9,10}=Y_9+W_{10}$\\
\hline\hline        
$\cS_4$: & $\at_1\cap \ct_2:\ \ \chi^{(4)}_{12}=Z_1+X_2'+Y_2''$        
		 & $\ct_2\cap \jt_3:\ \ \chi^{(4)}_{23}=Y_2'+Z_2''+Z_3'+W_3''$\\
         & $\at_1\cap \jt_3:\ \ \chi^{(4)}_{13}=Y''_1+Z_1'+W_3'+X_3''$ 
         & $\ct_2\cap \jt_5:\ \ \chi^{(4)}_{25}=Y_2+Y_5'+Z_5''$\\
         & $\at_1\cap \jt_5:\ \ \chi^{(4)}_{15}=Z_1''+W_1'+X_5'+Y_5''$ 
         & $\jt_3\cap \jt_5:\ \ \chi^{(4)}_{35}=W_3+Y_5$\\
\hline
		 & $\at_6\cap \ct_7:\ \ \chi^{(4)}_{67}=Z_6+X_7'+Y_7''$        
		 & $\ct_7\cap \jt_8:\ \ \chi^{(4)}_{78}=Y_7'+Z_7''+Z_8'+W_8''$\\
         & $\at_6\cap \jt_8:\ \ \chi^{(4)}_{68}=Y''_6+Z_6'+W_8'+X_8''$ 
         & $\ct_7\cap \jt_{10}:\ \ \chi^{(4)}_{7,10}=Y_7+Y_{10}'+Z_{10}''$\\
         & $\at_6\cap \jt_{10}:\ \ \chi^{(4)}_{6,10}=Z_6''+W_6'+X_{10}'+Y_{10}''$ 
         & $\jt_8\cap \jt_{10}:\ \ \chi^{(4)}_{8,10}=W_8+Y_{10}$\\

\hline\hline
$\cS_5$: & $\It_1\cap \et_2:\ \ \chi^{(5)}_{12}=Y'_1+Z_1''+W_2'+X_2''$
		 & $\et_2\cap \gt_3:\ \ \chi^{(5)}_{23}=Z_2'+W_2''+Y_3'+Z_3''$\\
         & $\It_1\cap \gt_3:\ \ \chi^{(5)}_{13}=Y_1+X_3'+Y_3''       $ 
         & $\et_2\cap \gt_4:\ \ \chi^{(5)}_{24}=W_2+Z_4'+W_4''$\\
         & $\It_1\cap \gt_4:\ \ \chi^{(5)}_{14}=X_1'+Y_1''+W_4'+X_4''$ 
         & $\gt_3\cap \gt_4:\ \ \chi^{(5)}_{34}=Y_3+W_4$\\
\hline\hline
$\cS_6$: & $\It_6\cap \et_7:\ \ \chi^{(6)}_{67}=Y'_6+Z_6''+W_7'+X_7''$
		 & $\et_7\cap \gt_8:\ \ \chi^{(6)}_{78}=Z_7'+W_7''+Y_8'+Z_8''$\\
         & $\It_6\cap \gt_8:\ \ \chi^{(6)}_{68}=Y_6+X_8'+Y_8''       $ 
         & $\et_7\cap \gt_9:\ \ \chi^{(6)}_{79}=W_7+Z_9'+W_9''$\\
         & $\It_6\cap \gt_9:\ \ \chi^{(6)}_{69}=X_6'+Y_6''+W_{10}'+X_9''$ 
         & $\gt_8\cap \gt_9:\ \ \chi^{(6)}_{89}=Y_8+W_9$\\
\hline
\end{tabular}
\end{center}
\caption{FG coordinates $\chi^{(a)}_{ij}$ of 4-holed spheres in terms of the edge coordinates in $\{\Oct(i)\}$. $i,j$ denote 
that $\chi^{(a)}_{ij}$ is composed with coordinates from octahedra $\Oct(i)$ and $\Oct(j)$. 
We have used the notations in fig.\ref{fig:octahedra} where the octahedra are glued through the triangles labelled by $a,b,c,d,e,f,g,h,i,j$. Each ``tilde triangle" with the tilde label, say $\at_2$, labels the triangles symmetric to the triangle $a_2$ with respect to the equator of $\Oct(2)$. We refer to \cite{Han:2021tzw} for more details (where the ``prime triangles'' are the same as the tilde triangles used here). Here $X_{i},Y_i,Z_i,W_i$ ($i=1,\cdots,10$) are the tetrahedron edge coordinates from the 4 ideal tetrahedra in $\Oct(i)$.}
\label{tab:edges} 
\end{table}

The FN coordinates $\{2L_{ab}\}$ in $\cP_{\partial M_+}$ and the FN coordinates  $\{2L'_{ab}\}$ in $\cP_{\partial M_-}$ are defined in terms of the FG coordinates on $\{\cS_a\}$ as
\begin{subequations}
\begin{align}
  2L_{12}=\chi^{(1)}_{34}+\chi^{(1)}_{35}+\chi^{(1)}_{45}-3i\pi\,, \quad
& 2L'_{12}=\chi^{(1)}_{89}+\chi^{(1)}_{8,10}+\chi^{(1)}_{9,10}-3i\pi\,,\\
  2L_{21}=\chi^{(2)}_{34}+\chi^{(2)}_{35}+\chi^{(2)}_{45}-3i\pi\,,\quad
& 2L'_{21}=\chi^{(2)}_{89}+\chi^{(2)}_{8,10}+\chi^{(2)}_{9,10}-3i\pi\,,\\
  2L_{13}=\chi^{(1)}_{24}+\chi^{(1)}_{25}+\chi^{(1)}_{45}-3i\pi\,, \quad
& 2L'_{13}=\chi^{(1)}_{79}+\chi^{(1)}_{7,10}+\chi^{(1)}_{9,10}-3i\pi\,,\\
  2L_{31}=\chi^{(3)}_{24}+\chi^{(3)}_{25}+\chi^{(3)}_{45}-3i\pi\,,\quad
& 2L'_{31}=\chi^{(3)}_{79}+\chi^{(3)}_{7,10}+\chi^{(3)}_{9,10}-3i\pi\,,\\
  2L_{14}=\chi^{(1)}_{23}+\chi^{(1)}_{25}+\chi^{(1)}_{35}-3i\pi\,, \quad
& 2L'_{14}=\chi^{(1)}_{78}+\chi^{(1)}_{7,10}+\chi^{(1)}_{8,10}-3i\pi\,,\\
  2L_{41}=\chi^{(4)}_{23}+\chi^{(4)}_{25}+\chi^{(4)}_{35}-3i\pi\,,\quad
& 2L'_{41}=\chi^{(4)}_{78}+\chi^{(4)}_{7,10}+\chi^{(4)}_{8,10}-3i\pi\,,\\
  2L_{15}=\chi^{(1)}_{23}+\chi^{(1)}_{24}+\chi^{(1)}_{34}-3i\pi\,, \quad
& 2L'_{16}=\chi^{(1)}_{78}+\chi^{(1)}_{79}+\chi^{(1)}_{89}-3i\pi\,,\\
  2L_{51}=\chi^{(5)}_{23}+\chi^{(5)}_{24}+\chi^{(5)}_{34}-3i\pi\,,\quad
& 2L'_{61}=\chi^{(6)}_{78}+\chi^{(6)}_{79}+\chi^{(6)}_{89}-3i\pi\,,\\
  2L_{23}=\chi^{(2)}_{14}+\chi^{(2)}_{15}+\chi^{(2)}_{45}-3i\pi\,, \quad
& 2L'_{23}=\chi^{(2)}_{69}+\chi^{(2)}_{6,10}+\chi^{(2)}_{9,10}-3i\pi\,,\\
  2L_{32}=\chi^{(3)}_{14}+\chi^{(3)}_{15}+\chi^{(3)}_{45}-3i\pi\,,\quad
& 2L'_{32}=\chi^{(3)}_{69}+\chi^{(3)}_{6,10}+\chi^{(3)}_{9,10}-3i\pi\,,\\
  2L_{24}=\chi^{(2)}_{13}+\chi^{(2)}_{15}+\chi^{(2)}_{35}-3i\pi\,, \quad
& 2L'_{24}=\chi^{(2)}_{68}+\chi^{(2)}_{6,10}+\chi^{(2)}_{8,10}-3i\pi\,,\\
  2L_{42}=\chi^{(4)}_{13}+\chi^{(4)}_{15}+\chi^{(4)}_{35}-3i\pi\,,\quad
& 2L'_{42}=\chi^{(4)}_{68}+\chi^{(4)}_{6,10}+\chi^{(4)}_{8,10}-3i\pi\,,\\
  2L_{25}=\chi^{(2)}_{13}+\chi^{(2)}_{14}+\chi^{(2)}_{34}-3i\pi\,, \quad
& 2L'_{26}=\chi^{(2)}_{68}+\chi^{(2)}_{69}+\chi^{(2)}_{89}-3i\pi\,,\\
  2L_{52}=\chi^{(5)}_{13}+\chi^{(5)}_{14}+\chi^{(5)}_{34}-3i\pi\,,\quad
& 2L'_{62}=\chi^{(6)}_{68}+\chi^{(6)}_{69}+\chi^{(6)}_{89}-3i\pi\,,\\
  2L_{34}=\chi^{(3)}_{12}+\chi^{(3)}_{15}+\chi^{(3)}_{25}-3i\pi\,, \quad
& 2L'_{34}=\chi^{(3)}_{67}+\chi^{(3)}_{6,10}+\chi^{(3)}_{7,10}-3i\pi\,,\\
  2L_{43}=\chi^{(4)}_{12}+\chi^{(4)}_{15}+\chi^{(4)}_{25}-3i\pi\,,\quad
& 2L'_{43}=\chi^{(4)}_{67}+\chi^{(4)}_{6,10}+\chi^{(4)}_{7,10}-3i\pi\,,\\
  2L_{35}=\chi^{(3)}_{12}+\chi^{(3)}_{14}+\chi^{(3)}_{24}-3i\pi\,, \quad
& 2L'_{36}=\chi^{(3)}_{67}+\chi^{(3)}_{69}+\chi^{(3)}_{79}-3i\pi\,,\\
  2L_{53}=\chi^{(5)}_{12}+\chi^{(5)}_{14}+\chi^{(5)}_{24}-3i\pi\,,\quad
& 2L'_{63}=\chi^{(6)}_{67}+\chi^{(6)}_{69}+\chi^{(6)}_{79}-3i\pi\,,\\
  2L_{45}=\chi^{(4)}_{12}+\chi^{(4)}_{13}+\chi^{(4)}_{23}-3i\pi\,, \quad
& 2L'_{46}=\chi^{(4)}_{67}+\chi^{(4)}_{68}+\chi^{(4)}_{78}-3i\pi\,,\\
  2L_{54}=\chi^{(5)}_{12}+\chi^{(5)}_{13}+\chi^{(5)}_{23}-3i\pi\,,\quad
& 2L'_{64}=\chi^{(6)}_{67}+\chi^{(6)}_{68}+\chi^{(6)}_{78}-3i\pi\,.
\end{align}
\label{eq:all_FN_coordinates}
\end{subequations}

The conjugate momenta $\cT_{ab}$ and $\cT'_{ab}$ can be easily calculated through $-(\bB_\pm^\top)^{-1}\cdot \vec{\Phi}_\pm$. See also the appendix of \cite{Han:2021tzw} for the explicit expressions for $\cT_{ab}$. 

The FG coordinates $\{\cX_a,\cY_a,\cX'_a,\cY'_a\}_{a=1}^5$ are chosen to be
\be
\ba{lllll}
\cX_1=\chi^{(1)}_{25}\,,\quad
&\cX_2=\chi^{(2)}_{15}\,,\quad
&\cX_3=\chi^{(3)}_{15}\,,\quad
&\cX_4=\chi^{(4)}_{15}\,,\quad
&\cX_5=\chi^{(5)}_{14}\,,\\[0.15cm]
\cY_1=\chi^{(1)}_{23}\,,\quad
&\cY_2=\chi^{(2)}_{14}\,,\quad
&\cY_3=\chi^{(3)}_{45}-2\pi i\,,\quad
&\cY_4=-\chi^{(4)}_{35}+2\pi i\,,\quad
&\cY_5=\chi^{(5)}_{34}-2\pi i\,,\\[0.15cm]
\cX'_1=\chi^{(1)}_{79}\,,\quad
&\cX'_2=\chi^{(2)}_{6,10}\,,\quad
&\cX'_3=\chi^{(3)}_{6,10}\,,\quad
&\cX'_4=\chi^{(4)}_{6,10}\,,\quad
&\cX'_5=\chi^{(6)}_{6,9}\,,\\[0.15cm]
\cY'_1=\chi^{(1)}_{78}\,,\quad
&\cY'_2=\chi^{(2)}_{69}\,,\quad
&\cY'_3=\chi^{(3)}_{9,10}-2\pi i\,,\quad
&\cY'_4=-\chi^{(4)}_{8,10}+2\pi i\,,\quad
&\cY'_5=\chi^{(6)}_{89}-2\pi i\,.
\ea
\ee

\section{Gluing of holes from different 4-holed spheres to from $\partial(S^3\backslash\Gamma_5 )$}
\label{app:identify_holes}

Recall the FG coordinates $(\cX_a,\cY_a)$ on $\cS_a$ whose definitions are given in Appendix \ref{app:FG_FN_coordinates}. 
If we label the holes on each $\cS_a$ by numbers $1,2,3,4$, and identify $z^{(a)}_{12}=e^{\cX_a}\,,\,z^{(a)}_{13}=e^{\cY_a}\,,\,\forall a=1,\cdots,5$, the way of gluing holes from different 4-holed sphere to form $\partial(S^3\backslash\Gamma_5)$ is unique. 

 Denote the $i$-th ($i=1,2,3,4$) hole in $\cS^a$ as $\fp_i^{(a)}$. 
The gluing (denoted by $\sim$ below) of holes between different $\cS_a$'s is 
 \be\ba{llllll}
  \fp_1^{(1)}\sim \fp_3^{(4)}\,,\quad 
 &\fp_2^{(1)}\sim \fp_3^{(3)}\,,\quad
 &\fp_3^{(1)}\sim \fp_3^{(5)}\,,\quad
 &\fp_4^{(1)}\sim \fp_4^{(2)}\,,\quad
 &\fp_1^{(2)}\sim \fp_1^{(3)}\,,\quad\\[0.2cm]
  \fp_2^{(2)}\sim \fp_1^{(4)}\,,\quad
 &\fp_3^{(2)}\sim \fp_1^{(5)}\,,\quad
 &\fp_2^{(3)}\sim \fp_2^{(4)}\,,\quad
 &\fp_4^{(3)}\sim \fp_2^{(5)}\,,\quad
 &\fp_4^{(4)}\sim \fp_4^{(5)}\,,\quad
 \ea
 \label{eq:identify_holes}
 \ee
which is graphically illustrated in fig.\ref{fig:identify_holes}.
 \begin{figure}[h!]
 \centering
 \includegraphics[width=0.4\textwidth]{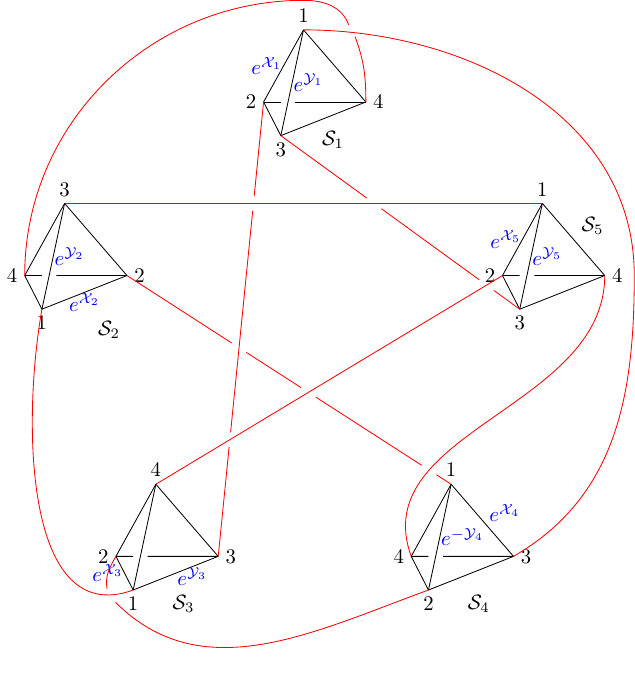}
 \caption{Identifying the holes from different 4-holed spheres. The numbers 1,2,3,4 one each 4-holed sphere $\cS_a$ denote the holes and the red lines demonstrate the gluing of holes from different $\cS_a$'s. Each red line corresponds to a blue line (open or closed) in fig.\ref{fig:glue}. Each tetrahedron graph here is the same as the ideal triangulation ({\it in black}) as in fig.\ref{fig:XY_choice}.}
 \label{fig:identify_holes}
 \end{figure}
 This means the $z^{(a)}_{\fp_i\fp_j}=-(y^{(a)}_{\fp_i\fp_j})^2$ in the trace coordinates formulas \eqref{eq:trace_coordinates} in different $\cS_a$'s correspond to the FG coordinates in the following way up to signs and $\pm 2\pi i$.
 \be\ba{llllll}
 z^{(1)}_{12} \rightarrow \cX_{25}^{(1)}\,,
&z^{(1)}_{13} \rightarrow \cX_{23}^{(1)}\,,
&z^{(1)}_{14} \rightarrow \cX_{35}^{(1)}\,,
&z^{(1)}_{23} \rightarrow \cX_{24}^{(1)}\,,
&z^{(1)}_{24} \rightarrow \cX_{45}^{(1)}\,,
&z^{(1)}_{34} \rightarrow \cX_{34}^{(1)}\,,\\[0.2cm]
 z^{(2)}_{12} \rightarrow \cX_{15}^{(2)}\,,
&z^{(2)}_{13} \rightarrow \cX_{14}^{(2)}\,,
&z^{(2)}_{14} \rightarrow \cX_{45}^{(2)}\,,
&z^{(2)}_{23} \rightarrow \cX_{13}^{(2)}\,,
&z^{(2)}_{24} \rightarrow \cX_{35}^{(2)}\,,
&z^{(2)}_{34} \rightarrow \cX_{34}^{(2)}\,,\\[0.2cm]
 z^{(3)}_{12} \rightarrow \cX_{15}^{(3)}\,,
&z^{(3)}_{13} \rightarrow \cX_{45}^{(3)}\,,
&z^{(3)}_{14} \rightarrow \cX_{14}^{(3)}\,,
&z^{(3)}_{23} \rightarrow \cX_{25}^{(3)}\,,
&z^{(3)}_{24} \rightarrow \cX_{12}^{(3)}\,,
&z^{(3)}_{34} \rightarrow \cX_{24}^{(3)}\,,\\[0.2cm]
 z^{(4)}_{12} \rightarrow \cX_{15}^{(4)}\,,
&z^{(4)}_{13} \rightarrow \cX_{35}^{(4)}\,,
&z^{(4)}_{14} \rightarrow \cX_{13}^{(4)}\,,
&z^{(4)}_{23} \rightarrow \cX_{25}^{(4)}\,,
&z^{(4)}_{24} \rightarrow \cX_{12}^{(4)}\,,
&z^{(4)}_{34} \rightarrow \cX_{23}^{(4)}\,,\\[0.2cm]
 z^{(5)}_{12} \rightarrow \cX_{14}^{(5)}\,,
&z^{(5)}_{13} \rightarrow \cX_{34}^{(5)}\,,
&z^{(5)}_{14} \rightarrow \cX_{13}^{(5)}\,,
&z^{(5)}_{23} \rightarrow \cX_{24}^{(5)}\,,
&z^{(5)}_{24} \rightarrow \cX_{12}^{(5)}\,,
&z^{(5)}_{34} \rightarrow \cX_{23}^{(5)}\,.
 \ea\ee

\section{Trace coordinates from the snake rule}
\label{app:snake}

In this appendix, we describe the snake rule calculating the holonomies around one or two holes following \cite{Dimofte:2013lba}, which leads to the trace coordinate expressions \eqref{eq:trace_coordinates} and \eqref{eq:trace_coordinates_2}. 
Notations can refer to Section \ref{subsubsec:second_simplicity}.

Let us first fix the labels for the holes on a 4-holed sphere $\cS_a$, hence the edges $\{e_{\fp_1\fp_2}\}$ on its ideal triangulation, to be consistent with fig.\ref{fig:XY_choice}. 

There are three rules for transporting a {\it snake} -- an arrow pointing from one vertex of the triangle to another with a {\it fin} facing inside the triangle, each corresponds to a matrix as follows. (The inverse transportation of each type corresponds to the inverse of the relevant matrix.
\be
\includegraphics[width=0.25\textwidth]{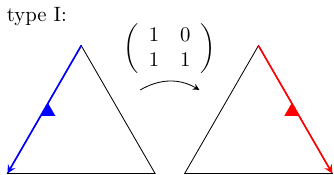}\quad\qquad
\includegraphics[width=0.25\textwidth]{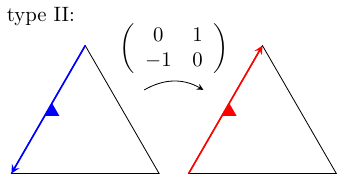}\quad\qquad
\includegraphics[width=0.3\textwidth]{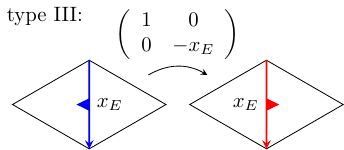}\,.
\label{eq:snake_rule}
\ee
Type I and II correspond to transporting a snake within a triangle and III correspond to moving a snake from one triangle to its adjacent triangle. Any holonomy of a closed loop can be calculated by multiplying the transportation matrices corresponding to moving a snake along the holonomy. 

Choose the snake starting on edge $e_{13}$ pointing from hole 1 to hole 3 whose fin faces the triangle bounded by $e_{12},e_{23},e_{13}$ as shown in fig.\ref{fig:XY_choice}. The holonomies around single holes 1,2 and 3 read (from left to right) 
\begin{subequations}
\begin{align}
h_1&=\left(
\begin{array}{cc}
 1 & 0 \\
 0 & -z_{13} \\
\end{array}
\right)\left(
\begin{array}{cc}
 1 & 0 \\
 1 & 1 \\
\end{array}
\right)\left(
\begin{array}{cc}
 1 & 0 \\
 0 & -z_{14} \\
\end{array}
\right)\left(
\begin{array}{cc}
 1 & 0 \\
 1 & 1 \\
\end{array}
\right)\left(
\begin{array}{cc}
 1 & 0 \\
 0 & -z_{12} \\
\end{array}
\right)\left(
\begin{array}{cc}
 1 & 0 \\
 1 & 1 \\
\end{array}
\right)\,,\\
h_2&=\left(
\begin{array}{cc}
 1 & 0 \\
 -1 & 1 \\
\end{array}
\right)\left(
\begin{array}{cc}
 0 & 1 \\
 -1 & 0 \\
\end{array}
\right)\left(
\begin{array}{cc}
 1 & 0 \\
 0 & -z_{12} \\
\end{array}
\right)\left(
\begin{array}{cc}
 1 & 0 \\
 1 & 1 \\
\end{array}
\right)\left(
\begin{array}{cc}
 1 & 0 \\
 0 & -z_{24} \\
\end{array}
\right)\left(
\begin{array}{cc}
 1 & 0 \\
 1 & 1 \\
\end{array}
\right)\left(
\begin{array}{cc}
 1 & 0 \\
 0 & -z_{23} \\
\end{array}
\right)\left(
\begin{array}{cc}
 0 & 1 \\
 -1 & 0 \\
\end{array}
\right)\left(
\begin{array}{cc}
 1 & 0 \\
 -1 & 1 \\
\end{array}
\right)\left(
\begin{array}{cc}
 0 & 1 \\
 -1 & 0 \\
\end{array}
\right)\,,\\
h_3&=\left(
\begin{array}{cc}
 0 & -1 \\
 1 & 0 \\
\end{array}
\right)\left(
\begin{array}{cc}
 1 & 0 \\
 1 & 1 \\
\end{array}
\right)\left(
\begin{array}{cc}
 1 & 0 \\
 0 & -z_{23} \\
\end{array}
\right)\left(
\begin{array}{cc}
 1 & 0 \\
 1 & 1 \\
\end{array}
\right)\left(
\begin{array}{cc}
 1 & 0 \\
 0 & -z_{34} \\
\end{array}
\right)\left(
\begin{array}{cc}
 1 & 0 \\
 1 & 1 \\
\end{array}
\right)\left(
\begin{array}{cc}
 1 & 0 \\
 0 & -z_{13} \\
\end{array}
\right)\left(
\begin{array}{cc}
 0 & 1 \\
 -1 & 0 \\
\end{array}
\right)\,.
\end{align}
\end{subequations}
As the snake is not in the neighbourhood of hole 4, the holonomy around hole 4 needs a ``special edge'' to transport the snake to its neighbourhood. We choose the special edge to be $e_{13}$. Then 
\begin{multline}
h_4=\left(
\begin{array}{cc}
 1 & 0 \\
 0 & -z_{13} \\
\end{array}
\right)\left(
\begin{array}{cc}
 1 & 0 \\
 1 & 1 \\
\end{array}
\right)\left(
\begin{array}{cc}
 0 & -1 \\
 1 & 0 \\
\end{array}
\right)\left(
\begin{array}{cc}
 1 & 0 \\
 1 & 1 \\
\end{array}
\right)\left(
\begin{array}{cc}
 1 & 0 \\
 0 & -z_{34} \\
\end{array}
\right)\left(
\begin{array}{cc}
 1 & 0 \\
 1 & 1 \\
\end{array}
\right)\left(
\begin{array}{cc}
 1 & 0 \\
 0 & -z_{24} \\
\end{array}
\right)\\
\left(
\begin{array}{cc}
 1 & 0 \\
 1 & 1 \\
\end{array}
\right)\left(
\begin{array}{cc}
 1 & 0 \\
 0 & -z_{14} \\
\end{array}
\right)\left(
\begin{array}{cc}
 0 & -1 \\
 1 & 0 \\
\end{array}
\right)\left(
\begin{array}{cc}
 1 & 0 \\
 -1 & 1 \\
\end{array}
\right)\left(
\begin{array}{cc}
 1 & 0 \\
 0 & -1/z_{13} \\
\end{array}
\right)\,.
\end{multline}
With the chosen lift $y_{\fp_1\fp_2}:=\sqrt{-z_{\fp_1\fp_2}}$, $h_1,h_2,h_3$ and $h_4$ are all $\SL(2,\bC)$ elements.
Then the traces of these holonomies reproduce the results of \eqref{eq:trace}:
\be\ba{ll}
\tr(h_1)=\sqrt{-z_{12} z_{13} z_{14}}+\frac{1}{\sqrt{-z_{12} z_{13} z_{14}}}=\lambda_1+\lambda_1^{-1}\,,\quad&
\tr(h_2)=\sqrt{-z_{12} z_{23} z_{24}}+\frac{1}{\sqrt{-z_{12} z_{23} z_{24}}}=\lambda_2+\lambda_2^{-1}\,,\quad\\[0.15cm]
\tr(h_3)=\sqrt{-z_{13} z_{23} z_{34}}+\frac{1}{\sqrt{-z_{13} z_{23} z_{34}}}=\lambda_3+\lambda_3^{-1}\,,\quad&
\tr(h_4)=\sqrt{-z_{14} z_{24} z_{34}}+\frac{1}{\sqrt{-z_{14} z_{24} z_{34}}
}=\lambda_4+\lambda_4^{-1}\,.
\ea
\label{eq:trace_app}
\ee
The holonomies $h_1,h_2,h_3$ and $h_4$ satisfy the closure constraints by the snake rule:
\be
h_1h_2h_3h_4\equiv\mat{cc}{1&0\\0&1}\,.
\label{eq:closure_snake}
\ee
The holonomies $h_{12}$ around holes 1,2 and $h_{23}$ around holes 2,3 and $h_{13}$ around holes 1,3 are simply $h_{12}=h_1h_2$, $h_{23}=h_2h_3$ and $h_{13}=h_1h_3$ respectively since they are all calculated starting from the same snake. 
The traces \eqref{eq:trace_coordinates} can be immediately obtained, plugging \eqref{eq:lambda} or \eqref{eq:trace} into which gives \eqref{eq:trace_coordinates_2}. 
The trace coordinates $\{\fm_1,\fm_2,\fm_2,\fm_3,\fkt_1,\fkt_2,\fkt_3\}$ satisfy the polynomial \eqref{eq:polynomial}.

\section{Expectation values of the Fock-Goncharov operators}
\label{app:expectation}

In this appendix, we calculate the expectation values of the operators $\bmu,\bnu,e^{\f{2\pi i}{k}\bfm}$ and $e^{\f{2\pi i}{k}\bfn}$ with the coherent states basis $\Psi^0_\rho(\mu|m)$ defined in \eqref{eq:coherent_for_amplitude}. Note that $\im(\mu)$ and $\im(\nu)$ remain classical, so we treat $\mu,\nu\in\R$ for notational simplicity in this appendix.

Recall the operation actions \eqref{eq:z_on_f} (or equivalently \eqref{eq:z_on_f_2}) of $\z,\z'',\zbt,\zbt''$ on any function $f(\mu|m)$. They generate the operation actions of $\bmu,\bnu,e^{\f{2\pi i}{k}\bfm}$ and $e^{\f{2\pi i}{k}\bfn}$ on $f(\mu|m)$ in the following way. 
\be
\bmu f(\mu|m) = \mu f(\mu|m)\,,\quad \bnu f(\mu|m) = -\f{k}{2\pi i} \partial_\mu f(\mu|m)\,,\quad
e^{\f{2\pi i}{k}\bfm} f(\mu|m) =e^{\f{2\pi i}{k}m} f(\mu|m)\,,\quad
e^{\f{2\pi i}{k}\bfn} f(\mu|m) = f(\mu|m+1)\,.
\ee
The complex-plane part of the coherent state $\psi^0_z(\mu)$ defined in \eqref{eq:coherent_2} is normalized: 
\be
\int_\R\rd\mu\, \bar{\psi}^0_z(\mu) \psi^0_z(\mu) =1\,.
\ee
On the other hand, the torus part $\xi_{(x,y)}(m)$ \eqref{eq:coherent_Vk} can be expressed in terms of the Jacobi theta function
\be
{\xi_{(x, y)}(m)={\sqrt[4]{2} }\,{k^{-3 / 4}}e^{-\frac{k y(y-i x)}{4 \pi}} \vartheta_3\left(
X_m,\tau
\right)\,,\quad
\left|\ba{rll}
X_m&=&\frac{1}{2}\left(-\frac{2 \pi m}{k}+x+i y\right)\\
\tau&=&e^{-\frac{\pi}{k}}
\ea\right.,}
\ee
which is normalized only at the large-$k$ approximation \cite{Gazeau:2009zz}:
\be
\sum_{m=0}^{k-1} \bar{\xi}_{(x,y)}(m)\xi_{(x,y)}(m)
=\sqrt{2}k^{-3/2} e^{-\f{ky^2}{2\pi}}
\sum_{m=0}^{k-1}\left|\vartheta_3\left(X_m,\tau
\right)\right|^2
\xrightarrow{k\rightarrow\infty} 1\,.
\ee
Therefore, $\Psi^0_\rho(\mu|m)$ is only normalized at the large-$k$ approximation. 
In this approximation, we calculate the expectation values of the operators $\bmu,\bnu,e^{\f{2\pi i}{k}\bfm}$ and $e^{\f{2\pi i}{k}\bfn}$ on $f(\mu|m)$ under the $\Psi^0_\rho(\mu|m)$ basis.
Recall that $(z,x,y)\in\bC\times \bT^2$ are related to the classical phase space coordinates $(\mu_0,\nu_0)\in\R^2$ and {$(m_0,n_0)\in [0,k)^{\times 2}$} by
\be
z=\f{\sqrt{2}\pi}{k}(\mu_0+i \nu_0)\,,\quad
x=\f{2\pi}{k}m_0\,,\quad
{y= \f{2\pi}{k}n_0\,.}
\ee
We get (we omit the parameters and variables in the coherent states unless necessary for conciseness) 
\begin{subequations}
\begin{align}
&\sum_{m=0}^{k-1}\int_\R\rd\mu\, \bar{\Psi}^0 \bmu \Psi^0 
&&=\lb \sum_{m=0}^{k-1}\bar{\xi}\xi\rb\,
\lb\f{2}{k}\rb^{\f12}\int_\R\rd\mu \,\mu\, e^{-\f{2\pi}{k} \lb\mu-\f{k}{\pi\sqrt{2}}\re(z)\rb^2}
=\f{k}{\pi\sqrt{2}}\re(z) \sum_{m=0}^{k-1}\int_\R\rd\mu\, \bar{\Psi}^0 \Psi^0
\xrightarrow{k\rightarrow\infty} \mu_0\,,\\
&\sum_{m=0}^{k-1}\int_\R\rd\mu\, \bar{\Psi}^0 \bnu \Psi^0
&&=-\f{k}{2\pi i}\lb \sum_{m=0}^{k-1}\bar{\xi}\xi\rb\,
\lb\f{2}{k}\rb^{\f14}\int_\R\rd\mu\, \bar{\psi}^0 \partial_\mu\lb e^{-\f{\pi}{k}\lb\mu-\f{k}{\pi\sqrt{2}}\re(z)\rb^2} e^{-i\sqrt{2}\mu\im(z)}\rb \nn\\
& &&= \lb\sum_{m=0}^{k-1}\bar{\xi}\xi\rb \int_\R\rd\mu\, \lb -i\mu+\f{k}{\sqrt{2}\pi i}\bz\rb \bar{\psi}^0\psi^0
\xrightarrow{k\rightarrow\infty} \nu_0\,,\\
&\sum_{m=0}^{k-1}\int_\R\rd\mu\, \bar{\Psi}^0 e^{\f{2\pi i}{k}\bfm} \Psi^0 
&&=\lb\int_\R\rd\mu\, \bar{\psi}^0\psi^0\rb \,
\sum_{m=\lfloor -(k-1)/2 \rfloor}^{\lfloor (k-1)/2\rfloor} \bar{\xi}e^{\f{2\pi i}{k}m}\xi \nn\\
& &&\xrightarrow{k\rightarrow\infty}\sum_{q\in\Z}\int_\R\rd m \, e^{\f{2\pi i}{k}m}\bar{\xi}\xi \,e^{2\pi i q m}\xrightarrow{k\rightarrow\infty} e^{ix}\equiv e^{\f{2\pi i}{k}m_0}\,,
\label{eq:ev_m}\\
&\sum_{m=0}^{k-1}\int_\R\rd\mu\, \bar{\Psi}^0 e^{\f{2\pi i}{k}\bfn} \Psi^0 
&&=\lb\int_\R\rd\mu\, \bar{\psi}^0\psi^0\rb \,
\sum_{m=\lfloor -(k-1)/2 \rfloor}^{\lfloor (k-1)/2\rfloor} \bar{\xi}(m)\xi(m+1)\nn\\
& &&\xrightarrow{k\rightarrow\infty}\sum_{q\in\Z}\int_\R\rd m \, \bar{\xi}(m)\xi(m+1)\,e^{2\pi i qm} \xrightarrow{k\rightarrow\infty}e^{iy}\equiv e^{\f{2\pi i}{k}n_0}\,.
\label{eq:ev_n}
\end{align}
\end{subequations}
In \eqref{eq:ev_m} and \eqref{eq:ev_n}, we have shifted the summation over $m$ by $\lfloor -(k-1)/2\rfloor$ where $\lfloor \alpha\rfloor$ denotes the floor function of $\alpha$ which enters the greatest integer less than $\alpha\in\R$ and used the Poisson resummation to change the summation of $m$ to integral. It is permitted by the periodicity of functions $\xi_{(x,y)}(m)$ and $e^{\f{2\pi i}{k}m}$, \ie they are invariant by changing $m\mapsto m+k$. We, therefore, conclude that 
\be
\la \bmu \ra \xrightarrow{k\rightarrow\infty} \mu_0\,,\quad 
\la \bnu \ra \xrightarrow{k\rightarrow\infty}\nu_0\,,\quad
\la e^{\f{2\pi i}{k}\bfm} \ra\xrightarrow{k\rightarrow\infty}e^{\f{2\pi i}{k}m_0}\,,\quad 
\la e^{\f{2\pi i}{k}\bfn} \ra\xrightarrow{k\rightarrow\infty}e^{\f{2\pi i}{k}n_0}\,.
\ee

\section{Positive angle structure for the new coordinates of $\cP_{\partial M_{+\cup -}}$}
\label{app:positive_angle}

In this appendix, we give some examples of the change of the positive angle structure according to the symplectic coordinate transformation from the original edge coordinates $(\vec{\cQ}^\pm,\vec{\cP}^\pm)$ to the final coordinates $(\vec{\fQ},\vec{\fP})$. The existence of these examples guarantees that the positive angle structure of the final coordinate is non-empty. 

Let us assume that the positive angles for the initial 10 ideal octahedra possess the symmetry: $\alpha_{+,I}=\alpha_{-,I}=\alpha\,\, \forall I=1,\cdots, 15$ and $\beta_{\epsilon,ix}=\beta_{\epsilon,iy}=\beta_{\epsilon,iz}=\beta_{\epsilon,i}\,\,\forall i=1,\cdots,5\,,\forall\epsilon=\pm$. Then one can solve that $\alpha\equiv Q/4$ and $\beta_{+,i}=\beta_{-,i}\,,\forall i=1,\cdots,5$ from the constraints $\alpha_{\text{new}}(\cC_A)=0\,,\forall A=1,\cdots, 18$. They indeed satisfy all the inequalities of \eqref{eq:positive_angle_oct} as long as $|\beta_{\epsilon,i}|<Q/4$. As a numerical check, let $\alpha=Q/4$ and $\beta_{\epsilon,i}=Q/6\,,\forall i=1,\cdots,5\,,\forall\epsilon=\pm$. Then the positive angles $(\vec{\alpha}_{M_\pm},\vec{\beta}_{M_\pm})\in\fP_{M_\pm}$ are calculated to be
\begin{align}
\vec{\alpha}_{M_+}&=\bB_+\vec{\beta}_++\bA_+\vec{\alpha}_++\f{Q}{2}\vec{t}_+
=\{\f12,\f13,\f16,\f13,-\f13,-\f13,-\f16,-\f16,-\f16,0,\f12,\f13,\f23,\f23,\f23\}Q \,,\\[0.15cm]
\vec{\beta}_{M_+}&=-(\bB_+^{-1})^\top \vec{\alpha}_+
=\{0,-\f14,-\f14,-\f12,0,\f12,-\f14,\f14,0,\f12,\f12,\f12,-\f12,\f12,-\f12\}Q \,,\\[0.15cm]
\vec{\alpha}_{M_-}&=\bB_-\vec{\beta}_-+\bA_-\vec{\alpha}_-+\f{Q}{2}\vec{t}_-
=\{-\f12,-\f13,-\f16,-\f13,\f13,\f13,\f16,\f16,\f16,0,\f12,\f23,\f13,\f13,\f13\}Q \,,\\[0.15cm]
\vec{\beta}_{M_-}&=-(\bB_-^{-1})^\top \vec{\alpha}_-
=\{0,\f14,\f14,\f12,0,-\f12,\f14,-\f14,0,-\f12,-\f12,-\f12,\f12,-\f12,\f12\}Q\,.
\end{align}
The positive angles satisfying \eqref{eq:positive_angle_new12} for $M_{+\cup -}$ are then
\begin{align}
\vec{\alpha}_{\text{new}}&=
\{\f12,\f13,\f16,\f13,-\f13,-\f13,-\f16,-\f16,-\f16,0,\f23,\f13,0,0,0,0,0,0,0,0,0,0,0,0,0,0,0,0,0,0\}Q\,,\\[0.15cm]
\vec{\beta}_{\text{new}}&=
\{0,-\f12,-\f12,-1,0,1,-\f12,\f12,0,1,-\f12,\f12,0,\f14,\f14,\f12,0,-\f12,\f14,-\f14,0,-\f12,-\f12,-\f12,\f12,-\f12,-\f12,-\f13,-\f23,-\f23\}Q\,,
\end{align}
which confirms the vanishing positive angle $\alpha_{\text{new}}(\cC_A)$ for all the constraints $\{\cC_A\}_{A=1}^{18}$.

\section{Fenchel-Nielsen twist computed by the snake rule}
\label{app:twist}

In this appendix, we use the snake rule on the cusps boundaries \cite{Dimofte:2011gm} (which are different from \eqref{eq:snake_rule}) to compute the coordinates $T_f=\log(\tau_f)$ corresponding to the B-cycles holonomy eigenvalue of the torus cusps and $T_b=\log(\tau_b)$ corresponding to the FN twist of the annulus cusps. The results depend on the choices of path but are different by a linear function of FN lengths $L_f$'s and $L_b$'s. We choose the paths that 
are consistent with the choices in \cite{Han:2015gma}. 

In general, $\{\cT_{ab}-\cT'_{ab}\}$ and $\{T_{ab}\}\equiv\{T_f,T_b\}$ may not be the same but are related by the following lemma. 

\begin{lemma}
\be
\cT_{ab}-\cT'_{ab}=T_{ab}+\zeta_{ab}(\{2L_{ab}\})\,,\quad
{\widetilde{\cT}_{ab}-\widetilde{\cT}'_{ab}=\tilde{T}_{ab}-\zeta_{ab}(\{2L_{ab}\})\,.}
\ee
where $\zeta_{ab}$ is a linear function of the set $\{2L_{ab}\}_{(ab)}$ with real linear coefficients while an imaginary constant term. 
\end{lemma}
\begin{proof}
{Similar to $\cT_{ab}$ and $\cT'_{ab}$, $T_{ab}$'s are some linear functions of $\fZ_i, P_{\fZ_i}$ {with real coefficients in the linear terms and the constant term takes the form $\pi i\cdot c_{ab}$ ($c_{ab}\in\R$) as it comes from the linear combination of the affine translations} (see Appendix A.3.3 in \cite{Dimofte:2013lba}, see also \cite{Han:2015gma} for the example related to our model). By the symplectic transformation in \eqref{QIPI}, each $T_{ab}$ is expressed as a linear function of $(\mathfrak{Q}_I,\mathfrak{P}_I)$. By definition, $T_{ab}$ is a function on the phase space for $\partial M_{+\cup -}$, then it is a linear function of the symplectic coordinates $\{2L_{ab},\cT_{ab}-\cT_{ab}',\cX_5,\cY_5,\cX_5',\cY_5'\}$ (but not of $\{\Gamma_A\}$ when the constraints $\{\cC_A\}$ are imposed). 
Moreover, its Poisson bracket with $2L_{cd}$ must be $\{2L_{cd},T_{ab}\}=\delta_{(ab),(cd)}$ and $\{T_{ab},T_{cd}\}=0\,,\forall (ab),(cd)$, which means $T_{ab}$ can only be a linear combination of $\{2L_{ab}\}$ and $\cT_{ab}-\cT_{ab}'$. The same argument applies to the tilde sector. The tilded variables are just complex conjugates of the non-tilded ones, and $\overline{\zeta_{ab}(2L_{ab})}=-\zeta_{ab}(2L_{ab})$. Hence the second equation in \eqref{cTcTT} holds.}
\end{proof}

Due to the fact that we have chosen $L'_{ab}$ in a symmetric way as $L_{ab}$, it turns out that all $\zeta_{ab}$'s have only constant terms. Explicitly, we use the snake rule described in the following to compute $T_{ab}$ (they are also used in \cite{Han:2015gma}) and find \eqref{eq:zeta0}, which we copy here:
\be
\zeta_{12}=0\,,\,\,
\zeta_{13}=\pi i\,,\,\,
\zeta_{14}=\pi i\,,\,\,
\zeta_{15}=\pi i\,,\,\,
\zeta_{23}=0\,,\,\,
\zeta_{24}=-2\pi i\,,\,\,
\zeta_{25}=0\,,\,\,
\zeta_{34}=-\pi i\,,\,\,
\zeta_{35}=0\,,\,\,
\zeta_{45}=-\pi i\,.
\label{eq:zeta}
\ee 

We now describe the snake rule for cusp boundaries. Dress the vertex (or the angle) of a disc cusp, which is a triangle, in an ideal tetrahedron $\triangle$ by $\fz$ ($\fz=z,z', z''$) when this vertex is connected to an edge of $\triangle$ dressed with $\fz$. 
Assume the oriented paths on the cusp boundary are all non-intersecting. 
The snake rule on a cusp boundary can be separated into two types on a single disc cusp, each corresponding to an operation on the logarithmic FN coordinate: 
\be
\includegraphics[width=0.25\textwidth]{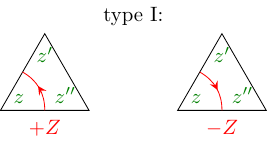}\quad\qquad
\includegraphics[width=0.25\textwidth]{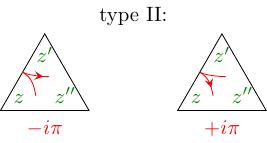}\,.
\label{eq:cusp_snake_rule}
\ee
Type I corresponds to the part of a path crossing an angle dressed with $\zb$ in a counter-clockwise (\resp clockwise) direction. It adds $+\fZ=\log \fz$ (\resp $-\fZ=\log(\fz^{-1})$) to the FN coordinate for the path. On the other hand,  Type II corresponds to the part of a path bouncing against an edge of the disc cusp in a  clockwise (\resp counter-clockwise) direction relative to the disc cusp. It adds $-i\pi$ (\resp $+i\pi$) to the FN coordinate for the path. 

Each FN coordinate corresponds to a path on the cusp boundary as shown in fig.\ref{fig:TAll_snake_rule}.
\eqref{eq:cusp_snake_rule} provides another way to formulate $\{L_{ab},L'_{ab}\}$ other than performing the symplectic transformation from the FG coordinates as in \eqref{eq:all_FN_coordinates}. As an example, $L_{12}, L'_{12}$ correspond to the A-cycles (with no winding) (red paths in fig.\ref{fig:T12_snake}) of the torus cusp connecting $\cS_1$ and $\cS_2$ and $ T_{12}-T'_{12}$ (note that $T_{ab}\neq \cT_{ab}\,,T'_{ab}\neq \cT'_{ab} $ in general) corresponds to the B-cycle (blue path in fig.\ref{fig:T12_snake}) of the same torus cusp. 

It is easy to read $T_{ab}$ and $T'_{ab}$ from fig.\ref{fig:TAll_snake_rule}:
\be
\ba{lllll}
 T_{12}=X_3''-Y_3+Z_3\,,\,
&T_{13}=Z''_5-W''_5\,,\,
&T_{14}=Y_2'-Z_2'\,,\,
&T_{23}=X_4''-Y_4''\,,\,
&T_{24}=-Y''_5+X''_5\,,\\[0.15cm]
 T'_{12}=-X_8''+Y_8-Z_8\,,\, 
&T'_{13}=-Z''_{10}+W''_{10}\,,\,
&T'_{14}=-Y_7'+Z_7'\,,\,
&T'_{23}=-X_9''+Y_9''\,,\,
&T'_{24}=Y''_{10}-X''_{10}\,,\\[0.25cm]
T_{34}=-Z''_1+W'_1\,,\,
&T_{15}=-W''_4+Z''_4\,,\,
&T_{25}=X_1''-Y_1''\,,\,
&T_{35}=W_2'-X_2'\,,\,
&T_{45}=-Y'_3+Z_3-W''_3\,,\\[0.15cm]
T'_{34}=Z''_6-W'_6\,,\,
&T'_{15}=W''_9-Z''_9\,,\,
&T'_{25}=-X_6''+Y_6''\,,\,
&T'_{35}=-W_7'+X_7'\,,\,
&T'_{45}=Y'_8-Z_8+W''_8\,.
\ea
\ee
\begin{figure}[h!]
\centering
\begin{minipage}{0.3\textwidth}
\centering
\includegraphics[width=\textwidth]{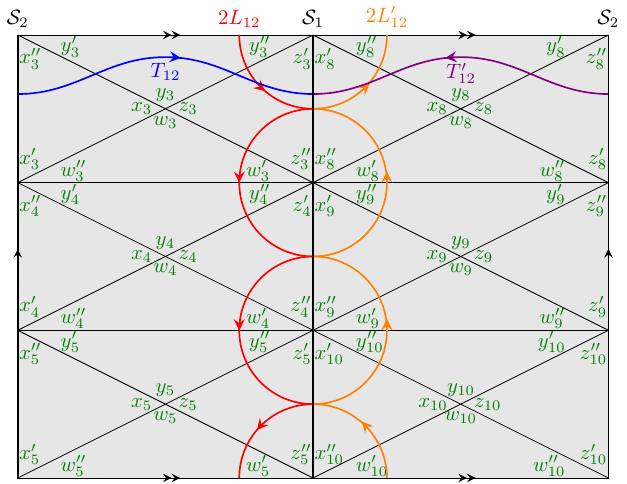}
\subcaption{$(ab)=(12)$}
\label{fig:T12_snake}
\end{minipage}
\begin{minipage}{0.3\textwidth}
\centering
\includegraphics[width=\textwidth]{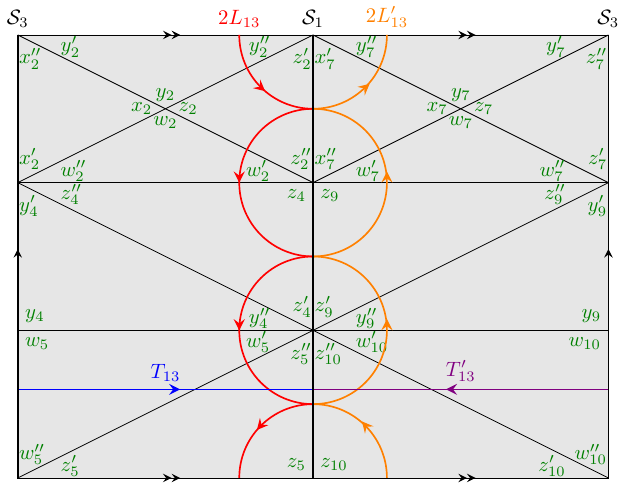}
\subcaption{$(ab)=(13)$}
\label{fig:T13_snake}
\end{minipage}
\begin{minipage}{0.3\textwidth}
\centering
\includegraphics[width=\textwidth]{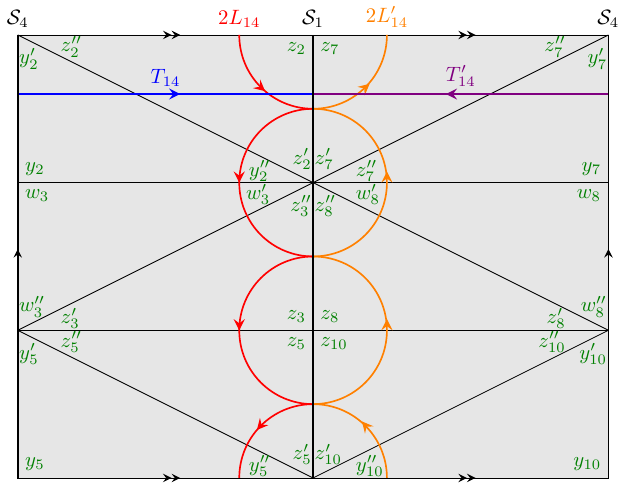}
\subcaption{$(ab)=(14)$}
\label{fig:T14_snake}
\end{minipage}
\begin{minipage}{0.3\textwidth}
\centering
\includegraphics[width=\textwidth]{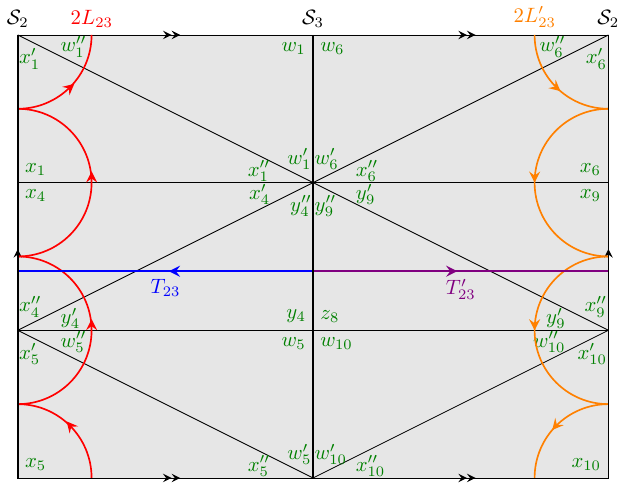}
\subcaption{$(ab)=(23)$}
\label{fig:T23_snake}
\end{minipage}
\begin{minipage}{0.3\textwidth}
\centering
\includegraphics[width=\textwidth]{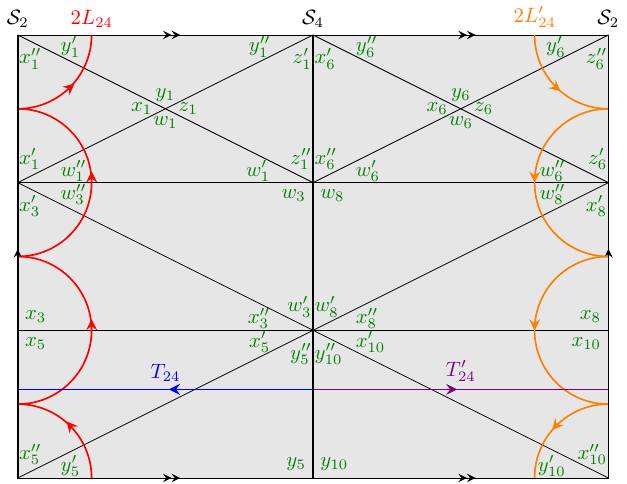}
\subcaption{$(ab)=(24)$}
\label{fig:T24_snake}
\end{minipage}
\begin{minipage}{0.3\textwidth}
\centering
\includegraphics[width=\textwidth]{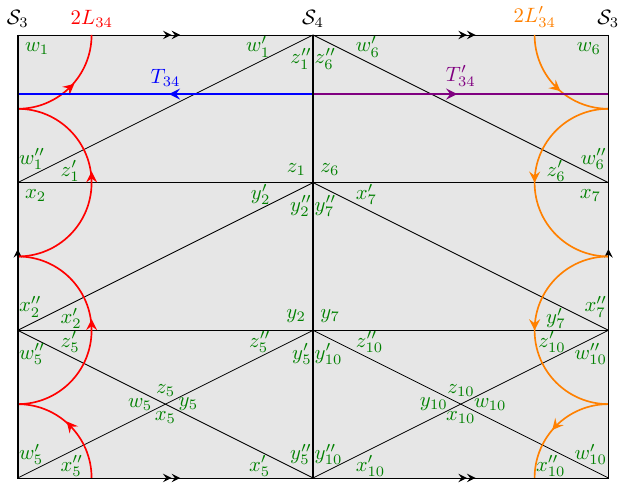}
\subcaption{$(ab)=(34)$}
\label{fig:T34_snake}
\end{minipage}
\begin{minipage}{0.3\textwidth}
\centering
\includegraphics[width=\textwidth]{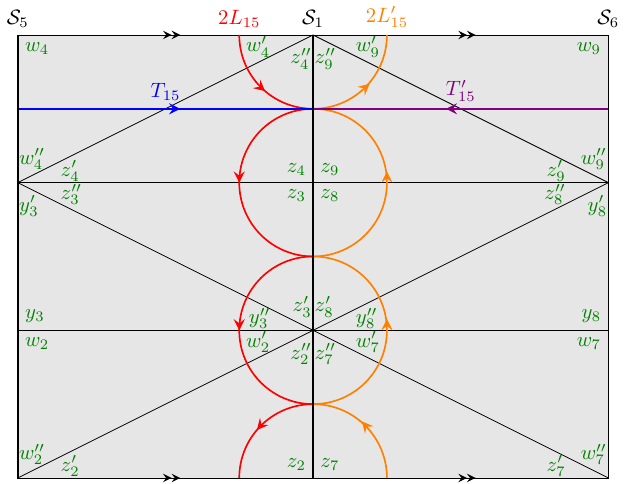}
\subcaption{$(ab)=(15)$}
\label{fig:T15_snake}
\end{minipage}
\qquad
\begin{minipage}{0.3\textwidth}
\centering
\includegraphics[width=\textwidth]{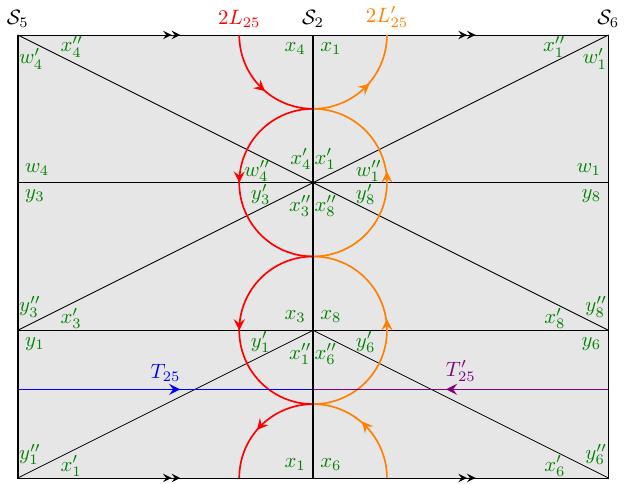}
\subcaption{$(ab)=(25)$}
\label{fig:T25_snake}
\end{minipage}\\
\begin{minipage}{0.3\textwidth}
\centering
\includegraphics[width=\textwidth]{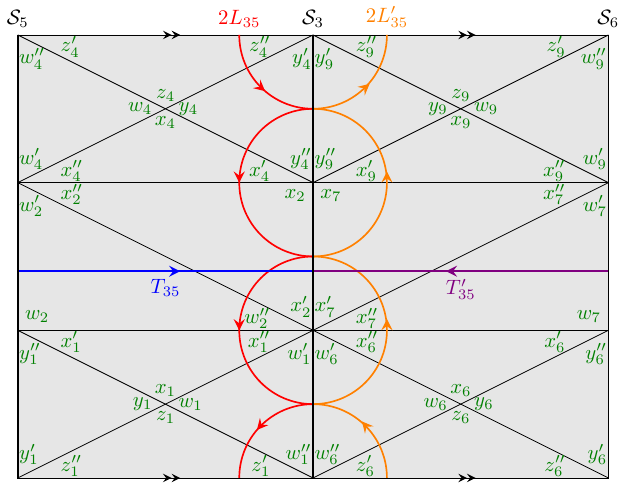}
\subcaption{$(ab)=(35)$}
\label{fig:T35_snake}
\end{minipage}
\qquad
\begin{minipage}{0.3\textwidth}
\centering
\includegraphics[width=\textwidth]{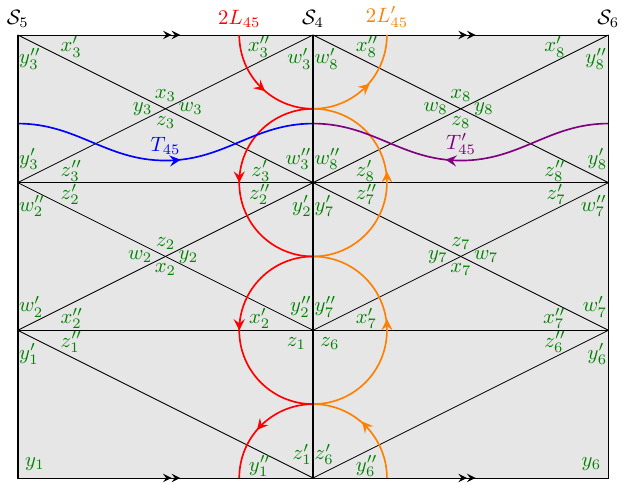}
\subcaption{$(ab)=(45)$}
\label{fig:T45_snake}
\end{minipage}
\caption{Paths on cusp boundary associated with which the FN coordinates on $M_{+\cup -}$ are defined. FN lengths $2L_{ab}$'s (\resp $2L'_{ab}$'s) on $M_+$ (\resp $M_-$) are associated to paths shown {\it in red} (\resp {\it in orange}) while FN twists $T_{ab}$'s (\resp $T'_{ab}$'s) on $M_+$ (\resp $M_-$) are associated to paths shown {\it in blue} (\resp {\it in violet}).
$2L_{ab}$ and $2L'_{ab}$ calculated by the cusp snake rules \eqref{eq:cusp_snake_rule} is the same as calculated by  \eqref{eq:all_FN_coordinates}. The paths for $T_{ab}$'s and $T'_{ab}$'s are chosen in a simple and symmetric way. The choices of paths for all $T_{ab}$'s are consistent with those in \cite{Han:2015gma}. Cusp boundaries in (a)--(f) are torus cusps, on each of which $2L_{ab}$ and $2L'_{ab}$ are associated to the A-cycle of the torus and $T_f\equiv T_{ab}-T'_{ab}$ is associated to the B-cycle. 
} 
\label{fig:TAll_snake_rule}
\end{figure}
Comparing $T_{ab}/2$ and $\cT_{ab}-\cT'_{ab}$ when imposing the gluing constraints \eqref{eq:oct_constraint} for octahedra and express in terms of the octahedron symplectic coordinates, one finds that they are different merely by a constant as follows.
\be
\{T_{ab}/2-\lb\cT_{ab}-\cT'_{ab}\rb\}_{(ab)} =i\pi \vec{t}_T\,,\quad
\vec{t}_T=
\{0, 1, 1, 2, 0, -2, 0, -1, 0, -1\}\,,
\ee
where the order of $\{(ab)\}$ is $(ab)=\{(12),(13),(14),(15),(23),(24),(25),(34),(35),(45)\}$. This leads to $\zeta_{ab}(2L_{ab})$ given in \eqref{eq:zeta}.

\bibliographystyle{bib-style} 
\bibliography{RC.bib}

\providecommand{\href}[2]{#2}\begingroup\raggedright\begin{thebibliography}{10}

\bibitem{Han:2021tzw}
M.~Han, ``{Four-dimensional spinfoam quantum gravity with a cosmological
  constant: Finiteness and semiclassical limit},'' Phys. Rev. D {\bf 104}
  (2021), no.~10, 104035,
  \href{http://arXiv.org/abs/2109.00034}{{\texttt{arXiv:2109.00034}}}.

\bibitem{Rovelli:2014ssa}
C.~Rovelli and F.~Vidotto, {\em {Covariant Loop Quantum Gravity}: {An
  Elementary Introduction to Quantum Gravity and Spinfoam Theory}}.
\newblock Cambridge Monographs on Mathematical Physics. Cambridge University
  Press, 11, 2014.

\bibitem{Perez:2012wv}
A.~Perez, ``{The Spin Foam Approach to Quantum Gravity},'' Living Rev. Rel.
  {\bf 16} (2013) 3,
  \href{http://arXiv.org/abs/1205.2019}{{\texttt{arXiv:1205.2019}}}.

\bibitem{Ponzano:1968se}
G.~Ponzano and T.~E. Regge, ``Semiclassical limit of Racah coefficients,''.

\bibitem{Freidel:2002dw}
L.~Freidel and D.~Louapre, ``{Diffeomorphisms and spin foam models},'' Nucl.
  Phys. B {\bf 662} (2003) 279--298,
  \href{http://arXiv.org/abs/gr-qc/0212001}{{\texttt{arXiv:gr-qc/0212001}}}.

\bibitem{Christodoulou:2012af}
M.~Christodoulou, M.~Langvik, A.~Riello, C.~Roken, and C.~Rovelli,
  ``{Divergences and Orientation in Spinfoams},'' Class. Quant. Grav. {\bf 30}
  (2013) 055009,
  \href{http://arXiv.org/abs/1207.5156}{{\texttt{arXiv:1207.5156}}}.

\bibitem{Turaev:1992hq}
V.~Turaev and O.~Viro, ``{State sum invariants of 3 manifolds and quantum 6j
  symbols},'' Topology {\bf 31} (1992) 865--902.

\bibitem{Oriti:2006se}
D.~Oriti, ``{The Group field theory approach to quantum gravity},''
  \href{http://arXiv.org/abs/gr-qc/0607032}{{\texttt{arXiv:gr-qc/0607032}}}.

\bibitem{Krajewski:2011zzu}
T.~Krajewski, ``{Group field theories},'' PoS {\bf QGQGS2011} (2011) 005,
  \href{http://arXiv.org/abs/1210.6257}{{\texttt{arXiv:1210.6257}}}.

\bibitem{Perini:2008pd}
C.~Perini, C.~Rovelli, and S.~Speziale, ``{Self-energy and vertex radiative
  corrections in LQG},'' Phys. Lett. B {\bf 682} (2009) 78--84,
  \href{http://arXiv.org/abs/0810.1714}{{\texttt{arXiv:0810.1714}}}.

\bibitem{Krajewski:2010yq}
T.~Krajewski, J.~Magnen, V.~Rivasseau, A.~Tanasa, and P.~Vitale, ``{Quantum
  Corrections in the Group Field Theory Formulation of the EPRL/FK Models},''
  Phys. Rev. D {\bf 82} (2010) 124069,
  \href{http://arXiv.org/abs/1007.3150}{{\texttt{arXiv:1007.3150}}}.

\bibitem{Engle:2007wy}
J.~Engle, E.~Livine, R.~Pereira, and C.~Rovelli, ``{LQG vertex with finite
  Immirzi parameter},'' Nucl. Phys. B {\bf 799} (2008) 136--149,
  \href{http://arXiv.org/abs/0711.0146}{{\texttt{arXiv:0711.0146}}}.

\bibitem{Freidel:2007py}
L.~Freidel and K.~Krasnov, ``{A New Spin Foam Model for 4d Gravity},'' Class.
  Quant. Grav. {\bf 25} (2008) 125018,
  \href{http://arXiv.org/abs/0708.1595}{{\texttt{arXiv:0708.1595}}}.

\bibitem{Frisoni:2021uwx}
P.~Frisoni, F.~Gozzini, and F.~Vidotto, ``{Numerical analysis of the
  self-energy in covariant loop quantum gravity},'' Phys. Rev. D {\bf 105}
  (2022), no.~10, 106018,
  \href{http://arXiv.org/abs/2112.14781}{{\texttt{arXiv:2112.14781}}}.

\bibitem{Riello:2013bzw}
A.~Riello, ``{Self-energy of the Lorentzian Engle-Pereira-Rovelli-Livine and
  Freidel-Krasnov model of quantum gravity},'' Phys. Rev. D {\bf 88} (2013),
  no.~2, 024011,
  \href{http://arXiv.org/abs/1302.1781}{{\texttt{arXiv:1302.1781}}}.

\bibitem{Dona:2018pxq}
P.~Don\`a, ``{Infrared divergences in the EPRL-FK Spin Foam model},'' Class.
  Quant. Grav. {\bf 35} (2018), no.~17, 175019,
  \href{http://arXiv.org/abs/1803.00835}{{\texttt{arXiv:1803.00835}}}.

\bibitem{Noui:2002ag}
K.~Noui and P.~Roche, ``{Cosmological deformation of Lorentzian spin foam
  models},'' Class. Quant. Grav. {\bf 20} (2003) 3175--3214,
  \href{http://arXiv.org/abs/gr-qc/0211109}{{\texttt{arXiv:gr-qc/0211109}}}.

\bibitem{Han:2010pz}
M.~Han, ``{4-dimensional Spin-foam Model with Quantum Lorentz Group},'' J.
  Math. Phys. {\bf 52} (2011) 072501,
  \href{http://arXiv.org/abs/1012.4216}{{\texttt{arXiv:1012.4216}}}.

\bibitem{Fairbairn:2010cp}
W.~J. Fairbairn and C.~Meusburger, ``{Quantum deformation of two
  four-dimensional spin foam models},'' J. Math. Phys. {\bf 53} (2012) 022501,
  \href{http://arXiv.org/abs/1012.4784}{{\texttt{arXiv:1012.4784}}}.

\bibitem{Regge:1961px}
T.~Regge, ``{GENERAL RELATIVITY WITHOUT COORDINATES},'' Nuovo Cim. {\bf 19}
  (1961) 558--571.

\bibitem{Hartle:1981cf}
J.~B. Hartle and R.~Sorkin, ``{Boundary Terms in the Action for the Regge
  Calculus},'' Gen. Rel. Grav. {\bf 13} (1981) 541--549.

\bibitem{Friedberg:1984ma}
R.~Friedberg and T.~D. Lee, ``{Derivation of Regge's Action From Einstein's
  Theory of General Relativity},'' Nucl. Phys. B {\bf 242} (1984) 145.

\bibitem{Regge:2000wu}
T.~Regge and R.~M. Williams, ``{Discrete structures in gravity},'' J. Math.
  Phys. {\bf 41} (2000) 3964--3984,
  \href{http://arXiv.org/abs/gr-qc/0012035}{{\texttt{arXiv:gr-qc/0012035}}}.

\bibitem{Barrett:1997gw}
J.~W. Barrett and L.~Crane, ``{Relativistic spin networks and quantum
  gravity},'' J. Math. Phys. {\bf 39} (1998) 3296--3302,
  \href{http://arXiv.org/abs/gr-qc/9709028}{{\texttt{arXiv:gr-qc/9709028}}}.

\bibitem{Barrett:1998gs}
J.~W. Barrett and R.~M. Williams, ``{The Asymptotics of an amplitude for the
  four simplex},'' Adv. Theor. Math. Phys. {\bf 3} (1999) 209--215,
  \href{http://arXiv.org/abs/gr-qc/9809032}{{\texttt{arXiv:gr-qc/9809032}}}.

\bibitem{Bianchi:2008ae}
E.~Bianchi and A.~Satz, ``{Semiclassical regime of Regge calculus and spin
  foams},'' Nucl. Phys. B {\bf 808} (2009) 546--568,
  \href{http://arXiv.org/abs/0808.1107}{{\texttt{arXiv:0808.1107}}}.

\bibitem{Conrady:2008mk}
F.~Conrady and L.~Freidel, ``{On the semiclassical limit of 4d spin foam
  models},'' Phys. Rev. D {\bf 78} (2008) 104023,
  \href{http://arXiv.org/abs/0809.2280}{{\texttt{arXiv:0809.2280}}}.

\bibitem{Haggard:2014xoa}
H.~M. Haggard, M.~Han, W.~Kami\'nski, and A.~Riello, ``{SL(2,C)
  Chern\textendash{}Simons theory, a non-planar graph operator, and 4D quantum
  gravity with a cosmological constant: Semiclassical geometry},'' Nucl. Phys.
  B {\bf 900} (2015) 1--79,
  \href{http://arXiv.org/abs/1412.7546}{{\texttt{arXiv:1412.7546}}}.

\bibitem{Plebanski:1977zz}
J.~F. Plebanski, ``{On the separation of Einsteinian substructures},'' J. Math.
  Phys. {\bf 18} (1977) 2511--2520.

\bibitem{Gaiotto:2009hg}
D.~Gaiotto, G.~W. Moore, and A.~Neitzke, ``{Wall-crossing, Hitchin Systems, and
  the WKB Approximation},''
  \href{http://arXiv.org/abs/0907.3987}{{\texttt{arXiv:0907.3987}}}.

\bibitem{Dimofte:2011gm}
T.~Dimofte, ``{Quantum Riemann Surfaces in Chern-Simons Theory},'' Adv. Theor.
  Math. Phys. {\bf 17} (2013), no.~3, 479--599,
  \href{http://arXiv.org/abs/1102.4847}{{\texttt{arXiv:1102.4847}}}.

\bibitem{Dimofte:2011ju}
T.~Dimofte, D.~Gaiotto, and S.~Gukov, ``{Gauge Theories Labelled by
  Three-Manifolds},'' Commun. Math. Phys. {\bf 325} (2014) 367--419,
  \href{http://arXiv.org/abs/1108.4389}{{\texttt{arXiv:1108.4389}}}.

\bibitem{Dimofte:2013lba}
T.~Dimofte, D.~Gaiotto, and R.~van~der Veen, ``{RG Domain Walls and Hybrid
  Triangulations},'' Adv. Theor. Math. Phys. {\bf 19} (2015) 137--276,
  \href{http://arXiv.org/abs/1304.6721}{{\texttt{arXiv:1304.6721}}}.

\bibitem{Dimofte:2014zga}
T.~Dimofte, ``{Complex Chern\textendash{}Simons Theory at Level k via the
  3d\textendash{}3d Correspondence},'' Commun. Math. Phys. {\bf 339} (2015),
  no.~2, 619--662,
  \href{http://arXiv.org/abs/1409.0857}{{\texttt{arXiv:1409.0857}}}.

\bibitem{andersen2014complex}
J.~E. Andersen and R.~Kashaev, ``Complex Quantum Chern-Simons,'' arXiv preprint
  arXiv:1409.1208 (2014).

\bibitem{Dimofte:2010wxa}
T.~D. Dimofte, {\em {Refined BPS Invariants, Chern-Simons Theory, and the
  Quantum Dilogarithm}}.
\newblock PhD thesis, Caltech, 2010.

\bibitem{Han:2015gma}
M.~Han, ``{4d Quantum Geometry from 3d Supersymmetric Gauge Theory and
  Holomorphic Block},'' JHEP {\bf 01} (2016) 065,
  \href{http://arXiv.org/abs/1509.00466}{{\texttt{arXiv:1509.00466}}}.

\bibitem{Fock:2003alg}
V.~V. Fock and A.~B. Goncharov, ``Moduli spaces of local systems and higher
  Teichmuller theory,'' 2003.

\bibitem{EllegaardAndersen:2011vps}
J.~Ellegaard~Andersen and R.~Kashaev, ``{A TQFT from Quantum Teichm\"uller
  Theory},'' Commun. Math. Phys. {\bf 330} (2014) 887--934,
  \href{http://arXiv.org/abs/1109.6295}{{\texttt{arXiv:1109.6295}}}.

\bibitem{andersen2013new}
J.~E. Andersen and R.~Kashaev, ``A new formulation of the Teichm$\backslash$"
  uller TQFT,'' arXiv preprint arXiv:1305.4291 (2013).

\bibitem{Engle:2007qf}
J.~Engle, R.~Pereira, and C.~Rovelli, ``{Flipped spinfoam vertex and loop
  gravity},'' Nucl. Phys. B {\bf 798} (2008) 251--290,
  \href{http://arXiv.org/abs/0708.1236}{{\texttt{arXiv:0708.1236}}}.

\bibitem{Ding:2010fw}
Y.~Ding, M.~Han, and C.~Rovelli, ``{Generalized Spinfoams},'' Phys. Rev. D {\bf
  83} (2011) 124020,
  \href{http://arXiv.org/abs/1011.2149}{{\texttt{arXiv:1011.2149}}}.

\bibitem{Haggard:2015ima}
H.~M. Haggard, M.~Han, and A.~Riello, ``{Encoding Curved Tetrahedra in Face
  Holonomies: Phase Space of Shapes from Group-Valued Moment Maps},'' Annales
  Henri Poincare {\bf 17} (2016), no.~8, 2001--2048,
  \href{http://arXiv.org/abs/1506.03053}{{\texttt{arXiv:1506.03053}}}.

\bibitem{Coman:2015lna}
I.~Coman, M.~Gabella, and J.~Teschner, ``{Line operators in theories of class
  $\mathcal{S}$, quantized moduli space of flat connections, and Toda field
  theory},'' JHEP {\bf 10} (2015) 143,
  \href{http://arXiv.org/abs/1505.05898}{{\texttt{arXiv:1505.05898}}}.

\bibitem{Teschner:2013tqy}
J.~Teschner and G.~S. Vartanov, ``{Supersymmetric gauge theories, quantization
  of $\mathcal{M}_{\mathrm{flat}}$, and conformal field theory},'' Adv. Theor.
  Math. Phys. {\bf 19} (2015) 1--135,
  \href{http://arXiv.org/abs/1302.3778}{{\texttt{arXiv:1302.3778}}}.

\bibitem{Nekrasov:2011bc}
N.~Nekrasov, A.~Rosly, and S.~Shatashvili, ``{Darboux coordinates, Yang-Yang
  functional, and gauge theory},'' Nucl. Phys. B Proc. Suppl. {\bf 216} (2011)
  69--93, \href{http://arXiv.org/abs/1103.3919}{{\texttt{arXiv:1103.3919}}}.

\bibitem{Goldman1986}
W.~Goldman, ``Invariant functions on Lie groups and Hamiltonian flows of
  surface group representations.,'' Inventiones mathematicae {\bf 85} (1986)
  263--302.

\bibitem{Gazeau:2009zz}
J.-P. Gazeau, {\em {Coherent states in quantum physics}}.
\newblock 2009.

\bibitem{Haggard:2015yda}
H.~M. Haggard, M.~Han, W.~Kami\'nski, and A.~Riello, ``{Four-dimensional
  Quantum Gravity with a Cosmological Constant from Three-dimensional
  Holomorphic Blocks},'' Phys. Lett. B {\bf 752} (2016) 258--262,
  \href{http://arXiv.org/abs/1509.00458}{{\texttt{arXiv:1509.00458}}}.

\bibitem{Haggard:2015nat}
H.~M. Haggard, M.~Han, W.~Kaminski, and A.~Riello,
  ``{$\operatorname{SL}(2,\mathbb{C})$ Chern-Simons theory, flat connections,
  and four-dimensional quantum geometry},'' Adv. Theor. Math. Phys. {\bf 23}
  (2019), no.~4, 1067--1158,
  \href{http://arXiv.org/abs/1512.07690}{{\texttt{arXiv:1512.07690}}}.

\bibitem{BenGeloun:2010qkf}
J.~Ben~Geloun, R.~Gurau, and V.~Rivasseau, ``{EPRL/FK Group Field Theory},''
  EPL {\bf 92} (2010), no.~6, 60008,
  \href{http://arXiv.org/abs/1008.0354}{{\texttt{arXiv:1008.0354}}}.

\bibitem{Dimofte:2013iv}
T.~Dimofte, M.~Gabella, and A.~B. Goncharov, ``{K-Decompositions and 3d Gauge
  Theories},'' JHEP {\bf 11} (2016) 151,
  \href{http://arXiv.org/abs/1301.0192}{{\texttt{arXiv:1301.0192}}}.

\bibitem{Bianchi:2010fj}
E.~Bianchi, D.~Regoli, and C.~Rovelli, ``{Face amplitude of spinfoam quantum
  gravity},'' Class. Quant. Grav. {\bf 27} (2010) 185009,
  \href{http://arXiv.org/abs/1005.0764}{{\texttt{arXiv:1005.0764}}}.

\bibitem{Alekseev:1994au}
A.~Y. Alekseev, H.~Grosse, and V.~Schomerus, ``{Combinatorial quantization of
  the Hamiltonian Chern-Simons theory. 2.},'' Commun. Math. Phys. {\bf 174}
  (1995) 561--604,
  \href{http://arXiv.org/abs/hep-th/9408097}{{\texttt{arXiv:hep-th/9408097}}}.

\bibitem{Alekseev:1994pa}
A.~Y. Alekseev, H.~Grosse, and V.~Schomerus, ``{Combinatorial quantization of
  the Hamiltonian Chern-Simons theory},'' Commun. Math. Phys. {\bf 172} (1995)
  317--358,
  \href{http://arXiv.org/abs/hep-th/9403066}{{\texttt{arXiv:hep-th/9403066}}}.

\bibitem{Alekseev:1995rn}
A.~Y. Alekseev and V.~Schomerus, ``{Representation theory of Chern-Simons
  observables},''
  \href{http://arXiv.org/abs/q-alg/9503016}{{\texttt{arXiv:q-alg/9503016}}}.

\bibitem{BJBCrowley_1979}
B.~J.~B. Crowley, ``Some generalisations of the Poisson summation formula,''
  Journal of Physics A: Mathematical and General {\bf 12} (nov, 1979) 1951.

\bibitem{Han:2013hna}
M.~Han, ``{On Spinfoam Models in Large Spin Regime},'' Class. Quant. Grav. {\bf
  31} (2014) 015004,
  \href{http://arXiv.org/abs/1304.5627}{{\texttt{arXiv:1304.5627}}}.

\bibitem{Bonzom:2009hw}
V.~Bonzom, ``{Spin foam models for quantum gravity from lattice path
  integrals},'' Phys. Rev. D {\bf 80} (2009) 064028,
  \href{http://arXiv.org/abs/0905.1501}{{\texttt{arXiv:0905.1501}}}.

\bibitem{Hellmann:2013gva}
F.~Hellmann and W.~Kaminski, ``{Holonomy spin foam models: Asymptotic geometry
  of the partition function},'' JHEP {\bf 10} (2013) 165,
  \href{http://arXiv.org/abs/1307.1679}{{\texttt{arXiv:1307.1679}}}.

\bibitem{hormander2015analysis}
L.~H{\"o}rmander, {\em The analysis of linear partial differential operators I:
  Distribution theory and Fourier analysis}.
\newblock Springer, 2015.

\bibitem{Bahr:2015gxa}
B.~Bahr and S.~Steinhaus, ``{Investigation of the Spinfoam Path integral with
  Quantum Cuboid Intertwiners},'' Phys. Rev. D {\bf 93} (2016), no.~10, 104029,
  \href{http://arXiv.org/abs/1508.07961}{{\texttt{arXiv:1508.07961}}}.

\bibitem{Imamura:2013qxa}
Y.~Imamura, H.~Matsuno, and D.~Yokoyama, ``{Factorization of the
  $S^3/\mathbb{Z}_n$ partition function},'' Phys. Rev. D {\bf 89} (2014),
  no.~8, 085003,
  \href{http://arXiv.org/abs/1311.2371}{{\texttt{arXiv:1311.2371}}}.

\bibitem{Faddeev:1995nb}
L.~D. Faddeev, ``{Discrete Heisenberg-Weyl group and modular group},'' Lett.
  Math. Phys. {\bf 34} (1995) 249--254,
  \href{http://arXiv.org/abs/hep-th/9504111}{{\texttt{arXiv:hep-th/9504111}}}.

\bibitem{kashaev1997hyperbolic}
R.~M. Kashaev, ``The hyperbolic volume of knots from the quantum dilogarithm,''
  Letters in mathematical physics {\bf 39} (1997), no.~3, 269--275.

\end{thebibliography}\endgroup

\end{document}